\documentclass[a4paper,fleqn]{cas-dc}

\usepackage[numbers]{natbib}
\usepackage{amsthm, amsmath,bbm,mathtools,physics,amssymb,subfig}
\usepackage[skins,breakable,hooks,theorems]{tcolorbox}

\newtheorem{theorem}{Theorem}

\newtheorem{lemma}{Lemma}
\newtheorem{definition}{Definition}
\begin{document}
\let\WriteBookmarks\relax
\def\floatpagepagefraction{1}
\def\textpagefraction{.001}
\shorttitle{A Non-intrusive Decentralized Approach to Stabilizing IBR-dominated AC Microgrids}
\shortauthors{T. Huang}

\title [mode = title]{A Non-intrusive Decentralized Approach to Stabilizing IBR-dominated AC Microgrids}                      
\tnotemark[1]

\tnotetext[1]{This work was partly supported by the U.S. National Science Foundation (NSF) Grant ECCS-2328205.}

%\tnotetext[2]{The second title footnote which is a longer text matter
%   to fill through the whole text width and overflow into
%   another line in the footnotes area of the first page.}

\author[1]{Tong Huang}[type=editor,
                        auid=000,bioid=1,
                        orcid = 0000-0002-2630-4825]
\cormark[1]
\ead{thuang7@sdsu.edu}

%\credit{Conceptualization of this study, Methodology, Software}

%\address[1]{, Street 129, 1043 NX Amsterdam, The Netherlands}
\affiliation[1]{organization={Department of Electrical and Computer Engineering, San Diego State University },
                %addressline={Jawahar Nagar}, 
                city={San Diego},
%               citysep={}, % Uncomment if no comma needed between city and postcode
                state={California},
                postcode={92182}, 
                country={USA}}

\cortext[cor1]{Corresponding author}
%\cortext[cor2]{Principal corresponding author}
%\fntext[fn1]{This is the first author footnote, but is common to third
%  author as well.}
%\fntext[fn2]{Another author footnote, this is a very long footnote and
%  it should be a really long footnote. But this footnote is not yet
%  sufficiently long enough to make two lines of footnote text.}

\begin{abstract}
This paper presents a non-intrusive, decentralized approach that stabilizes AC microgrids dominated by inverter-based resources (IBRs). By ``non-intrusive'' we mean that the approach does not require reprogramming IBRs' controllers to stabilize the microgrids. ``Decentralized'' is in the sense that the approach stabilizes the microgrids without communication among IBRs. Implementing the approach only requires very minimal information of IBR dynamics, \textcolor{black}{i.e., the $\mathcal{L}_2$ gain of an IBR,} and sharing such information with the non-IBR-manufacturer parties does not cause any concerns on intellectual property privacy.
The approach allows for plug-and-play operation of IBRs, while maintaining microgrid stability.
%The interface is practically implementable for non-IBR-manufacturer parties, as its design only needs a scalar that encapsulates input-output dynamics of an IBR, and does not require the detailed control schemes of IBRs which may not be disclosed to the non-manufacturer parties due to business privacy concerns of IBR manufacturers. 
%This paper also provide the IBR manufacturers with an algorithm for computing the behavior-summary scalar.
%Compared with conventional methods, the proposed method has three advantages: 1) implementing the method does not require reprogramming IBR controllers or accessing internal state variables of IBRs; 2) designing the interface only requires very minimal information of IBR dynamics, and sharing such information with the non-manufacturer parties does not cause any concerns on intellectual property privacy; and 3) it explicitly addresses the complexity resulting from IBR's high-order dynamics. 
The proposed approach is tested by simulating 2-IBR and \textcolor{black}{10-IBR} microgrids where lines and IBRs are modeled in the electromagnetic transient time scale. Simulations show that oscillations with increasing amplitudes may occur, when two stable microgrids are networked. Simulations also suggest that the proposed approach can mitigate such a system-level symptom.
\end{abstract}

%\includegraphics{figs/cas-grabs.pdf}

\begin{highlights}
\item The proposed approach stabilizes AC microgrids dominated by inverter-based resources (IBRs) in a non-intrusive, and decentralized fashion, through a power-electronics (PE) interface.

\item Designing the PE interface only needs a scalar that encapsulates input-output dynamics of an IBR, and does not require the detailed control schemes of the IBR \textcolor{black}{or network topology information}. 

\item The proposed approach can address the high-order dynamics due to the tight interaction among voltage and current controllers, and network dynamics in the EMT time scale.
\end{highlights}


\maketitle

\section{Introduction}
%\subsection{Background}
As many countries are decarbonizing their energy infrastructure, a growing number of Inverter-based Resources (IBRs\footnote{The explanation of acronyms in this paper is listed in Appendix \ref{app:aronyms}}), including energy storage, rooftop solar panels, and electric vehicle charging stations, are emerging in distribution grids \cite{XIE2022100640}. 
However, integrating large-scale IBRs will pose unprecedented challenges to distribution grid management, since today's distribution grids are not designed for hosting massive IBRs, and distribution system operators (DSOs) generally cannot directly control IBRs at grid edges. With the concept of microgrids \cite{985003}, a large amount of IBRs in a distribution grid can be managed via a ``divide-and-conquer'' strategy: the distribution grid can be divided into several networked microgrids, and each microgrid manages its own generation and loads \cite{8536454}. With such an architecture, the management complexity for DSOs is significantly reduced, as the DSOs only need to coordinate several microgrids, instead of controlling massive IBRs in a centralized manner \cite{THneuralTSG}. A microgrid has three operational modes: a grid-connected mode \cite{985003}, an islanded mode \cite{985003}, and a hybrid mode \cite{huangilicTSG}. Under normal conditions, a microgrid can enter the grid-connected mode where the loads in the microgrid can be balanced by the energy from both local generation and the host distribution system. When the host distribution grid fails to deliver energy, a microgrid can either balance its load autonomously by its local generation (i.e., the islanded mode), or network with its neighboring microgrids and balance loads collaboratively (i.e., the hybrid mode) \cite{huangilicTSG}.

One key challenge of operating microgrids in the islanded or hybrid mode is how to ensure the microgrid stability \cite{MG_def}. Compared with large-scale transmission systems whose dynamics are governed by thousands of giant rotating machines, the microgrids powered by IBRs are more sensitive to disturbances that include connection or disconnection of IBRs, renewable fluctuations, and line faults, due to lack of physical inertia in generation resources and the small scale of the microgrids. As a result, the disturbances may compromise the quality of electricity services by incurring sustained oscillations or even instability. Exacerbating the challenge, today's IBR manufacturers tune their IBRs at a device level without much consideration of system-level performance of networked IBRs. However, the non-manufacturer parties (NMPs), e.g., DSOs, microgrid operators (\textcolor{black}{MGOs}), and IBR owners, who \textcolor{black}{are concerned about} security of networked IBRs, typically do not know the detailed control schemes of IBRs and cannot reprogram the IBRs' controllers. This is because the manufacturers are reluctant to share their detailed control schemes with the NMPs due to concerns on intellectual property (IP) privacy. Without the consideration of the system-level performance, IBRs might fight with other, causing undesirable oscillations or instability. Such incidences occurred in transmission systems, e.g., the sub-synchronous control interactions (SSCI) in Texas \cite{ERCOT_SSO} and oscillations in High Voltage DC systems that contain multiple converters \cite{yin2019review}. In the context of microgrids, it is possible that networking two stable microgrids leads to oscillations with increasing amplitudes (shall be shown in Section \ref{sec:case_study}). Therefore, as more and more IBRs are emerging at grid edges, it is imperative to develop technologies that certify system-level stability of networked IBRs.
%\subsection{State of the art and its limitations}

%\subsection{State of the Art}
Existing approaches to stability certification for electrical energy systems can be classified into three categories: centralized, impedance-based, and passivity-based approaches. In the \emph{centralized} approaches, system operators (SOs) are assumed to be able to collect dynamical models of key components in the systems, and they \emph{assess} the system stability by performing time-domain simulations \cite{1338120}, by conducting small-signal analysis \cite{6880421}, or by searching for system behavior-summary functions, e.g., the Lyapunov functions \cite{THneuralTSG,Kabalan2017}, and energy functions \cite{481632,FOUAD1988233}. The drawbacks of these centralized approaches are listed as follows: 1) IBR manufacturers can only share a ``black-box'' model with SOs for simulation purposes, due to concerns on IP privacy. Consequently, detailed IBRs' models are not available for performing analytical stability assessment \cite{THneuralTSG,481632,FOUAD1988233,Kabalan2017}. 2) Some approaches \cite{huang2021neural,1338120} are computationally intractable when addressing high-order systems. For an IBR-rich microgrid, wide-range behaviors of interested lie in the electromagnetic transient (EMT) time scale, and they are described by high-order dynamics. 
3) Most approaches \cite{481632,FOUAD1988233,Kabalan2017} cannot provide SOs with actionable guidance of enforcing system stability. Beyond stability analysis, controls enforcing stability are much needed.

The impedance-based and passivity-based approaches address the drawbacks of the centralized approaches by developing device-level stability conditions that each IBR needs to satisfy locally to ensure the stability of its host system. One way to design such conditions is by checking if the impedance ratio satisfies the Nyquist stability criterion, where the impedance ratio is defined by the IBR output impedance and the equivalent impedance of the host grid. For example, reference \cite{FLEE_02} proposes impedance specifications for stable DC resources and a data-driven way to measure the specifications. Reference \cite{Sun_09} reviews impedance specifications for stability assessment of AC generation resources. Reference \cite{Sun_11} points out that different impedance-based criteria should be used for assessing stability of voltage-source systems and current-source systems. Reference \cite{XWang_14} generalizes the impedance-based stability criteria from a single-converter-infinity-bus system to a network with multiple converters. Based on the impedance-based analysis, reference \cite{XWang_22} proposes a participation function that aims to pinpoint the root causes of instability. Reference \cite{bikash_assessment} performs the impedance-stability assessment with black-box converter models. In addition to stability assessment, there is a large body of literature that enforces the impedance-based stability conditions by tuning IBR control parameters \cite{Control_para_17,apparent_impedance_para_tuning}, and adding active dampers \cite{active_damper}. The passivity theory is another common tool for designing the device-level stability conditions. For example, reference \cite{self_displ} introduces the concept of self-disciplined stabilization in the context of DC microgrids. The stability condition for each IBR is the passivity of the single-input-single-output (SISO) transfer function of the IBR. 
Reference \cite{tsinghua_condition} proposes the distributed, passivity-like stability condition based on low-order nodal dynamics and power flow equations. 
Reference \cite{UCF_paper} develops the stability condition for conventional generators in transmission systems based on the passivity shortage framework. Reference \cite{amit_dist} learns a neural network-structured storage function for each IBR and leverages the storage function as stability conditions to certify microgrid stability. Reference \cite{IIT_paper} presents the passivity-based stability condition for IBRs to assess small-signal stability of both fast and slow behaviors of IBR interconnections. 

Unfortunately, the existing impedance/passivity-based approaches have the following limitations:
1) In references \cite{Control_para_17,IIT_paper, self_displ,apparent_impedance_para_tuning,UCF_paper,tsinghua_condition}, the stability conditions are enforced in an \emph{intrusive} manner, i.e., one has to reprogram the controllers of generation resources to enforce the conditions. This is undesirable for both NMPs and IBR manufacturers. The IBR controllers are typically packaged into the inverters and cannot be reprogrammed by the NMPs, for protecting IP privacy and reducing IBRs' vulnerability to cyberattacks. The control schemes of commercial inverters are typically deliberately designed and extensively tested by IBR manufacturers for achieving certain functions, such as voltage and current regulation. Hence, the IBR manufacturers might be reluctant to completely abandon or radically change their mature control schemes for enforcing the stability condition \cite{IIT_paper}. Besides, since many IBRs have been installed in the grid, it is costly or even infeasible to reprogram the controllers of these existing IBRs.
2) The complexity of dynamics of IBR-dominated, AC microgrids is ignored by \cite{self_displ, amit_dist, UCF_paper, tsinghua_condition}.
    For example, reference \cite{self_displ} only considers the SISO dynamics of converter interfaces in DC microgrids, while the IBR's dynamics in an AC microgrid can have multiple inputs and outputs. References \cite{ amit_dist,UCF_paper,tsinghua_condition} only address the slow dynamics of generation units but ignores the interactions among network dynamics and fast IBR controllers in the EMT time scale. Modelling full-order network dynamics is necessary in an IBR-rich microgrid, as some inverters may have high-frequency dynamics \cite{TimGreenModel}. 
3) References \cite{bikash_assessment, IIT_paper,FLEE_02,amit_dist,XWang_22,Sun_09,Sun_11,XWang_14} only address stability assessment in a distributed manner without providing guidance of how to stabilize an unstable microgrid. 4) Some impedance-based approaches \cite{FLEE_02,Sun_09,Sun_11} simplify the dynamics of the host systems of an IBR as an ideal voltage source in series with impedance. Such a simplification is valid when the IBR connects to a strong grid (e.g., a large-scale transmission/distribution system). However, when an IBR connects to a microgrid, the complexity of dynamics of its host microgrid cannot be ignored. 5) While developed based on the ``black-box'' IBR models, some impedance-based approaches \cite{active_damper} require topology information of the host grid including line parameters and network connectivity. However, since the topology information can change dynamically due to potentially open boundaries among microgrids, stability assessment results and stability enforcement performance may change accordingly, making it challenging to achieve the plug-and-play operation of IBRs.

%\subsection{Contribution Summary}
This paper introduces a \emph{first-of-its-kind}, \emph{non-intrusive}, and \emph{decentralized} approach to stabilizing the IBR-dominated AC microgrids. The paper's contribution is summarized as follows:

\noindent 1) The approach stabilizes the IBR-dominated AC microgrids in a non-intrusive, and decentralized fashion. The ``non-intrusive, and decentralized'' is in the sense that the design and operation of the PE interface do not require reprogramming IBR controllers, the topology information, or communications among IBRs. This allows the NMPs to stabilize microgrids in the EMT time scale and achieve plug-and-play operation of IBRs. \textcolor{black}{The non-intrusive feature cannot be achieved by the methods in \cite{Control_para_17, IIT_paper,self_displ,apparent_impedance_para_tuning, amit_dist,tsinghua_condition,UCF_paper}}.

\noindent 2) Designing the PE interface only needs a scalar that encapsulates input-output dynamics of an IBR, and does not require the detailed control schemes of the IBR \textcolor{black}{or network topology information}. Exposing such a scalar to NMPs will not cause any IP concerns for IBR manufacturers, as the detailed IBR control schemes cannot be inferred only based on the scalar. \textcolor{black}{Compared with our approach, some existing methods require either the detailed IBR models \cite{Control_para_17,self_displ,apparent_impedance_para_tuning, UCF_paper,tsinghua_condition} or the topology information \cite{active_damper} to enforce stability.}

\noindent 3) The proposed approach can address the high-order dynamics due to the tight interaction among voltage and current controllers, and network dynamics in the EMT time scale\textcolor{black}{, whereas such complexity of dynamics of the IBR network is ignored by some existing methods \cite{self_displ, amit_dist, UCF_paper, tsinghua_condition}.}

%The rest of this paper is organized as follows: Section \ref{sec:formulation} mathematically describes the dynamics of an IBR-dominated microgrid; Section \ref{sec:protocol} presents the decentralized stability condition; Section \ref{sec:enforcement} introduces a power-electronics interface that stabilizes the microgrids; Section \ref{sec:case_study} tests the performance of the interface; and Section \ref{sec:concl} summarizes this paper.

%The stability EMT is designed based on the passivity theory and can be enforced in a decentralized manner at grid edges. Enforcing the protocol allows for plug-and-play operation of IBRs, while maintaining system-level stability of the microgrids hosting the IBRs.
%The protocol is enforced by a novel interface in a non-intrusive manner. The interface is practically implementable for a third party (i.e., a non-IBR-manufacturer party), as its design only needs a scalar that encapsulates input-output dynamics of an IBR, and does not require the detailed control schemes of IBRs which may not be disclosed to the third party due to business privacy concerns of IBR manufacturers.

%\input{Section/2Dynamics}
\section{Microgrid Dynamics} \label{sec:formulation}
This section describes the nodal and network dynamics of the microgrid with $N$ IBRs. The microgrid dynamics is organized into a feedback architecture lending itself to developing a device-level stability condition.
\subsection{Dynamics of IBRs} \label{subsection:IBR-dynamics}
This paper considers two types of IBRs: grid-forming (GFM) and grid-following (GFL) IBRs. Figures \ref{fig:IBRn} and \ref{fig:IBRn-GFL} present the representative architectures of these two types of IBRs. The dynamics of the representative GFM and GFL IBRs are elaborated in Appendices \ref{app: GFM_IBR} and \ref{app: GFL_IBR}. It can be observed from Figures \ref{fig:IBRn} and \ref{fig:IBRn-GFL} that both GFM and GFL IBRs interact with the rest of the microgrid via terminal voltages $\mathbf{v}_{\text{o}n}$ and terminal currents $\mathbf{i}_{\text{o}n}$, while they are governed by different internal state vector $\textbf{x}_n$\footnote{$\mathbf{x}_n$ will be $[\phi_{\text{d}n}, \phi_{\text{q}n}, \gamma_{\text{d}n}, \gamma_{\text{q}n}, i_{\text{ld}n}, i_{\text{lq}n}, v_{\text{od}n}, v_{\text{oq}n}]^{\top}$, if the $n$-th IBR is GFM and its dynamics is presented in Appendix \ref{app: GFM_IBR}; $\mathbf{x}_n$ will be $[\gamma_{\text{d}n}, \gamma_{\text{q}n}, i_{\text{ld}n}, i_{\text{lq}n}, v_{\text{od}n}, v_{\text{oq}n}]^{\top}$, if the $n$-th IBR is GFL and its dynamics is presented in \ref{app: GFL_IBR}. Each state in $\textbf{x}_n$ is explained in Appendices \ref{app: GFM_IBR} and \ref{app: GFL_IBR}.}.
This paper focuses on stabilizing the fast, system-level dynamics of microgrids in the EMT time scale where the setpoints dispatched by secondary controllers are assumed to be constant.
The small-signal dynamics of an IBR in such a time scale can be described by
\begin{subequations} \label{eq:compact_deviation}
  \begin{align}
  &\Delta \dot{\mathbf{x}}_n = A_n \Delta \mathbf{x}_n + B_n \Delta \mathbf{i}_{\text{odq}n}\\
  & \Delta \mathbf{v}_{\text{odq}n} = C_n \Delta \mathbf{x}_n
  \end{align}
\end{subequations}
where the ``$\Delta$'' variables are the deviations of the corresponding variables from their steady states; $\textbf{i}_{\text{odq}n}$ ($:= [i_{\text{od}n}, i_{\text{oq}n}]^{\top}$) is the terminal current $\textbf{i}_{\text{o}n}$ represented in the direct-quadrature (d-q) reference frame of the $n$-th IBR; $\mathbf{v}_{\text{odq}n}$ ($:= [v_{\text{od}n}, v_{\text{oq}n}]^{\top}$) is the terminal voltage $\textbf{v}_{\text{o}n}$ represented in the d-q frame; and matrices $A_n$, $B_n$, and $C_n$ are derived from the IBR dynamics presented in Appendices \ref{app: GFM_IBR} and \ref{app: GFL_IBR}. The input-output relationship of the dynamics of IBR $n$ is shown in the central block of Figure \ref{fig:IO-IBR}-(a). The input $\Delta \mathbf{i}_{\text{odq}n}$ and output $\Delta \mathbf{v}_{\text{odq}n}$ interact with the rest of the microgrid in a common reference frame (i.e., D-Q frame). Next, we present the reference frame transformation \cite{TimGreenModel,sauer2017power} that converts the d-q variables to the D-Q frame.

\begin{figure}
    \centering
    \includegraphics[width = 0.85\linewidth]{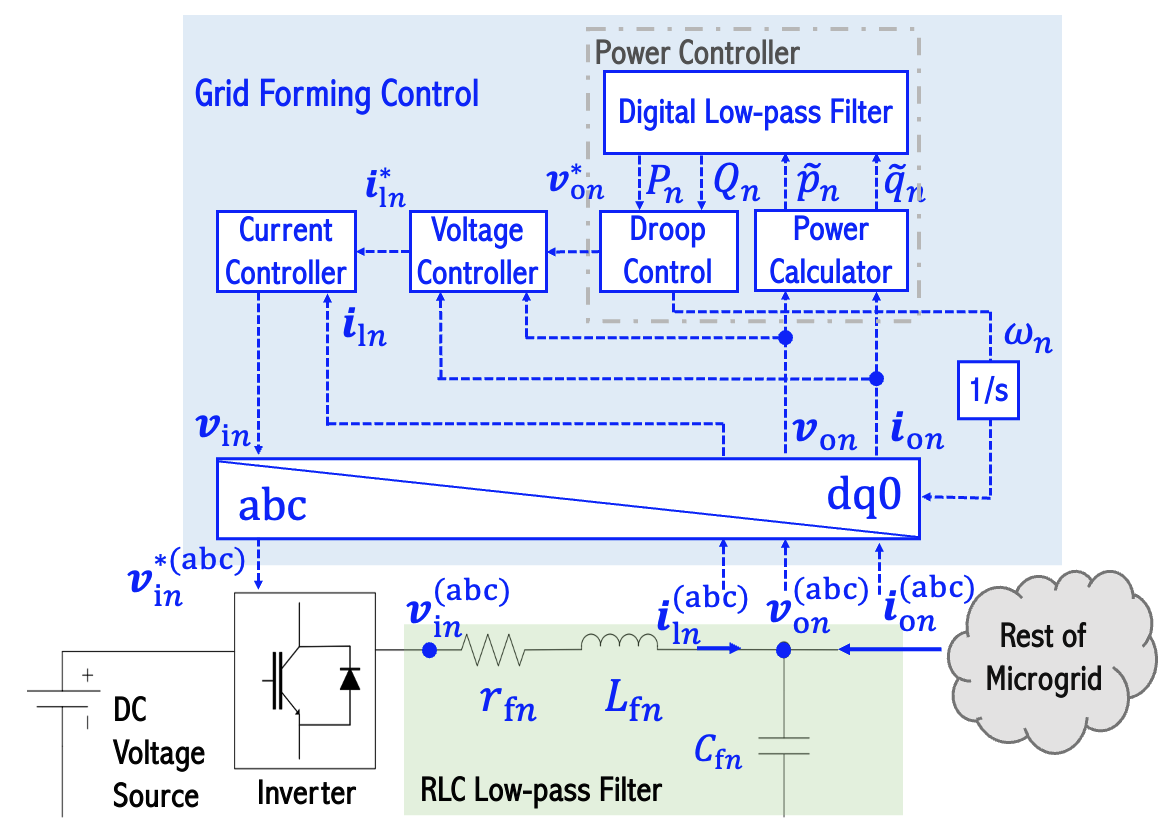}
    \caption{Cyber and physical architecture of a grid-forming IBR.}
    \label{fig:IBRn}
\end{figure}

\begin{figure}
    \centering
    \includegraphics[width = 0.85\linewidth]{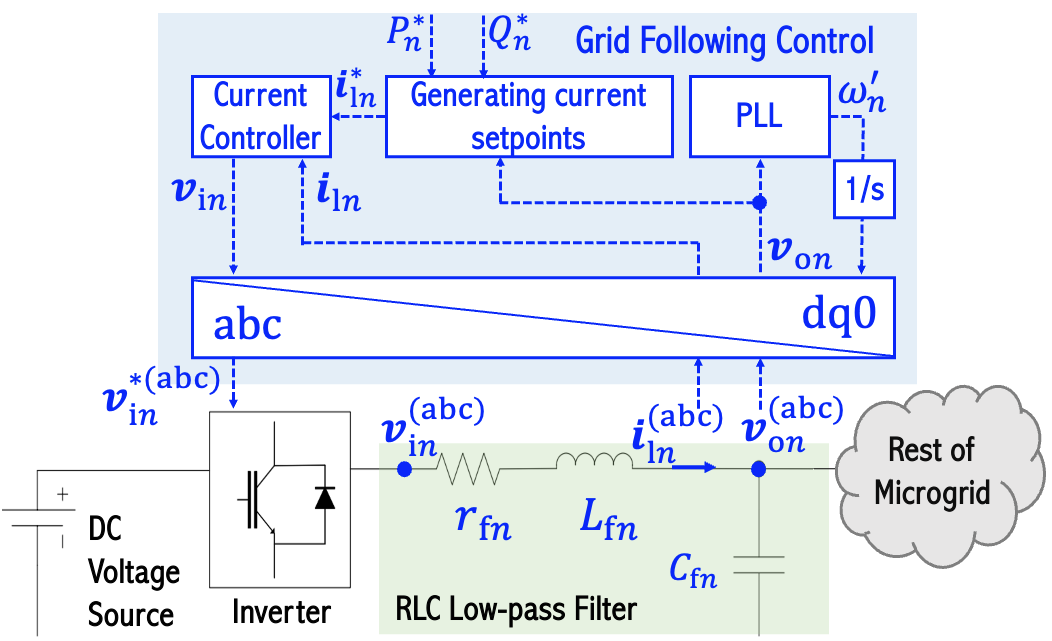}
    \caption{Cyber and physical architecture of a grid-following IBR.}
    \label{fig:IBRn-GFL}
\end{figure}

\begin{figure}
    \centering
    \includegraphics[width = \linewidth]{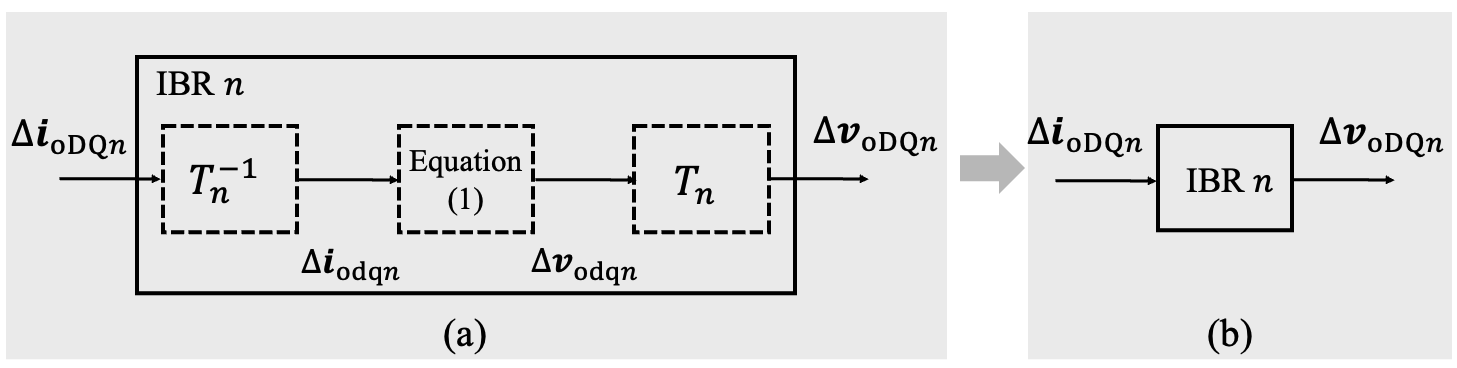}
    \caption{(a) Reference frame transformation; and (b) an input-output perspective of IBR dynamics}
    \label{fig:IO-IBR}
\end{figure}

%Figure \ref{fig:IBRn} presents the cyber-physical architecture of the $n$-th IBR, $n=1, 2, \ldots, N$ \cite{TimGreenModel}. The IBR includes a DC voltage source, an inverter, a resistor-inductor-capacitor (RLC) low-pass filter, and an inverter controller. %The dynamics of the low-pass filter and the inverter controller are elaborated below.

%\subsubsection{Reference frame transformation}
In Figure \ref{fig:IO-IBR}-(a), the output $\Delta \mathbf{v}_{\text{oDQ}n}:= [\Delta v_{\text{oD}n},
        \Delta v_{\text{oQ}n}]^{\top}$ is obtained by
$\Delta \mathbf{v}_{\text{oDQ}n} =
    T_n
    \Delta \mathbf{v}_{\text{odq}n}$
where
\begin{equation} \label{eq:dq2DQ}
    T_n =
    \begin{bmatrix}
        \cos{\delta_n} & -\sin{\delta_n}\\
        \sin{\delta_n} & \cos{\delta_n}
    \end{bmatrix}.
\end{equation}
Note that $\delta_n$ is assumed to be a constant, since it changes much slower than the states $\mathbf{x}_n$ in the time scale of interest.
Similarly, the relationship between $\Delta \mathbf{i}_{\text{oDQ}n}:=[\Delta i_{\text{oD}n}, \Delta i_{\text{oQ}n}]^{\top}$ and $\Delta \mathbf{i}_{\text{odq}n}$ are described by
$\Delta \mathbf{i}_{\text{odq}n} = T_n^{-1} \Delta \mathbf{i}_{\text{oDQ}n}$.     
With the above definitions, IBR $n$ can be viewed as a dynamic system that is driven by $\Delta \mathbf{i}_{\text{oDQ}n}$ while outputting $\Delta \mathbf{v}_{\text{oDQ}n}$ (Figure \ref{fig:IO-IBR}-(b)).  
%\subsubsection{Input-output Dynamics of IBRs}

\subsection{Dynamics of Microgrid Network}
Assume that the microgrid with $N$ IBRs is three-phase balanced and hosts constant-impedance load. By the Kron reduction technique, the microgrid network can be reduced to a network with $N+1$ node and $M$ branches. One of the $N+1$ node is the neutral/reference point of the microgrid. Let set $\mathcal{N}:=\{0,1,2,\ldots, N\}$ collect the nodal indices of the Kron-reduced network where ``$0$'' denotes the nodal index for the neutral point. Let set $\mathcal{M}:=\{1,2,\ldots,M\}$ collect branch indices of the reduced network. Another way to represent branch $m$ is to use a pair $(i,j)_m$ where $i,j\in\mathcal{N}$ correspond to the two nodes of the two terminals of branch $m$. Suppose that $i<j$, we define the positive direction assigned to branch $m$ is from node $i$ to $j$. 

%\textbf{Change set L2 to another name!}
The $M$ branches in the Kron-reduced network can be divided into two categories. Let $\mathcal{E}_1$ collect the branches connecting to the neutral point via an IBR, while set $\mathcal{E}_2$ collects the rest of the branches. The dynamics of branches in $\mathcal{E}_1$ are governed by equations presented in Section \ref{subsection:IBR-dynamics}, whereas the dynamic behaviors of the branches in $\mathcal{E}_2$ are modeled by RL circuits with resistor $r_{\text{b}m}$ and inductance $L_{\text{b}m}$:
\begin{subequations} \label{eq:network_dyn}
	\begin{align}
		&L_{\text{b}m}\dot{i}_{\text{bD}m} = -r_{\text{b}m}i_{\text{bD}m} + \omega_0 L_{\text{b}m} i_{\text{bQ}m} + v_{{\text{bD}}m} \\
		&L_{\text{b}m}\dot{i}_{\text{bQ}m} = -r_{\text{b}m}i_{\text{bQ}j} - \omega_0 L_{\text{b}m} i_{\text{bD}m} + v_{{\text{bQ}}m} 
		%&\mathbf{y}_{\text{DQ}} = [C, C]\mathbf{x}_{\text{l}}
	\end{align}
\end{subequations}
where $m \in \mathcal{E}_2$; the subscript ``b'' reminds readers that the corresponding variables are used for describe branches without IBRs; the subscripts ``D'' and ``Q'' suggest the corresponding variables are in the common reference frame (the D-Q frame);
$v_{{\text{bD}}m}$ and $v_{{\text{bQ}}m}$ are the bus voltage differences of branch $(i,j)_m$ in the D- and Q- axis, i.e., $v_{{\text{bD}}m} = v_{\text{D}i}-v_{\text{D}j}$ and $v_{{\text{bQ}}m} = v_{\text{Q}i}-v_{\text{Q}j}$.

\begin{figure}
    \centering
    \includegraphics[width = 0.35\linewidth]{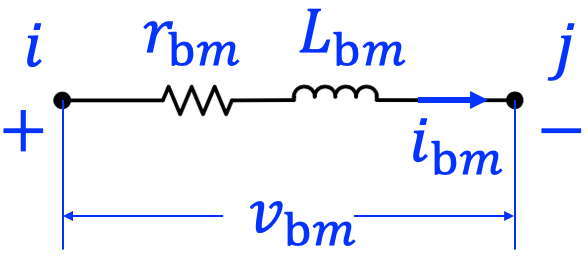}
    \caption{Branch $(i,j)_m$ in $\mathcal{E}_2$}
    \label{fig:branch-m}
\end{figure}

To characterize the relationship between branch currents $i_{\text{b}m}$ for $m\in \mathcal{M}$, we introduce a \emph{reduced incidence matrix} $C'\in \mathbb{R}^{N\times M}$ whose entries are $c'_{n,m}$ with $n\in \mathcal{N} \backslash \{0\}$ and $m\in \mathcal{M}$. Each entry $c'_{n,m}$ in matrix $C'$ is defined as follows: $c'_{n,m}=1$ if branch $m$ is incident at node $n$, and the reference direction of branch $m$ is away from node $n$; $c'_{n,m}=-1$ if branch $m$ is incident at node $n$, and the reference direction of branch $m$ is toward to node $n$; and $c'_{n,m}=0$ if branch $m$ is not incident at node $n$.

%The Kron-reduced network of the microgrid shown in Figure XX can be represented by a direct graph $\mathcal{G} = (\mathcal{N}, \mathcal{B})$. Set $\mathcal{N}$ collects all nodes in the microgrid, i.e., $\mathcal{N} = \{1, 2, \ldots, N\}$. Set $\mathcal{B} = \mathcal{B}_1 \cup \mathcal{B}_2$ where $\mathcal{B}_1 = \{(p,q)\in \mathcal{N}\times \mathcal{N}| p<q\}$ and $\mathcal{B}_2 = \{(0,q)| q \in \mathcal{N}\}$. For branch $(p,q)\in \mathcal{B}$, we define that the reference direction points to node $q$. Let the first $b$ entries of $\mathcal{B}$ denote the branches without connecting the neutral line, while let the last $m$ entries of $\mathcal{B}$ denote the branches connecting the neutral line.

With the reference direction defined before, one can assign indices of nodes and branches such that the reduced incidence matrix $C'$ has the following structure \cite{ilic2000dynamics}
\begin{equation} \label{eq: Cp}
    C' =
    \begin{bmatrix}
         C_0 & -I_N
    \end{bmatrix}
\end{equation}
where $C_0$ is the first $M-N$ columns of matrix $C'$; and $I_N$ is a $N$-dimension identity matrix.

Next, we present the compact form of Kirchhoff’s Current Law (KCL), with the incident matrix $C'$. Let $M'$ be $M-N$. The KCL of the microgrid network in terms of direct/quadrature current leads to
\begin{equation} \label{eq:KCL}
	C'\mathbf{i}_{\text{D}} = \mathbf{0}; \quad C'\mathbf{i}_{\text{Q}} = \mathbf{0}
\end{equation}
where $\mathbf{i}_{\text{D}} = [\mathbf{i}_{\text{bD}}^{\top},\mathbf{i}_{\text{sD}}^{\top}]^{\top}$ with $\mathbf{i}_{\text{bD}} = [i_{\text{bD}1}, \ldots, i_{\text{bD}M'}]^{\top}$, $\mathbf{i}_{\text{sD}} = [i_{\text{sD}1},\ldots, i_{\text{sD}N}]^{\top}$; and $\mathbf{i}_{\text{Q}} = [\mathbf{i}_{\text{bQ}}^{\top},\mathbf{i}_{\text{sQ}}^{\top}]^{\top}$ with $\mathbf{i}_{\text{bQ}} = [i_{\text{bQ}1}, \ldots, i_{\text{bQ}N}]^{\top}$, $\mathbf{i}_{\text{sQ}} = [i_{\text{sQ}1},\ldots, i_{\text{sQ}N}]^{\top}$. 
Plugging \eqref{eq: Cp} into \eqref{eq:KCL} leads to
\begin{equation} \label{eq:current}
	\mathbf{i}_{\text{sD}}=C_0 \mathbf{i}_{\text{bD}};  \quad \mathbf{i}_{\text{sQ}}=C_0 \mathbf{i}_{\text{bQ}}.
\end{equation}

Moreover, the relationship between the voltages across branches and the nodal voltages can be described by
\begin{equation} \label{eq:V_branch}
	\mathbf{v}_{\text{D}} = C'^{\top}\mathbf{v}_{\text{oD}}; \quad \mathbf{v}_{\text{Q}} = C'^{\top}\mathbf{v}_{\text{oQ}}
\end{equation}
In \eqref{eq:V_branch}, $\mathbf{v}_{\text{D}} = [\mathbf{v}_{\text{bD}}^{\top}, \mathbf{v}_{\text{oD}}^{\top}]^{\top}$ and $\mathbf{v}_{\text{Q}} = [\mathbf{v}_{\text{bQ}}^{\top}, \mathbf{v}_{\text{oQ}}^{\top}]^{\top}$, where the voltages across branches $\mathbf{v}_{\text{bD}} = [v_{{\text{bD}}1}, \ldots, v_{{\text{bD}}M'}]^{\top}$; $\mathbf{v}_{\text{bQ}} = [v_{{\text{bQ}}1}, \ldots, v_{{\text{bQ}}M'}]^{\top}$; and nodal voltages $\mathbf{v}_{\text{oD}} = [v_{\text{oD}1}, \ldots, v_{\text{oD}M'}]^{\top}$, and $\mathbf{v}_{\text{oQ}} = [v_{\text{oQ}1}, \ldots, v_{\text{oQ}M'}]^{\top}$ where $v_{\text{oD}m}$ and $v_{\text{oQ}m}$ are obtained by casting $v_{\text{od}m}$ and $v_{\text{oq}m}$ to the D-Q frame by \eqref{eq:dq2DQ}.

Plugging \eqref{eq: Cp} into \eqref{eq:V_branch} leads to \cite{ilic2000dynamics}
\begin{equation}\label{eq:voltage}
	\mathbf{v}_{\text{bD}} = C_0^{\top}\mathbf{v}_{\text{oD}}; \quad \mathbf{v}_{\text{bQ}} = C_0^{\top}\mathbf{v}_{\text{oQ}}
\end{equation}
Define the following vectors:
\begin{equation}
	\mathbf{i}_{\text{sDQ}} =
	\begin{bmatrix}
		\mathbf{i}_{\text{sD}}\\
		\mathbf{i}_{\text{sQ}}
	\end{bmatrix}; 
	\mathbf{i}_{\text{bDQ}} =
	\begin{bmatrix}
		\mathbf{i}_{\text{bD}}\\
		\mathbf{i}_{\text{bQ}}
	\end{bmatrix};
	\mathbf{v}_{\text{oDQ}} = 
	\begin{bmatrix}
		\mathbf{v}_{\text{oD}}\\
		\mathbf{v}_{\text{oQ}}
	\end{bmatrix}.
\end{equation}
The branch dynamics \eqref{eq:network_dyn} can be organized into
\begin{subequations} \label{eq:systme_H}
	\begin{align}
		&L\dot{\mathbf{i}}_{\text{bDQ}} = -R\mathbf{i}_{\text{bDQ}} + W \mathbf{i}_{\text{bDQ}} + C^{\top}\mathbf{v}_{\text{oDQ}}\\
		&\mathbf{i}_{\text{sDQ}} = C\mathbf{i}_{\text{bDQ}}
		\end{align}
\end{subequations}
where 
\begin{subequations} 
    \begin{align}
        & L = \text{diag}(L_1, \ldots, L_{M'}, L_1, \ldots, L_{M'})\\
        & R = \text{diag}(r_1, \ldots, r_{M'}, r_1, \ldots, r_{M'}) \\
        & W = \begin{bmatrix}
		0_{M'\times M'} & \omega_0 I_{M'}\\
		-\omega_0 I_{M'} & 0_{{M'}\times {M'}}
	\end{bmatrix}.
    \end{align}
\end{subequations}
Since \eqref{eq:systme_H} is linear, the following equations also hold:
\begin{subequations} \label{eq:systme_H}
	\begin{align}
		&L\Delta\dot{\mathbf{i}}_{\text{bDQ}} = -R\Delta\mathbf{i}_{\text{bDQ}} + W \Delta\mathbf{i}_{\text{bDQ}} + C^{\top}\Delta\mathbf{v}_{\text{oDQ}}\\
		&\Delta\mathbf{i}_{\text{sDQ}} = C\Delta\mathbf{i}_{\text{bDQ}}
		\end{align}
\end{subequations}
where the ``$\Delta$'' variables are the deviations of the original variables from their steady states.
%\textbf{[Correct $v_n$ to $v_o$; add the deviation version]}.

\subsection{A Feedback Perspective of Microgrid Dynamics}

The interaction between the IBRs and the microgrid network can be interpreted from a feedback perspective shown in Figure \ref{fig:feedback}. The IBR dynamics $\mathcal{F}_n$ ($n$= 1, 2, \ldots, N) constitute the feed-forward loop $\mathcal{F}$, whereas the feedback loop $\mathcal{B}$ results from the network dynamics \eqref{eq:systme_H}. The input of $\mathcal{F}$ is $\Delta \mathbf{i}_{\text{oDQ}}$ defined by
$\Delta \mathbf{i}_{\text{oDQ}} = -\Delta \mathbf{i}_{\text{sDQ}}$
where the negative \textcolor{black}{sign} results from the reference directions of $i_{\text{o}n}$ and $i_{\text{s}n}$ defined before: recall that the positive reference direction of $i_{\text{o}n}$ points into the IBR $n$, while the positive reference direction of $i_{\text{s}n}$ points into the network. The output of $\mathcal{F}$ is $\Delta \mathbf{v}_{\text{oDQ}}$ which drives the network dynamics \eqref{eq:systme_H}.

With Figure \ref{fig:feedback}, the dynamics of the microgrid with $N$ IBRs can be interpreted as follows. At time step $k$, current $\Delta \mathbf{i}_{\text{oDQ}n}[k]$ for $n = 1, 2, \ldots, N$ drives the dynamics of system $\mathcal{F}$ which updates the internal state variables $\Delta \mathbf{x}_n[k+1]$ and outputs
voltage $\Delta v_{\text{oDQ}n}[k]$. The voltages $\Delta \mathbf{v}_{\text{oDQ}n}[k]$ further drive the dynamics of the microgrid network to update the internal state variables of the network and produce $\Delta \mathbf{i}_{\text{sDQ}}[k+1]$. The updated currents $\Delta \mathbf{i}_{\text{sDQ}}[k+1]$ drives the dynamics of the IBRs, and the process described above repeats. Such a feedback perspective lends itself to introducing an IBR-level stability condition based on the passivity theory.

\begin{figure} 
    \centering
    \includegraphics[width = 0.6\linewidth]{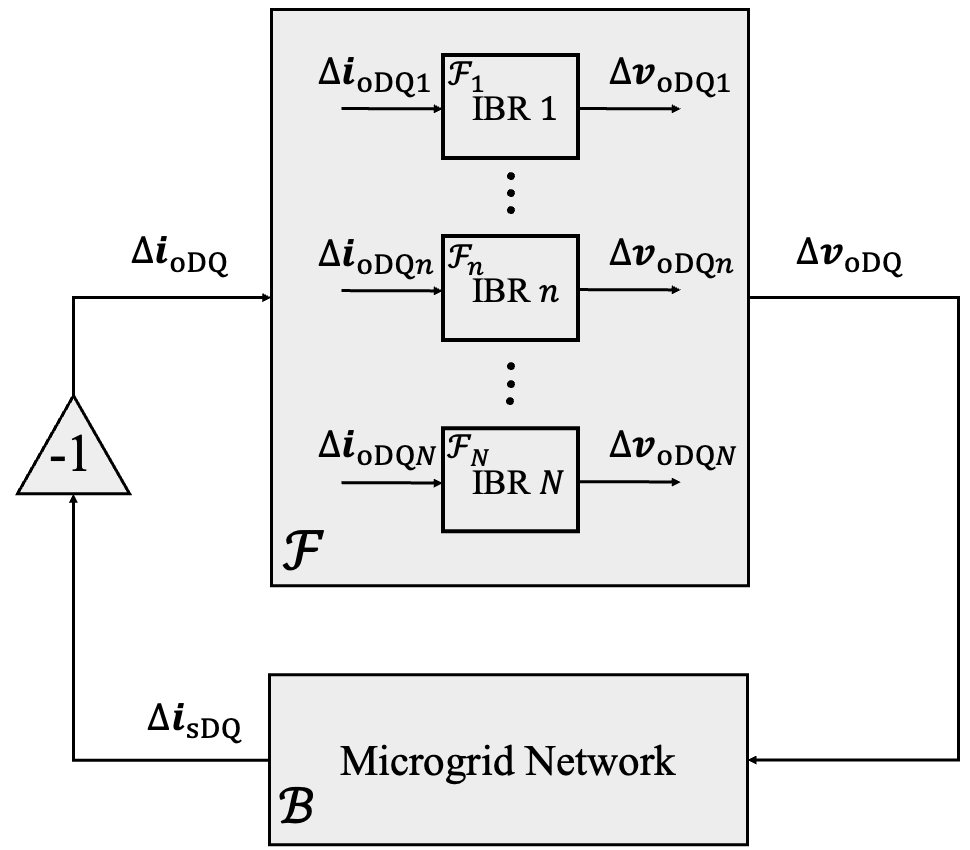}
    \caption{A feedback perspective of microgrid dynamics}
    \label{fig:feedback}
\end{figure}
\section{Decentralized Stability Condition} \label{sec:protocol}
This section answers the question of what condition each IBR should satisfy such that they establish a stable microgrid. 
%This section first introduces some definitions in the control theory. Then we present a lemma that provides guidance to design the stability condition. Finally, the condition is formally described and justified.
\subsection{Stability of Interconnected Systems}
We proceed to present several definitions necessary for introducing the stability condition. For a system $\mathcal{H}$ with input $\mathbf{u} \in \mathbb{R}^d$ and output $\mathbf{y} \in \mathbb{R}^d$, the next two definitions examine the input-output properties of $\mathcal{H}$:
\begin{definition} \label{def:OFP}
    (OFP \cite{hill1977stability}) The system $\mathcal{H}:\mathbf{u} \rightarrow \mathbf{y}$ is output feedback passive (OFP), if for all square integrable $\mathbf{u}(t)$ and some $\sigma>0$,
    \begin{equation} \label{eq:OFP_ineq}
        \int_0^t\mathbf{u}(\tau)^{\top} \mathbf{y}(\tau) d\tau -\sigma \int_0^t\mathbf{y}(\tau)^{\top} \mathbf{y}(\tau) d\tau  \ge 0,
    \end{equation}
    with a zero initial condition. Moreover, $\sigma$ is called the passivity index.
\end{definition}
\begin{definition}
    ($\mathcal{L}_2$ Gain \cite{hill1977stability}) The system $\mathcal{H}: \mathbf{u}\rightarrow \mathbf{y}$ has finite $\mathcal{L}_2$ gain $\gamma>0$ if for all square integrable $\mathbf{u}$
    \begin{equation}
         \int_0^t\mathbf{y}(\tau)^{\top} \mathbf{y}(\tau) d\tau \le \gamma^2 \int_0^t\mathbf{u}(\tau)^{\top} \mathbf{u}(\tau) d\tau,
    \end{equation}
    with a zero initial condition.
\end{definition}

The link between asymptotic stability and the output feedback passivity is established by the following lemma \cite{hill1977stability}:
\begin{lemma} \label{lemma:DHill}
    %(Asymptotic stability of interconnected passive systems)
    (Corollary 1 in \cite{hill1977stability}) The equilibrium point $\mathbf{o}$ of the closed-loop system in Figure \ref{fig:feedback} is asymptotically stable, if both subsystems $\mathcal{F}$ and $\mathcal{B}$ are output feedback passive.
\end{lemma}

Lemma \ref{lemma:DHill} guides one to design a decentralized condition for each IBR to ensure system-level stability. Subsection \ref{subs: OFP_network} examines the OFP property of the feedback loop $\mathcal{B}$ in Figure \ref{fig:feedback}. Subsection \ref{subs: Protocol} introduces the condition that ensures the OFP property of the feed-forward loop $\mathcal{F}$.

\subsection{Output Feedback Passivity of Microgrid Networks} \label{subs: OFP_network}
To establish the asymptotic stability, Lemma \ref{lemma:DHill} requires the RL network $\mathcal{B}$ to be OFP. While it is well known that a RL network is passive, how to quantify the extent that the RL network is passive has not been well studied yet in the power and energy community. The OFP property of \eqref{eq:systme_H} in the DQ frame is established by the following theorem:
\begin{theorem}
 \label{thm:network}
	(Network Passivity Index) The microgrid network dynamics \eqref{eq:systme_H} is OFP with input $\Delta \mathbf{v}_{\text{oDQ}}$ and output $\Delta \mathbf{i}_{\text{sDQ}}$, if matrix $C^{\top}C$ has at least one positive eigenvalue.
\end{theorem}

\begin{proof}
	%Consider a storage function $V=\frac{1}{2} \mathbf{i}_{\text{lDQ}}^{\top} L\mathbf{i}_{\text{lDQ}}$. Note that $\dot{V}=\mathbf{i}_{\text{lDQ}}^{\top} L\dot{\mathbf{i}}_{\text{lDQ}}$.
By definition,
	\begin{equation*}
	\begin{aligned}
		&\int_0^t \Delta \mathbf{i}_{\text{sDQ}}^{\top}\Delta \mathbf{v}_{\text{oDQ}} d \tau = \int_0^t \Delta \mathbf{i}_{\text{bDQ}}^{\top}C^{\top}\Delta \mathbf{v}_{\text{oDQ}} d \tau\\
		%&= \int_0^t \Delta\mathbf{i}_{\text{bDQ}}^{\top}\left(L\Delta\dot{\mathbf{i}}_{\text{bDQ}} + R\Delta\mathbf{i}_{\text{bDQ}} - W \Delta\mathbf{i}_{\text{bDQ}}\right) d \tau\\
		&= \int_0^t \left(\Delta\mathbf{i}_{\text{bDQ}}^{\top}L\Delta\dot{\mathbf{i}}_{\text{bDQ}} + \Delta\mathbf{i}_{\text{bDQ}}^{\top}R\Delta\mathbf{i}_{\text{bDQ}} - \Delta\mathbf{i}_{\text{bDQ}}^{\top}W\Delta\mathbf{i}_{\text{bDQ}} \right) d \tau\\
		%&=\dot{V} +\mathbf{i}_{\text{lDQ}}^{\top}R\mathbf{i}_{\text{lDQ}}
	\end{aligned}
	\end{equation*}
%Owing to $R>0$, $\mathcal{H}$ is strictly passive with storage function $V$.
Note that $W = -W^{\top}$ and $\Delta\mathbf{i}_{\text{bDQ}}^{\top}W\Delta\mathbf{i}_{\text{bDQ}}$ is a scalar. Then,
\begin{equation} \label{eq:proof1}
	\begin{aligned}
    \Delta\mathbf{i}_{\text{bDQ}}^{\top}W\Delta\mathbf{i}_{\text{bDQ}} &= \left(\Delta\mathbf{i}_{\text{bDQ}}^{\top}W\Delta\mathbf{i}_{\text{bDQ}}\right)^\top\\
    %&= \Delta\mathbf{i}_{\text{bDQ}}^{\top} W^{\top} \Delta\mathbf{i}_{\text{bDQ}}\\
    &= -\Delta\mathbf{i}_{\text{bDQ}}^{\top} W \Delta\mathbf{i}_{\text{bDQ}}.
    \end{aligned}
\end{equation}
Equation \eqref{eq:proof1} leads to $2\Delta\mathbf{i}_{\text{bDQ}}^{\top}W\Delta\mathbf{i}_{\text{bDQ}} = 0$, implying 
\begin{equation} \label{eq:proof2}
    \Delta\mathbf{i}_{\text{bDQ}}^{\top}W\Delta\mathbf{i}_{\text{bDQ}} = 0.
\end{equation}
Based on \eqref{eq:proof1} and \eqref{eq:proof2},
\begin{equation*}
	\begin{aligned}
		&\int_0^t \Delta \mathbf{i}_{\text{sDQ}}^{\top}\Delta \mathbf{v}_{\text{oDQ}} d \tau 
        = V(t) -  V_0 + \int_0^t\Delta\mathbf{i}_{\text{bDQ}}^{\top}R\Delta\mathbf{i}_{\text{bDQ}} d \tau\\
		%&=\dot{V} +\mathbf{i}_{\text{lDQ}}^{\top}R\mathbf{i}_{\text{lDQ}}
	\end{aligned}
	\end{equation*}
where 
\begin{subequations} 
    \begin{align}
        & V(t):=0.5\Delta\mathbf{i}_{\text{bDQ}}^{\top}(t)L\Delta\mathbf{i}_{\text{bDQ}}(t) \\
        & V_0:=0.5\Delta\mathbf{i}_{\text{bDQ}}^{\top}(0)L\Delta\mathbf{i}_{\text{bDQ}}(0).
    \end{align}
\end{subequations}

As matrices $L\succ 0$,
\begin{equation} \label{eq: long_ineq}
	\begin{aligned}
		&\int_0^t \Delta \mathbf{i}_{\text{sDQ}}^{\top}\Delta \mathbf{v}_{\text{oDQ}} d \tau 
        \ge - V_0 + \int_0^t\Delta\mathbf{i}_{\text{bDQ}}^{\top}R\Delta\mathbf{i}_{\text{bDQ}} d \tau\\
        &\ge - V_0 + \lambda_{\text{Rmin}}\int_0^t\Delta\mathbf{i}_{\text{bDQ}}^{\top}\Delta\mathbf{i}_{\text{bDQ}} d \tau\\
        &\ge - V_0 + \frac{\lambda_{\text{Rmin}}}{\lambda_{\text{Cmax}}}\int_0^t\Delta\mathbf{i}_{\text{sDQ}}^{\top}\Delta\mathbf{i}_{\text{sDQ}} d \tau
		%&=\dot{V} +\mathbf{i}_{\text{lDQ}}^{\top}R\mathbf{i}_{\text{lDQ}}
	\end{aligned}
	\end{equation}
where $\lambda_{\text{Rmin}}$ is the minimal eigenvalue of $R$; $\lambda_{\text{Cmax}}$ is the maximal eigenvalue of $C^{\top}C$; and $\lambda_{\text{Cmax}}>0$ as $C^{\top}C \succeq 0$.
The third line of \eqref{eq: long_ineq} is due to the fact that
\begin{equation}
	\begin{aligned}
    \Delta\mathbf{i}_{\text{sDQ}}^{\top}\Delta\mathbf{i}_{\text{sDQ}}&=\Delta\mathbf{i}_{\text{bDQ}}^{\top}C^{\top}C\Delta\mathbf{i}_{\text{bDQ}}\\
    & \le \lambda_{\text{Cmax}}\Delta\mathbf{i}_{\text{bDQ}}^{\top}\Delta\mathbf{i}_{\text{bDQ}.}
		%&=\dot{V} +\mathbf{i}_{\text{lDQ}}^{\top}R\mathbf{i}_{\text{lDQ}}
	\end{aligned}
\end{equation}
 The inequality \eqref{eq:OFP_ineq} is evaluated with a zero initial condition. By setting $\Delta\mathbf{i}_{\text{bDQ}}(0) = 0$, it follows that $V_0=0$ and dynamics \eqref{eq:systme_H} is OFP with passivity index $\lambda_{\text{Rmin}}/\lambda_{\text{Cmax}}$.
\end{proof}

\emph{Remark:} The proof of Theorem \ref{thm:network} reveals that the passivity index of an RL network depends not only on the minimal branch resistance, but also on the branches' connectivity.

\subsection{IBR-level Stability Condition} \label{subs: Protocol}
Theorem \ref{thm:network} suggests that the feedback loop $\mathcal{B}$ in Figure \ref{fig:feedback} is OFP. According to Lemma \ref{lemma:DHill}, the system-level asymptotic stability can be established, if the feed-forward loop $\mathcal{F}$ is OFP. This observation inspires us to design the following IBR-level stability condition that leads to the microgrid-level stability:

\noindent\textbf{Condition 1:} \emph{For $n = 1, 2, \ldots, N$, the dynamics of IBR $n$ with input $\Delta \mathbf{i}_{\text{odq}n}$ and output $\Delta \mathbf{v}_{\text{odq}n}$ is OFP.}

%\textbf{Say something about no-source network and energy resulting from disturbances}
The ``P(assive)'' in Condition 1 should not be confused with the ``passive element'' defined in the circuit theory \cite{alexander2013fundamentals}. In the circuit theory, the passive element is an element that is ``not capable of generating energy'' \cite{alexander2013fundamentals}. However, whether an OFP component in the sense of Definition \ref{def:OFP} is capable of generating energy or not depends on the definition of its inputs and outputs. If an IBR satisfies Condition 1, it does not mean that the IBR cannot produce energy that powers its host microgrid, and it essentially means that the IBR cannot produce energy that leads disturbances to be sustained or amplified. Section \ref{sec:case_study} shows an example that an IBR satisfies Condition 1 but produces energy. Next we show satisfying Condition 1 leads to asymptotic stability.

\begin{theorem}
    The equilibrium point of the closed-loop system in Figure \ref{fig:feedback} is asymptotically stable if Condition 1 is satisfied.
\end{theorem}

\begin{proof}
    Condition 1 requires each IBR to be OFP, i.e., there exist $\sigma_n>0$ such that, for $n = 1,2, \ldots, N$,
    \begin{equation}
        \int_0^t \Delta \mathbf{i}_{\text{odq}n}^{\top}\Delta\mathbf{v}_{\text{odq}n} d \tau - \sigma_n \int_0^t \Delta \mathbf{v}_{\text{odq}n}^{\top}\Delta\mathbf{v}_{\text{odq}n}d \tau \ge 0
    \end{equation}
    According to Figure \ref{fig:IO-IBR}-(a), $\Delta \mathbf{i}_{\text{odq}n} = T_n^{-1} \Delta \mathbf{i}_{\text{oDQ}n}$ and $\Delta \mathbf{v}_{\text{odq}n} = T_n^{-1} \Delta \mathbf{v}_{\text{oDQ}n}$, then
    \begin{equation}
        \begin{aligned}
            &\int_0^t \Delta \mathbf{i}_{\text{oDQ}n}^{\top}(T_n^{-1})^{\top}T_n^{-1}\Delta\mathbf{v}_{\text{oDQ}n} d \tau -\\
            & \sigma_n \int_0^t \Delta \mathbf{v}_{\text{oDQ}n}^{\top}(T_n^{-1})^{\top}T_n^{-1}\Delta\mathbf{v}_{\text{oDQ}n}d \tau \ge 0
        \end{aligned}
    \end{equation}
    Note that $(T_n^{-1})^{\top}T_n^{-1}=I$. This leads to
    \begin{equation}
        \int_0^t \Delta \mathbf{i}_{\text{oDQ}n}^{\top}\Delta\mathbf{v}_{\text{oDQ}n} d \tau - \sigma_n \int_0^t \Delta \mathbf{v}_{\text{oDQ}n}^{\top}\Delta\mathbf{v}_{\text{oDQ}n}d \tau \ge 0
    \end{equation}
    Define $\underline{\sigma} := \min_{n} \sigma_n$. It follows that
    \begin{equation} \label{eq:Ninequality}
        \int_0^t \Delta \mathbf{i}_{\text{oDQ}n}^{\top}\Delta\mathbf{v}_{\text{oDQ}n} d \tau - \underline{\sigma} \int_0^t \Delta \mathbf{v}_{\text{oDQ}n}^{\top}\Delta\mathbf{v}_{\text{oDQ}n}d \tau \ge 0
    \end{equation}
    for $n = 1, 2, \ldots, N$. By summing up the $N$ inequalities in \eqref{eq:Ninequality}, we have
    \begin{equation*} 
        \sum_{n = 1}^{N}\int_0^t \Delta \mathbf{i}_{\text{oDQ}n}^{\top}\Delta\mathbf{v}_{\text{oDQ}n} d \tau - \underline{\sigma} \sum_{n = 1}^{N}\int_0^t \Delta \mathbf{v}_{\text{oDQ}n}^{\top}\Delta\mathbf{v}_{\text{oDQ}n}d \tau \ge 0.
    \end{equation*}
    Since $N$ is finite, the finite summation and integration operators can be interchanged, i.e.,
    \begin{equation*} 
        \int_0^t\sum_{n = 1}^{N} \Delta \mathbf{i}_{\text{oDQ}n}^{\top}\Delta\mathbf{v}_{\text{oDQ}n} d \tau - \underline{\sigma} \int_0^t\sum_{n = 1}^{N} \Delta \mathbf{v}_{\text{oDQ}n}^{\top}\Delta\mathbf{v}_{\text{oDQ}n}d \tau \ge 0.
    \end{equation*}
    Note that 
    \begin{subequations} 
    \begin{align}
        & \sum_{n = 1}^{N} \Delta \mathbf{i}_{\text{oDQ}n}^{\top}\Delta\mathbf{v}_{\text{oDQ}n} = \Delta \mathbf{i}_{\text{oDQ}}^{\top}\Delta\mathbf{v}_{\text{oDQ}} \\
        & \sum_{n = 1}^{N} \Delta \mathbf{v}_{\text{oDQ}n}^{\top}\Delta\mathbf{v}_{\text{oDQ}n} = \Delta \mathbf{v}_{\text{oDQ}}^{\top}\Delta\mathbf{v}_{\text{oDQ}}.
    \end{align}
\end{subequations}
This leads to
    \begin{equation*} 
        \int_0^t \Delta \mathbf{i}_{\text{oDQ}}^{\top}\Delta\mathbf{v}_{\text{oDQ}} d \tau - \underline{\sigma} \int_0^t \Delta \mathbf{v}_{\text{oDQ}}^{\top}\Delta\mathbf{v}_{\text{oDQ}}d \tau \ge 0.
    \end{equation*}
    By Definition \ref{def:OFP}, the subsystem $\mathcal{F}$ in Figure \ref{fig:feedback} is OFP with passivity index $\underline{\sigma}$. In addition, since subsystem $\mathcal{B}$ is OFP according to Theorem \ref{thm:network}, the asymptotic stability of equilibrium of the system in Figure \ref{fig:feedback} is established by Lemma \ref{lemma:DHill}.
\end{proof}

\section{Non-intrusive Stability Enforcement} \label{sec:enforcement}
In this section, we first illustrate the basic idea of enforcing Condition 1. Then we conceptualize the architecture of an interface that enforces Condition 1 in a non-intrusive way. We also define the information needed to design the interface.
\subsection{Basic Idea of Stability Enforcement}
Condition 1 at IBR $n$ can be enforced by the scheme shown in Figure \ref{fig:passivation} where $\alpha_n$, $\beta_n$, and $\kappa_n$ are tunable parameters; and $I$ is an identity matrix. The next lemma guides one to tune $\alpha_n$, $\beta_n$, and $\kappa_n$ to enforce Condition 1:
\begin{figure}
    \centering
    \includegraphics[width = 1.5in]{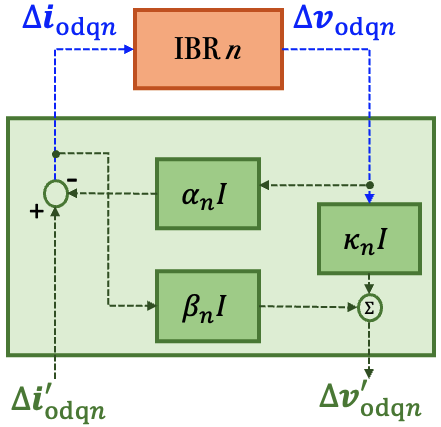}
    \caption{Basic idea of enforcing the Stability Condition}
    \label{fig:passivation}
\end{figure}

\begin{lemma} \label{lemma: Panos}
    (Theorem 4 in \cite{Panos_thm}) 
    The closed-loop system in Figure \ref{fig:passivation} with input $\Delta \mathbf{i}'_{\text{odq}n}$ and output $\Delta \mathbf{v}'_{\text{odq}n}$ is OFP with $\sigma_n = 0.5(\frac{1}{\beta_n}+\frac{\alpha_n}{\kappa_n})>0$, if
    \begin{equation} \label{eq: constraints}
        \beta_n \ge \kappa_n \gamma_n >0, \quad \kappa_n>\alpha_n \beta_n>0,
    \end{equation}
    where $\gamma_n$ is the $\mathcal{L}_2$ gain of the IBR $n$ with input $\Delta \mathbf{i}_{\text{odq}n}$ and output $\Delta \mathbf{v}_{\text{odq}n}$ in Figure \ref{fig:passivation}.
\end{lemma}
%Lemma \ref{lemma: Panos} is a direct outcome of Theorem 4 in [Panos]. 
Suppose that an IBR manufacturer provides an $\mathcal{L}_2$ gain $\gamma_n$, the NMPs can leverage condition \eqref{eq: constraints} to find $\alpha_n$, $\beta_n$, and $\kappa_n$. As a result, the closed-loop system shown in Figure \ref{fig:passivation} satisfies Condition 1. The remaining question is: \emph{how does the IBR manufacturer compute $\gamma_n$?}

%\begin{algorithm}
%    \caption{Non-manufacturer Parties' Algorithm} \label{alg:3ed_algorithm}
%        \begin{algorithmic}[1]
%        \State \textbf{input}: $\gamma_n$
%        \State $\bar{\beta} \leftarrow 0.5$; $\bar{\sigma} \leftarrow 0.5$; $\texttt{flag} \leftarrow 0$ 
%        %\State \algorithmiccomment{$0.5$ can be replaced by any positive number}
%        \While{$k_1 = 1, 2, \ldots, K$}
%        \State Pick an arbitrary $\sigma_0>0$ and $\sigma_0 \ne \bar{\sigma}$
%        \While{$k_2 = 1, 2, \ldots, K$}
%            \State Pick an arbitrary $\beta_0>0$ and $\beta_0 \ne \bar{\beta}$
%            \State $\kappa_0 \leftarrow \beta_0/(1.05\gamma_n)$; \quad $\alpha_0 \leftarrow 2\sigma_n\beta_0/\gamma_n - 1/\gamma_n$
%            \If{$\kappa_0 > \alpha_0 \beta_0 \wedge \alpha_0>0$ } 
%            \State $\alpha_n \leftarrow \alpha_0$; $\beta_n \leftarrow \beta_0$; $\kappa_n \leftarrow \kappa_0$; $\texttt{flag} \leftarrow 1$
%            \State \textbf{break}
%            \Else $\quad \bar{\beta} \leftarrow \beta_0$
%            \EndIf
%        \EndWhile
%            \If{$\texttt{flag}==1$} \textbf{break}
%            \Else $\quad \bar{\sigma} \leftarrow \sigma_0$
%            \EndIf
%        \EndWhile
%        \State\Return $\alpha_n$, $\beta_n$, $\kappa_n$
%        \end{algorithmic}
%\end{algorithm}

\subsection{$\mathcal{L}_2$ Gain for IBRs}
The following Lemma can be leveraged by IBR manufacturers to obtain $\gamma_n$:
\begin{lemma} \label{lemma: Khalil}
    \cite{nonlinear_ctr} Assume that the real part of every eigenvalue of matrix $A_n$ in \eqref{eq:compact_deviation} is strictly negative. Let $G_n(s) = C_n(sI-A_n)^{-1}B_n$. Then, the $\mathcal{L}_2$ gain of dynamics \eqref{eq:compact_deviation} is $\sup_{\omega \in \mathbb{R}} \norm{G_n(\mathbbm{j} \omega)}_2$.
\end{lemma}
In Lemma \ref{lemma: Khalil}, $\norm{\cdot}_2$ is the $\mathcal{L}_2$ norm; transfer functions $G_n(s)$ can be obtained by the ``$\texttt{ss2tf}$'' function in MATLAB based on matrices $A_n$, $B_n$, and $C_n$; $\mathbbm{j} = \sqrt{-1}$; and $\sup_{\omega \in \mathbb{R}} \norm{G_n(j \omega)}_2$ is the $H_{\infty}$ norm of $G_n(\mathbbm{j} \omega)$ \cite{nonlinear_ctr} which can be obtained by the ``$\texttt{hinfnorm}$'' function \textcolor{black}{\cite{bruinsma1990fast}} in MATLAB, given $G_n(s)$. Lemma \ref{lemma: Khalil} requires a stable matrix $A_n$. This is not a big assumption, as IBR control designers typically perform small-signal analysis to ensure device-level stability.
%\textcolor{black}{It is also possible to obtain the $\mathcal{L}_2$ gain in a data-driven manner based on recent advances \cite{learn_L2} in the automatic control community.}
%Lemma \ref{lemma: Khalil} leads to Algorithm \ref{alg:3ed_algorithm} by which the IBR manufacturers can compute the $\mathcal{L}_2$ gains for their IBRs.

%\begin{algorithm}
%    \caption{Manufacturers' Algorithm} \label{alg:man_algorithm}
%        \begin{algorithmic}[1]
%        \State \textbf{input}: $A_n$, $B_n$, and $C_n$
%        \State Obtain $G_n(s)$ by $G_n(s)= C_n(sI-A_n)^{-1}B_n$
%        \State $\gamma_n \leftarrow \texttt{hinfnorm}(G_n)$
%        \State \Return $\gamma_n$
%        \end{algorithmic}
%\end{algorithm}

\subsection{Architecture of Stability Enforcement Interfaces (SEI)}
This subsection conceptualizes an interface that enforce Condition 1, and the theoretical result in \cite{Panos_thm} is translated into electric energy systems for the first time. The physical layer of the interface is shown in Figure \ref{fig:physical_layer}. The interface comprises a three-phase, controlled volage source, and a three-phase controlled current source. The voltage $\Delta \mathbf{v}_{\text{abc}n}:=[\Delta v_{\text{a}n}, \Delta v_{\text{b}n}, \Delta v_{\text{c}n}]^{\top}$ of the voltage source and the current $\Delta \mathbf{i}_{\text{abc}n}:=[\Delta i_{\text{a}n}, \Delta i_{\text{b}n}, \Delta i_{\text{c}n}]^{\top}$ of the current source are determined by the terminal voltage measurement $\mathbf{v}_{\text{abc}n}:= [v_{\text{a}n}, v_{\text{b}n}, v_{\text{c}n}]^{\top}$ and current measurement $\mathbf{i}_{\text{abc}n}:= [i_{\text{a}n}, i_{\text{b}n}, i_{\text{c}n}]^{\top}$ of the IBR $n$. \textcolor{black}{This paper focuses on the control law that establishes the link between $\{\mathbf{v}_{\text{abc}n}, \mathbf{i}_{\text{abc}n}\}$ and $\{\Delta \mathbf{v}_{\text{abc}n}, \Delta \mathbf{i}_{\text{abc}n}\}$; the internal design of the controlled voltage and current sources is out of the scope of this paper.}

\begin{figure}
    \centering
    \includegraphics[width = 2.3in]{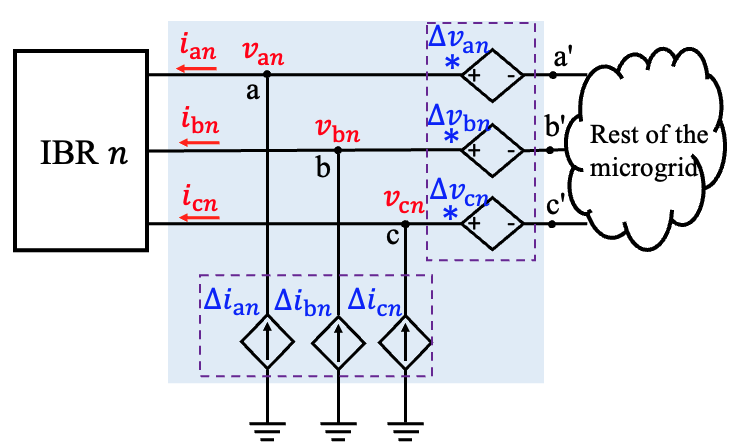}
    \caption{\textcolor{black}{Physical layer of the stability enforcement interface.}}
    \label{fig:physical_layer}
\end{figure}

\begin{figure}
    \centering
    \includegraphics[width = 3.5in]{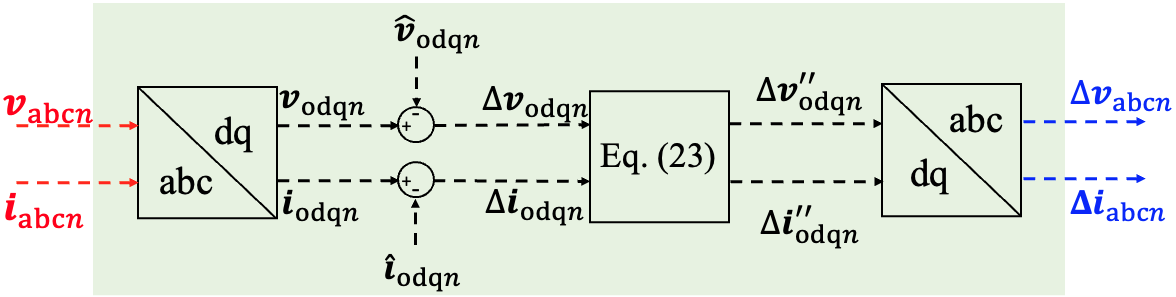}
    \caption{Cyber layer of the stability enforcement interface.}
    \label{fig:cyber_layer}
\end{figure}

Figure \ref{fig:cyber_layer} presents the cyber layer of the interface. In Figure \ref{fig:cyber_layer}, the three-phase variables $\mathbf{v}_{\text{abc}n}$ and $\mathbf{i}_{\text{abc}n}$ are first transformed into the d-q frame by the Park transformation
\cite{dq0trans}. Second, the deviation vectors $\Delta \mathbf{v}_{\text{odq}n}$ and $\Delta \mathbf{i}_{\text{odq}n}$ are obtained by subtracting the steady-state values $\hat{\mathbf{v}}_{\text{odq}n}$ and $\hat{\mathbf{i}}_{\text{odq}n}$ from $\mathbf{v}_{\text{odq}n}$ and $\mathbf{i}_{\text{odq}n}$. Third, $\Delta \mathbf{v}_{\text{odq}n}''$ and $\Delta \mathbf{i}_{\text{odq}n}''$ are computed by
\begin{subequations} \label{eq:source_control_law}
    \begin{align}
        & \Delta \mathbf{v}_{\text{odq}n}'' = (I-\kappa_nI) \Delta \mathbf{v}_{\text{odq}n} - \beta_n I \Delta \mathbf{i}_{\text{odq}n} \\
        & \Delta \mathbf{i}_{\text{odq}n}'' = -\alpha_n I \Delta \mathbf{v}_{\text{odq}n}.
    \end{align}
\end{subequations}
Finally, the vectors in the d-q frame $\Delta \mathbf{v}_{\text{odq}n}''$ and $\Delta \mathbf{i}_{\text{odq}n}''$ are transformed to the three-phase frame.

Equation \eqref{eq:source_control_law} is justified by transforming Figure \ref{fig:physical_layer} in the three-phase frame to the d-q frame. Figure \ref{fig:physical_layer_dq} presents the circuit in the d-q frame. According to Figure \ref{fig:passivation}, we have
\begin{subequations} \label{eq:step1}
    \begin{align}
        & \Delta \mathbf{v}_{\text{odq}n}' = \kappa_n I \Delta \mathbf{v}_{\text{odq}n} + \beta_n I \Delta \mathbf{i}_{\text{odq}n} \\
        & \Delta \mathbf{i}_{\text{odq}n}' = \alpha_n I \Delta \mathbf{v}_{\text{odq}n} + \Delta \mathbf{i}_{\text{odq}n}.
    \end{align}
\end{subequations}
In Figure \ref{fig:physical_layer_dq}, based on the Kirchhoff's circuit laws, we have
\begin{subequations} \label{eq:step2}
    \begin{align}
        & \Delta \mathbf{v}_{\text{odq}n}' = \Delta \mathbf{v}_{\text{odq}n} - \Delta \mathbf{v}_{\text{odq}n}'' \\
        & \Delta \mathbf{i}_{\text{odq}n}' = \Delta \mathbf{i}_{\text{odq}n} - \Delta \mathbf{i}_{\text{odq}n}''.
    \end{align}
\end{subequations}
Plugging \eqref{eq:step2} into \eqref{eq:step1} leads to \eqref{eq:source_control_law}.

\textcolor{black}{\emph{Remark 1:}} Designing the interface shown in Figures \ref{fig:physical_layer} and \ref{fig:cyber_layer} only requires an IBR manufacturer to provide the $\mathcal{L}_2$ gains of their IBRs which can be easily obtained via Lemma \ref{lemma: Khalil} by the manufacturer. The interface design does not need the information of detailed IBR control. While the IBR manufacturer may be reluctant to share such information with the NMPs due to privacy concerns on intellectual properties, revealing the $\mathcal{L}_2$ of the IBRs does not lead to such privacy issues, as it is impossible to infer the detailed control design of an IBR merely based on the $\mathcal{L}_2$ gains of the IBR. %\textcolor{blue}{In addition, we assume that the SEI's dynamics is much faster than the microgrid's dynamics. As a result, the SEI is represented by ideal, controlled three-phase voltage and current sources. This paper only address how to change to setpoints of these sources for the stabilization purpose, and how to realize these ideal voltage and current sources through power-electronic techniques is out of the scope of this paper.}

\textcolor{black}{\emph{Remark 2:} 
The control law \eqref{eq:source_control_law} is decentralized, as it only requires local IBR terminal voltage and current measurements. Since the control law \eqref{eq:source_control_law} only includes simple linear combination operations, and the computational effort required to determine the control command does not increase with system size, the proposed SEI is scalable to large-scale IBR-dominated systems.}
    
\begin{figure}
    \centering
    \includegraphics[width = 2.5in]{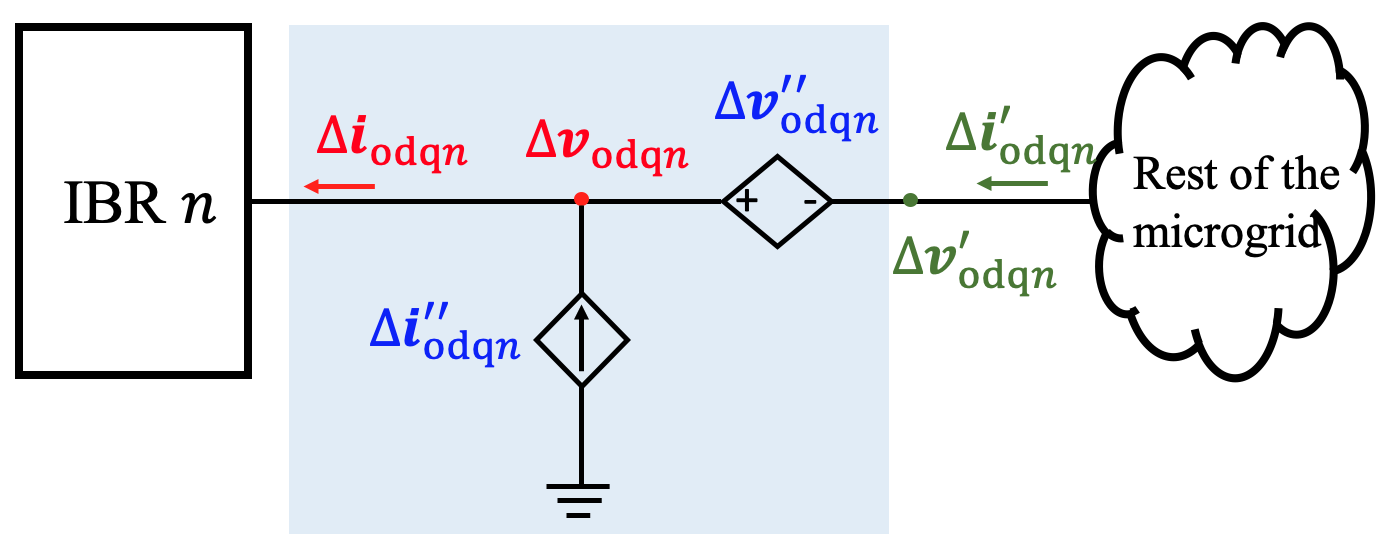}
    \caption{Physical layer of the stability enforcement interface in the d-q frame.}
    \label{fig:physical_layer_dq}
\end{figure}

\section{Case Study} \label{sec:case_study}
%\subsection{Test System Description}
%This section tests the effectiveness of the SEIs by simulating 2-IBR and 5-IBR microgrids. %First, we present a motivating example where connecting an IBR with a microgrid introduces system-level instability. Then, we show that enabling the interfaces at the grid edges can stabilize the system in a decentralized, non-intrusive manner.

\subsection{A Two-IBR Microgrid}

\subsubsection{A motivating example} \label{sec:motivating_eg}
The test system in Figure \ref{fig:two_ibr}-(a) contains two IBRs. All control parameters of IBR $1$ can be found in \cite{TimGreenModel}. For IBR $2$, $k_{\text{iv}2}= 78$, and the rest of the parameters are from \cite{TimGreenModel}. The two loads are constant-impedance, and the per-phase impedances of Loads $1$ and $2$ are 25$\Omega$ and 20$\Omega$, respectively. Before time $t=0.4$s, the test system contains two microgrids operating in an islanded mode. At $t=0.4$s, the two microgrids in Figure \ref{fig:two_ibr}-(a) are networked via the tie line and they enter the hybrid mode. Figure \ref{fig: motivation_example_iodq12} visualizes the terminal currents of the two IBRs, i.e., $\mathbf{i}_{\text{odq}1}$ and $\mathbf{i}_{\text{odq}2}$ in the d-q frame, from $0.2$s to $1$s. Figures \ref{fig: motivation_example_iodq12} presents that $\mathbf{i}_{\text{odq}1}$ and $\mathbf{i}_{\text{odq}2}$ are constant before the two microgrids are networked, i.e., $t<0.4$s. This suggests the two microgrids in the islanded mode are stable. However, after the two microgrids are networked, $\mathbf{i}_{\text{odq}1}$ and $\mathbf{i}_{\text{odq}2}$ keep oscillating.

%Figure \ref{fig: motivation_example_iodq12} examines the three-phase currents $\mathbf{i}_{\text{abc}1}$ and $\mathbf{i}_{\text{abc}2}$ in the d-q frame: before $t=0.4$s, both $\mathbf{i}_{\text{odq}1}$ and $\mathbf{i}_{\text{odq}2}$ can be stabilized at their nominal values. However, after the switch is closed at $t=0.4$s, both $\mathbf{i}_{\text{odq}1}$ and $\mathbf{i}_{\text{odq}2}$ keep oscillating with increasing amplitudes, suggesting that the two networked microgrids become unstable.
\begin{figure}
    \centering
    \includegraphics[width = 0.5\linewidth]{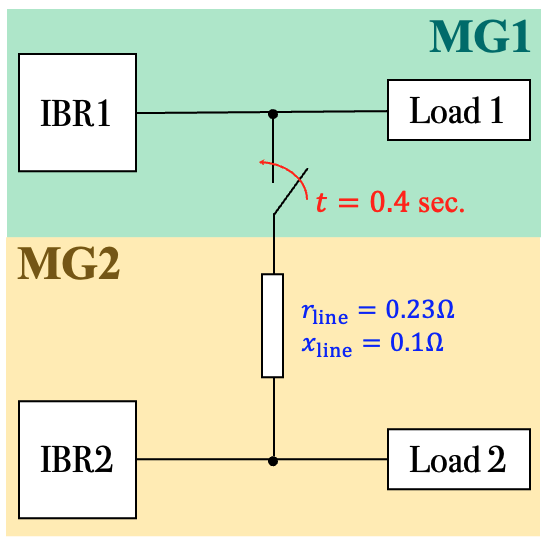}
    \caption{A 2-IBR microgrid (MG)}
    \label{fig:two_ibr}
\end{figure}

\begin{figure}
				\centering
				\subfloat[]{\includegraphics[width=0.48\linewidth]{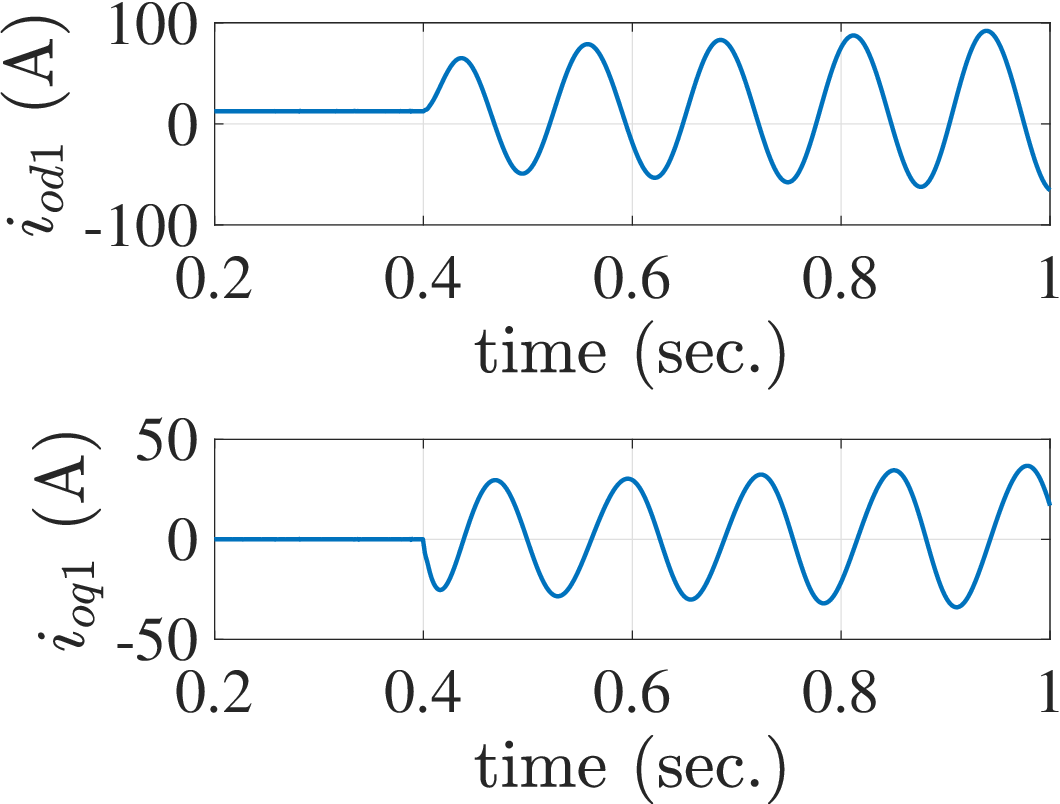}}
				\hfil
				\subfloat[]{\includegraphics[width=0.48\linewidth]{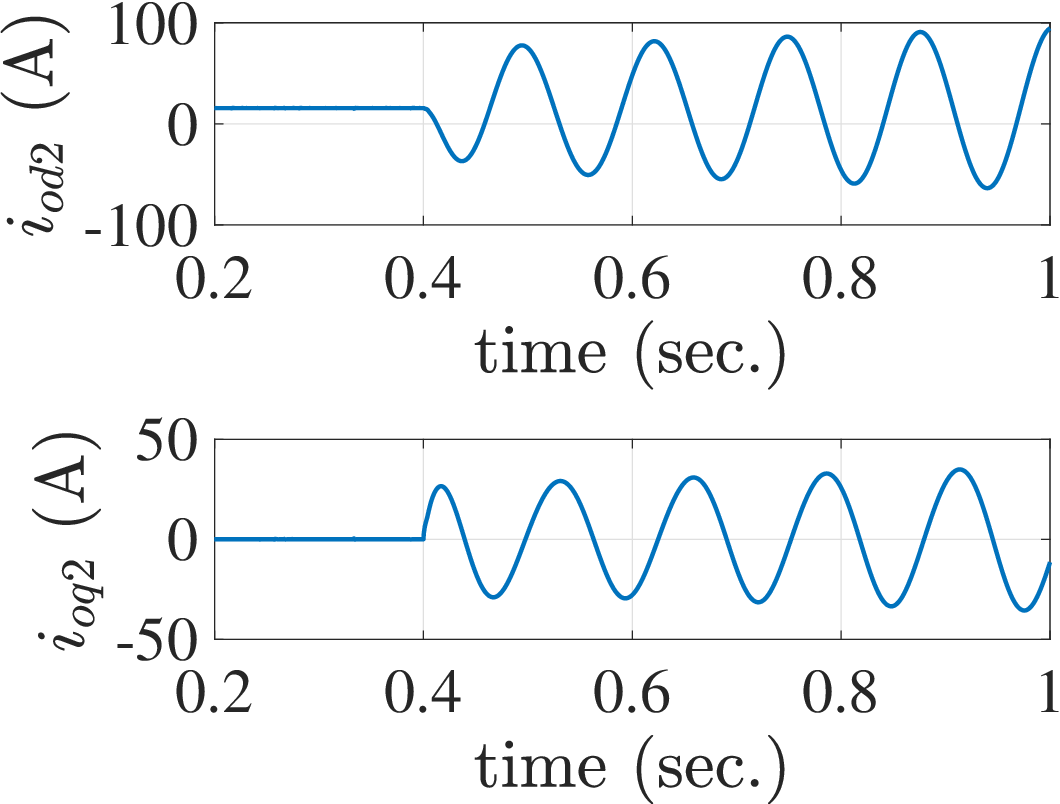}}
				\hfil
				\caption{(a) $\{i_{\text{od}1},i_{\text{oq}1}\}$ and (b) $\{i_{\text{od}2}, i_{\text{oq}2}\}$ \textcolor{black}{with constant-impedance loads}}
				\label{fig: motivation_example_iodq12}
			\end{figure}

\textcolor{black}{Next, we examine the importance of modeling fast, high-order dynamics of IBRs. If we only model the dynamics with the slow states, i.e., the droop controllers, under the disturbance at $0.4$s, the real power output $P_1$ of IBR 1 is visualized by the orange dashed waveform in Figure \ref{fig: detailVSsimple}. With the simplified model, it can be observed that the two networked IBRs are stable. However, if the \emph{dynamics of both fast and slow states} are modeled, under the same disturbance, the blue-solid curve in Figure \ref{fig: detailVSsimple} visualizes $P_1$, and it suggests the two networked IBRs are actually unstable since a growing oscillation is incurred. Such instability cannot be observed from the simulation with the simplified model. Therefore, modeling the dynamics of the fast states is also important for the stability analysis.}

\begin{figure}
    \centering
    \includegraphics[width = 1.8in]{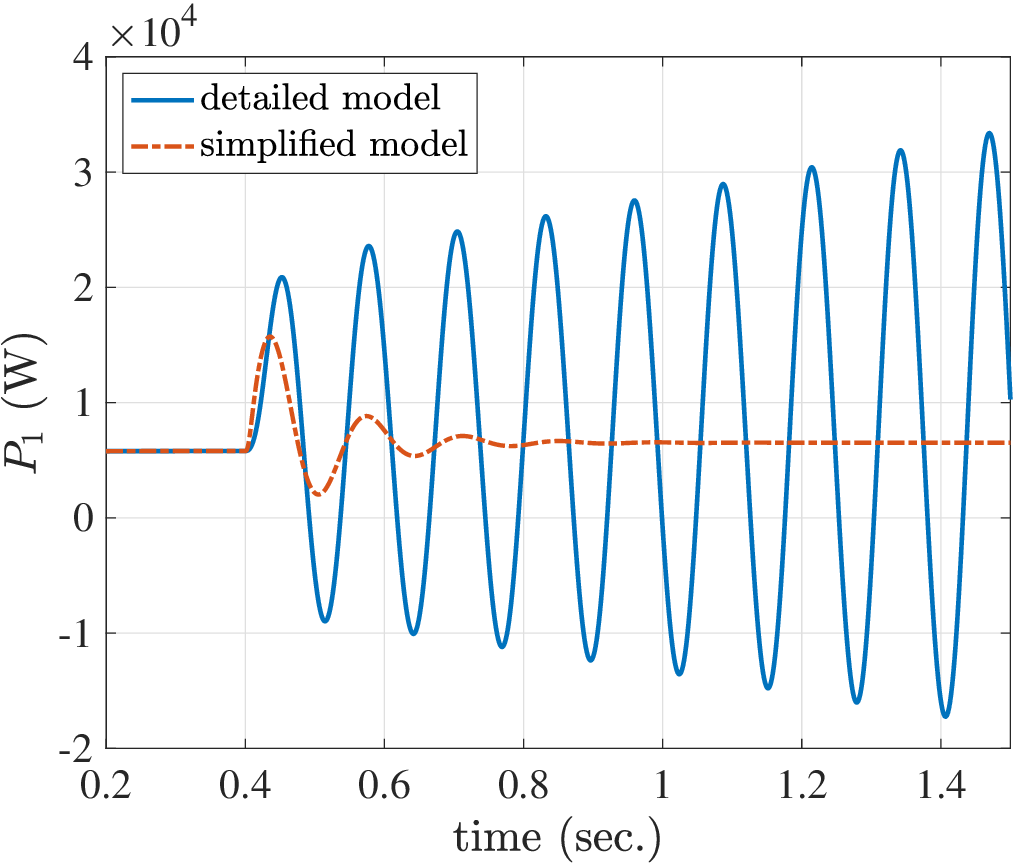}
    \caption{\textcolor{black}{Response comparison between detailed and simplified models.}}
    \label{fig: detailVSsimple}
\end{figure}

%\subsection{Performance of Protocol Enforcement Interfaces}
\subsubsection{System responses with stability enforcement interface} \label{ssubs:all-SEIs}
\textcolor{black}{With the same setting of Section \ref{sec:motivating_eg}, each IBR connects \textcolor{black}{to} a SEI shown in Figure \ref{fig:physical_layer}. The manufacturer of each IBR can use Lemma \ref{lemma: Khalil} to obtain the $\mathcal{L}_2$ gain $\gamma_n$ of the IBR. The $\mathcal{L}_2$ gains $\gamma_1$ and $\gamma_2$ for the two IBRs are $4.43$ and $2.9$, respectively. Based on the $\mathcal{L}_2$ gains, the parameters for the SEIs are: $\alpha_1 = 0.00045$, $\beta_1 = 1.67$, $\kappa_1=0.36$, $\alpha_2 = 0.00097$, $\beta_2 = 2.18$, and $\kappa_2=0.72$. The resulting $\sigma_1$ and $\sigma_2$ are $0.3$ and $0.23$.}

%With the $\mathcal{L}_2$ gain $\gamma_n$, NMPs can find the parameters of each SEI, i.e., $\alpha_n$, $\beta_n$, and $\kappa_n$, via Algorithm \ref{alg:3ed_algorithm}. It is worth noting that the manufacturer does not need to share the detailed model of their IBRs with the NMPs to enable them to design the SEI. The $\mathcal{L}_2$ gain $\gamma_n$ obtained from Algorithm \ref{alg:man_algorithm} and the interface parameters are: XXXXXXX It can be easily verified that condition \eqref{eq: constraints} is satisfied.

Figure \ref{fig: iodq_corrected} visualizes the d-q components $\mathbf{i}_{\text{odq}1}$ and $\mathbf{i}_{\text{odq}2}$: the SEIs can stabilize the currents at constant values after the two IBRs are networked, while both $\mathbf{i}_{\text{odq}1}$ and $\mathbf{i}_{\text{odq}2}$ would keep oscillating with increasing amplitudes if no SEI was installed. %Note designing the SEIs does not require the availability of the detailed models of IBRs for the NMPs.

    %\begin{table} 
    %\caption{Parameters of Protocol Enforcement Interface}
    %\centering
    %    \begin{tabular}{c|c|c|c|c|c}
    %        \hline
    %        \hline
    %        $n$ & $\gamma_n$ & $\sigma_n$ &$\alpha_n$ & $\beta_n$ & $\kappa_n$  \\
    %        \hline
    %         $1$&$6.45$&$3$&$0.0031$&$0.17$ &$0.0251$\\
    %         \hline
    %         $2$&$4.23$&$3.5$&$0.0118$&$0.15$ & $0.0338$\\
    %        \hline
    %        \hline
    %    \end{tabular}
    %    \label{tab:parameters}
    %\end{table}

\begin{figure}
				\centering
				\subfloat[]{\includegraphics[width=0.5\linewidth]{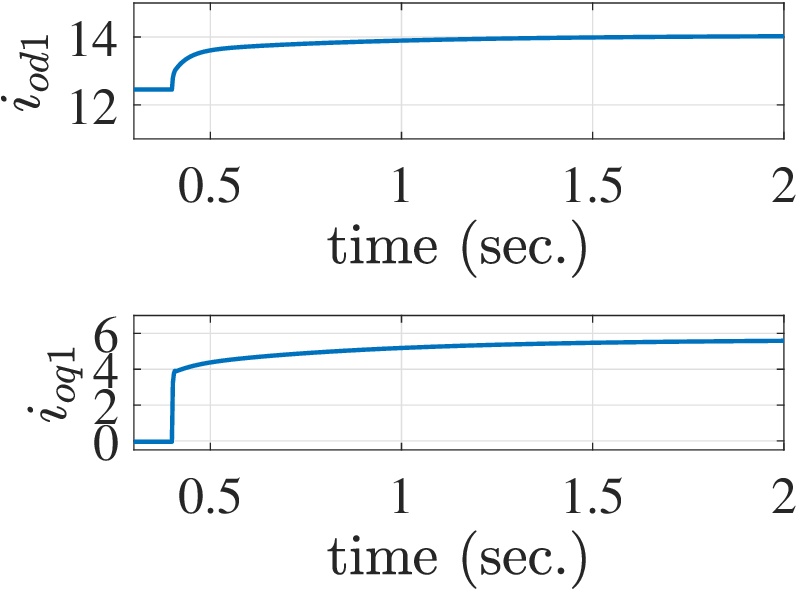}}
				\hfil
				\subfloat[]{\includegraphics[width=0.5\linewidth]{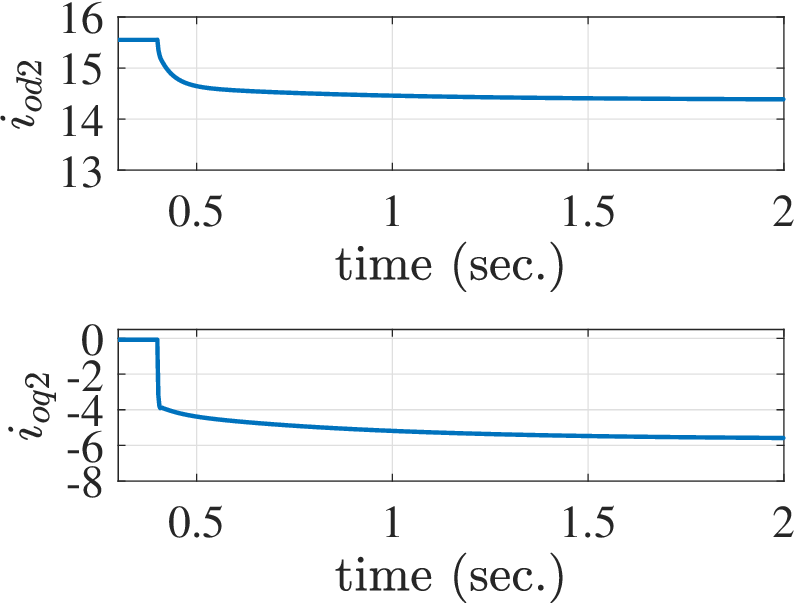}}
				\hfil
				\caption{\textcolor{black}{(a) $\textbf{i}_{\text{odq}1}$ and (b) $\textbf{i}_{\text{odq}2}$ with the SEIs.}}
				\label{fig: iodq_corrected}
			\end{figure}

\subsubsection{Energy changed by SEIs} \label{ssubs:all-SEIs-power}
\emph{Do the SEIs consume significant amount of energy to stabilize the microgrids?} We answer this question by comparing the energy consumed by the interfaces with the energy produced by the IBRs. For $n= 1,2$, denote by $P_n$, $P_{\text{c}n}$, and $P_{\text{v}n}$ the real power \emph{produced} by IBR $n$, the real power \emph{consumed} by the three-phase, shunt current source in the SEI at IBR $n$, and the real power \emph{consumed} by the three-phase, series voltage source in the SEI at IBR $n$, respectively. Denote by $E_n$, $E_{\text{c}n}$, and $E_{\text{v}n}$ the energy produced by IBR $n$, the energy consumed by the three-phase current source in the SEI at IBR $n$, and the energy consumed by the three phase voltage source in the SEI at IBR $n$. 

\textcolor{black}{Figure \ref{fig: IBR_power12} visualizes $P_n$, $P_{\text{c}n}$, and $P_{\text{v}n}$.
In Figure \ref{fig: IBR_power12}-(a), it can be observed that the real power used for stabilizing the microgrids, i.e., $P_{\text{c}1}$ and $P_{\text{v}1}$, is much less than $P_1$. By integrating $P_1$, $P_{\text{c}1}$, and $P_{\text{v}1}$ over a period, $E_1$, $E_{\text{c}1}$, and $E_{\text{v}1}$ over the period, can be computed. Table \ref{tab:energy} presents $E_1$, $E_{\text{c}1}$, and $E_{\text{v}1}$ over the transient process (i.e., the process from $0.4$s to $1.5$s) and the steady state (i.e., the process from $1.5$s to $2$s). Let $E_{\text{I}n}=E_{\text{c}1} + E_{\text{v}1}$ for $n=1,2$. The SEI at IBR 1 only takes a very small amount of energy, i.e., $2.4\%$ of total energy produced by IBR 1 during the transients, to stabilize the microgrids. In the steady state, the energy consumed by the SEI is only $2.8\%$ of the total energy produced by IBR $1$.}
Similarly, Figure \ref{fig: IBR_power12}-(b) shows that the absolute value of real power consumed by the interface at IBR $2$ is much smaller than that produced by IBR $2$. 
%The values of $E_2$, $E_{\text{c}2}$ and $E_{\text{v}2}$ over the transient process ($0.4$s - \textcolor{black}{$1.5$s}) and the steady state (\textcolor{black}{$1.5$s - $2$s}) are reported in Table \ref{tab:energy}. Compared with the energy produced by IBR $2$, the energy produced by IBR $2$ for the stabilization purpose is very small, i.e., \textcolor{black}{$0.2\%$} of $E_2$ during the transients and \textcolor{black}{$0.3\%$} of $E_2$ during the steady state.

%\begin{figure}
%				\centering
%				\subfloat[]{\includegraphics[width=0.5\linewidth]{Fig/Power_IBR1.png}}
%				\hfil
%				\subfloat[]{\includegraphics[width=0.5\linewidth]{Fig/Power_IBR2}}
%				\hfil
%				\caption{\textcolor{black}{(a) Time-domain evolution of $P_1$, $P_{\text{v}1}$, and $P_{\text{c}1}$ at IBR 1; and (b) time-domain evolution of $P_2$, $P_{\text{v}2}$, and $P_{\text{c}2}$}}
%				\label{fig:IBR_power}
%			\end{figure}

\begin{figure}
				\centering
				\subfloat[]{\includegraphics[width=0.8\linewidth]{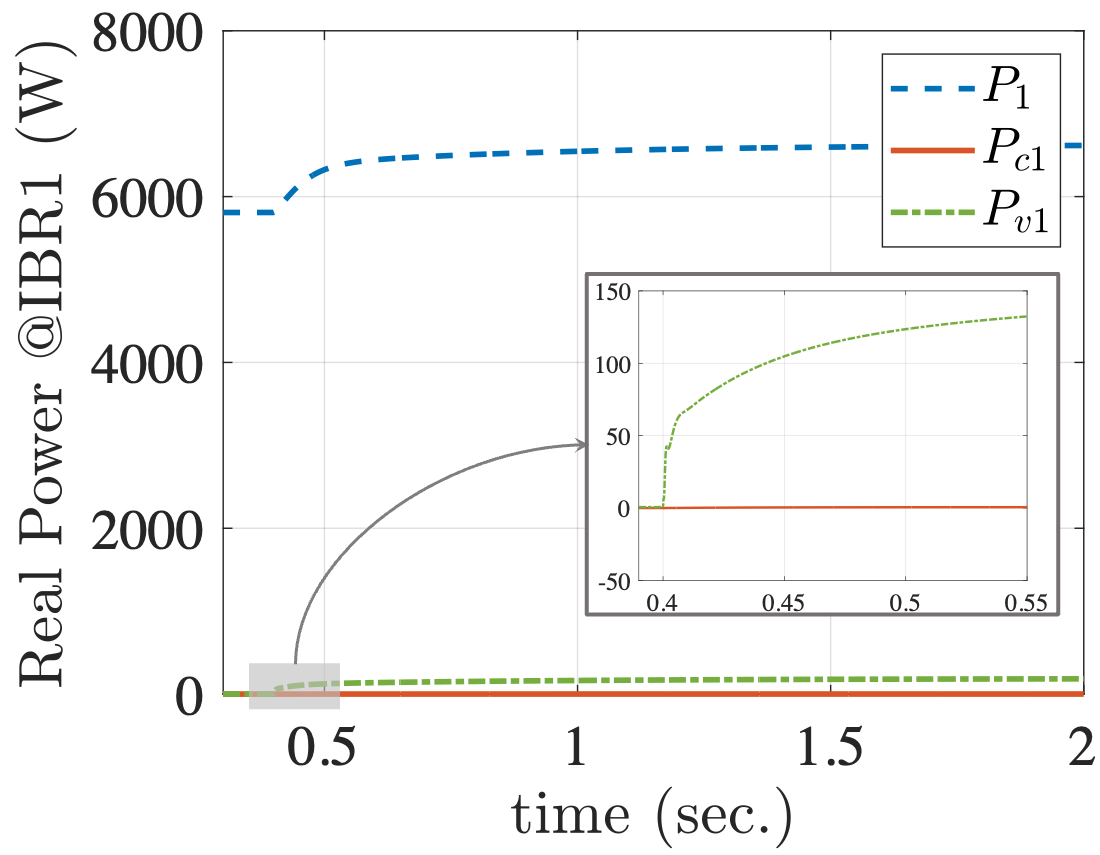}}
				\hfil
				\subfloat[]{\includegraphics[width=0.8\linewidth]{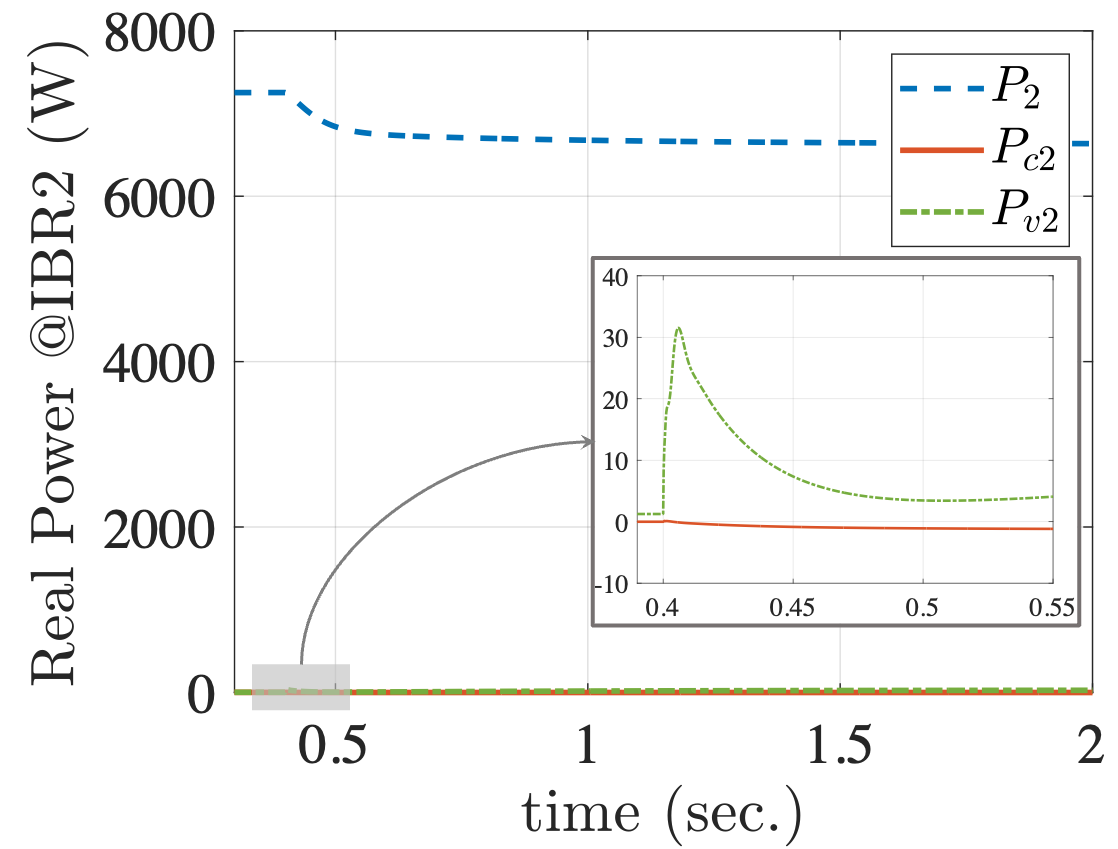}}
				\hfil
				\caption{\textcolor{black}{(a) $P_1$, $P_{\text{v}1}$, and $P_{\text{c}1}$ at IBR 1, and (b) $P_2$, $P_{\text{v}2}$, and $P_{\text{c}2}$ at IBR 2.}}
				\label{fig: IBR_power12}
			\end{figure}

 \begin{table} 
    \caption{\textcolor{black}{Energy Analysis for Networked Microgrids with Two IBRs}}
    \centering
        \begin{tabular}{c|c|c|c|c}
            \hline
            \hline
            \textbf{Period} & $E_1$ (J) & $E_{\text{c}1}$ (J) & $E_{\text{v}1}$ (J) & $\abs{E_{\text{I}1}/E_1}$ (\%)\\
            \hline
            $0.4$s - $1.5$s & $7140$ & $0.7$ & $170$ & $2.4\%$ \\
            \hline
            $1.5$s - $2$s & $3304$ & $0.4$ & $91$ & $2.8\%$\\
            \hline
            \hline
            \textbf{Period} & $E_2$ (J) & $E_{\text{c}2}$ (J) & $E_{\text{v}2}$ (J) & $\abs{E_{\text{I}2}/E_2}$ (\%)\\
            \hline
            $0.4$s - $1.5$s & $7391$ & $-1.4$ & $16$ & $0.2\%$ \\
            \hline
            $1.5$s - $2$s & $3320$ & $-0.7$ & $12$ & $0.3\%$ \\
            \hline
            \hline
        \end{tabular}
        \label{tab:energy}
    \end{table}

\subsubsection{\textcolor{black}{Performance in the presence of constant-power loads}} \label{subs:two_ibr_CPL} 
Next, we examine the performance of the SEIs in the presence of constant-power loads. The two constant-impedance loads in Section \ref{sec:motivating_eg} are replaced by two constant-power loads. In our simulation, the two constant-power loads are modeled by the Simulink block called ``Three-Phase Dynamic Load'' with the ``External Control of PQ'' option selected. The real power for Loads $1$ and $2$ is $5784$W and $7226$W, respectively; and there is no reactive power for both loads. At $t=0.4$s, the two microgrids are networked. After $0.4$s, instability similar to that in Figure \ref{fig: motivation_example_iodq12} can be observed. The figure of the instability is omitted for the sake of space. With each IBR equipped with a SEI, Figure \ref{fig: iodq_corrected_CPL} presents the terminal d-q currents of the two IBRs and suggests the system-level symptom is mitigated.

\begin{figure}
				\centering
				\subfloat[]{\includegraphics[width=0.5\linewidth]{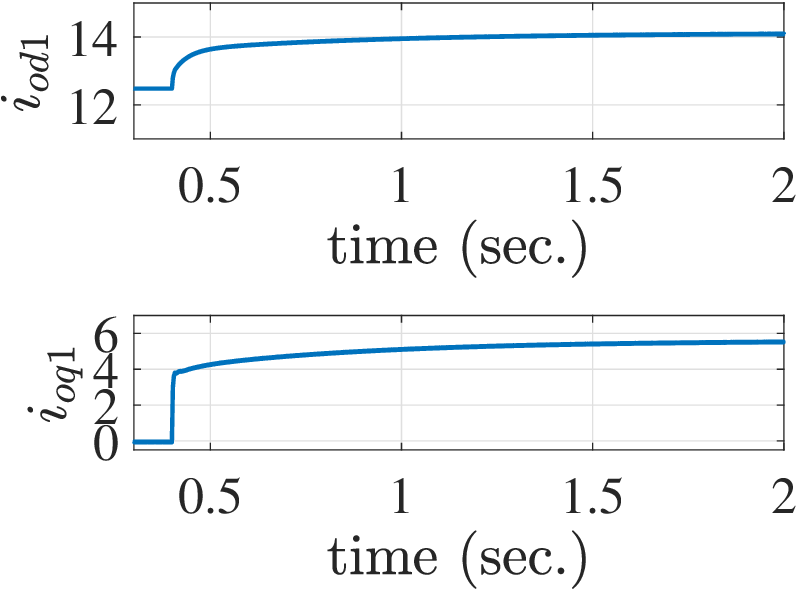}}
				\hfil
				\subfloat[]{\includegraphics[width=0.5\linewidth]{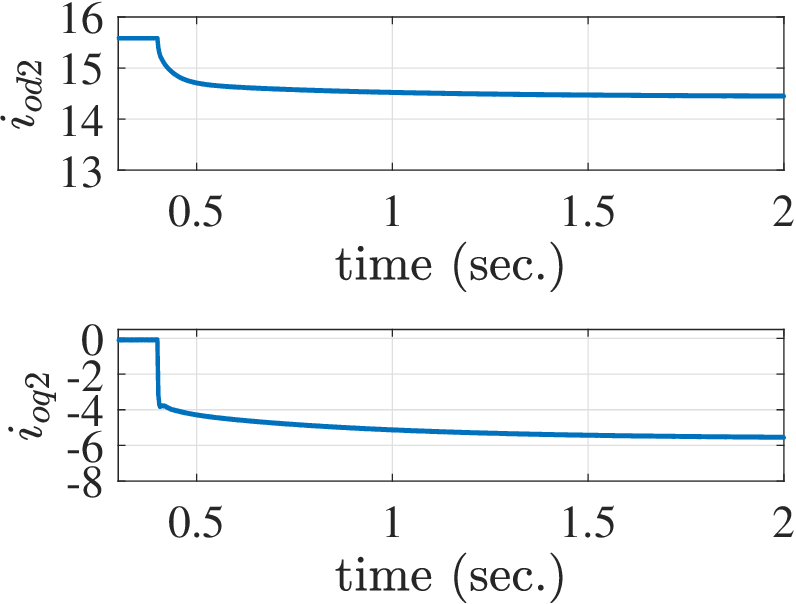}}
				\hfil
				\caption{\textcolor{black}{(a) $\textbf{i}_{\text{odq}1}$ and (b) $\textbf{i}_{\text{odq}2}$ with the SEIs and constant-power loads.}}
				\label{fig: iodq_corrected_CPL}
			\end{figure}

\subsubsection{\textcolor{black}{Comparison studies}}

\textcolor{black}{We compare the proposed approach with an existing passivity-based approach in \cite{tsinghua_condition}. Note that the approach in \cite{tsinghua_condition} requires one to reprogram the internal IBR controllers, which may be infeasible for NMPs, whereas the proposed approach can stabilize the system in a non-intrusive manner. The method in \cite{tsinghua_condition} is implemented by replacing the frequency droop controller with the angle droop controllers, and tuning the control parameters based on the condition derived in \cite{tsinghua_condition}. 
Under the disturbance, the terminal currents of the two IBRs are visualized by the orange-dashed curves in Figure \ref{fig: iodq_compare_network}. It can be observed that the method in \cite{tsinghua_condition} can stabilize the microgrids. With the SEIs, the terminal currents are presented by the blue-solid curves in Figure \ref{fig: iodq_compare_network}, suggesting that the SEIs can stabilize the microgrids with much less overshooting/undershooting. It is not surprising that the two approaches exhibit distinct behaviors under the same disturbances, due to different controllers. Figure \ref{fig: iodq_compare_network} suggests that both methods can stabilize the microgrids with a settling time of less than $0.7$s. However, the proposed SEIs achieve such a goal \emph{without reprogramming the controllers}.}

\begin{figure}
				\centering
				\subfloat[]{\includegraphics[width=0.8\linewidth]{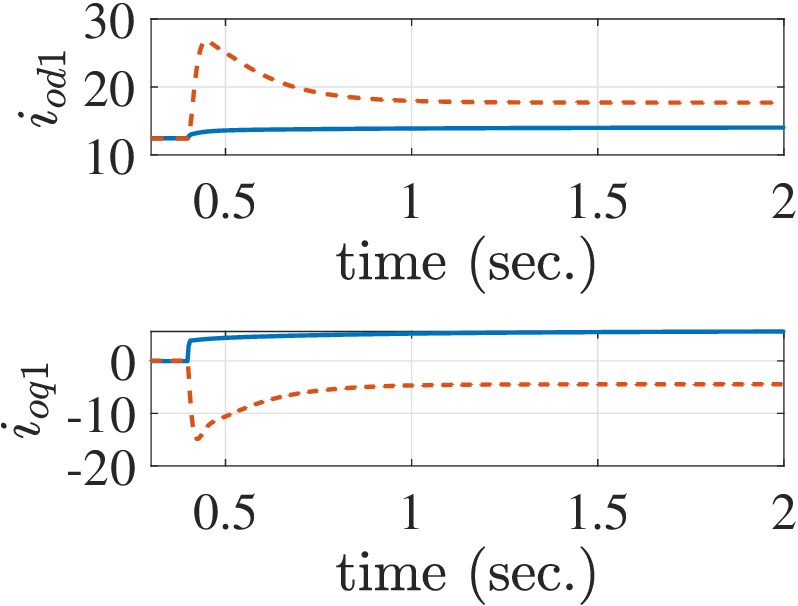}}
				\hfil
				\subfloat[]{\includegraphics[width=0.8\linewidth]{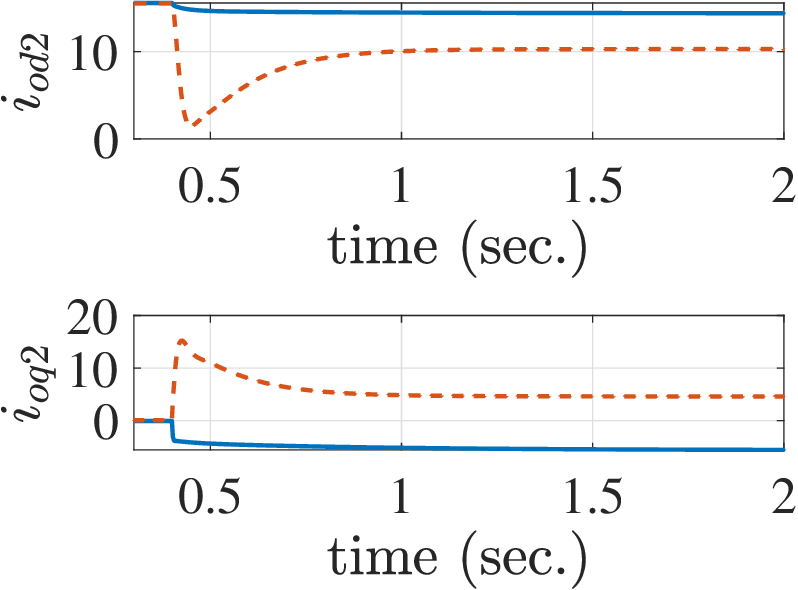}}
				\hfil
				\caption{\textcolor{black}{Comparison of the terminal currents of the proposed method (blue curves) and the intrusive method in \cite{tsinghua_condition} (orange-dashed curves) under the disturbance that the two microgrids are networked.}}
				\label{fig: iodq_compare_network}
			\end{figure}

We proceed to compare the proposed SEI with the impedance-based approach. Before the event at $t=0.4$s, the eigenvalues of the closed-loop dynamics of Microgrid 2 in Figure \ref{fig:two_ibr}-(a) are $-56100\pm \mathbbm{j}314$, $-4294\pm \mathbbm{j} 3869$, $-3883 \pm \mathbbm{j}3552$, $-364\pm \mathbbm{j}695$, and $-372\pm \mathbbm{j}697$. All of them lie in the left-half plane, suggesting the impedance ratio \cite{Sun_09,XWang_14} satisfies the Nyquist stability criterion. Therefore, the controller of IBR 2 in Figure \ref{fig:two_ibr}-(a) can be considered to be tuned by the impedance-based approach. It can be observed from Figure \ref{fig: motivation_example_iodq12} that the system indeed is stable before $0.4$s. However, after networking with Microgrid 1, IBR 2 incurs system-level instability, as shown in Figure \ref{fig: motivation_example_iodq12}. This observation should not be surprising, as the impedance-based approach requires the topology information. After $0.4$s, the microgrid topology changes. According to the impedance-based approach, one needs to re-tune IBR 2 based on the post-event topology that may not be available for an IBR at the grid edge. As discussed in Section \ref{ssubs:all-SEIs}, the proposed SEI can stabilize the microgrid \emph{without the availability of the microgrid topology information} under the same disturbance.

\subsubsection{Partial coverage of SEIs and robustness to measurement noise}
In Section \ref{ssubs:all-SEIs}, each IBR is equipped with a SEI. Here, we remove the SEI at IBR 1. The yellow-dashed curve in Figure \ref{fig: SNR} presents $i_{\text{od}2}$ under the same disturbance in Section \ref{sec:motivating_eg}, suggesting even with partial coverage of SEIs, the microgrid can be stabilized.
Furthermore, as the SEI at IBR 2 makes decisions based on the local measurements $\mathbf{v}_{\text{abc}2}$ and $\mathbf{i}_{\text{abc}2}$, we investigate the robustness of the SEI to noise in these measurements. Figure \ref{fig: SNR} shows the response of $i_{\text{od}2}$ with different noise levels in the sensors measuring $\mathbf{v}_{\text{abc}2}$ and $\mathbf{i}_{\text{abc}2}$. The noise level is quantified by the signal-to-noise ratio (SNR): a smaller SNR suggests larger measurement noise.
The SNR for an industry-graded analog-to-digital converter can be greater than $60$dB \cite{720455}. Figure \ref{fig: SNR} indicates that even with the sensors whose SNR is $50$dB or $30$dB, the SEI can still stabilize the microgrid. Note that since the SEI only uses local measurements \emph{without requiring long-distance communication}, the effect of communication delay is not investigated here.

\begin{figure}
    \centering
    \includegraphics[width = 1.8in]{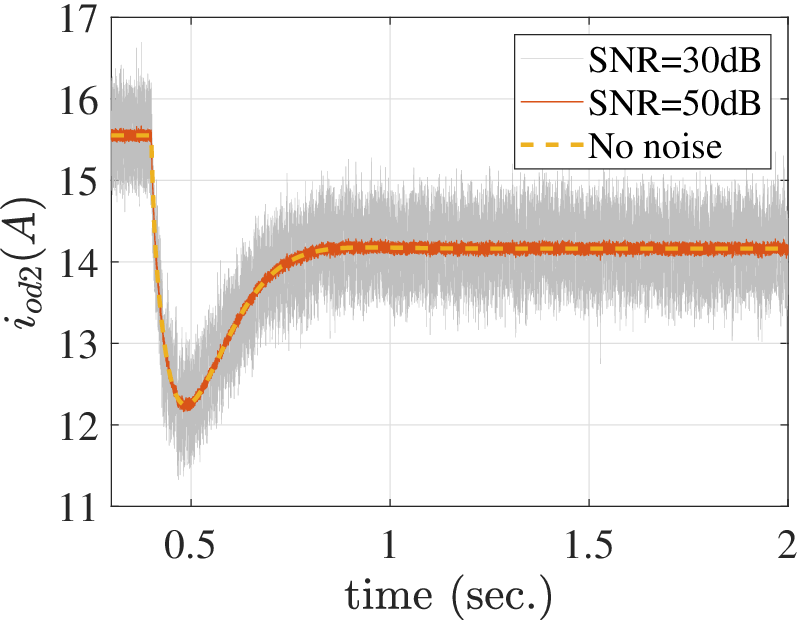}
    \caption{Visualization of $i_{\text{od}2}$ under different levels of noise.}
    \label{fig: SNR}
\end{figure}

\begin{figure}
    \centering
    \includegraphics[width = \linewidth]{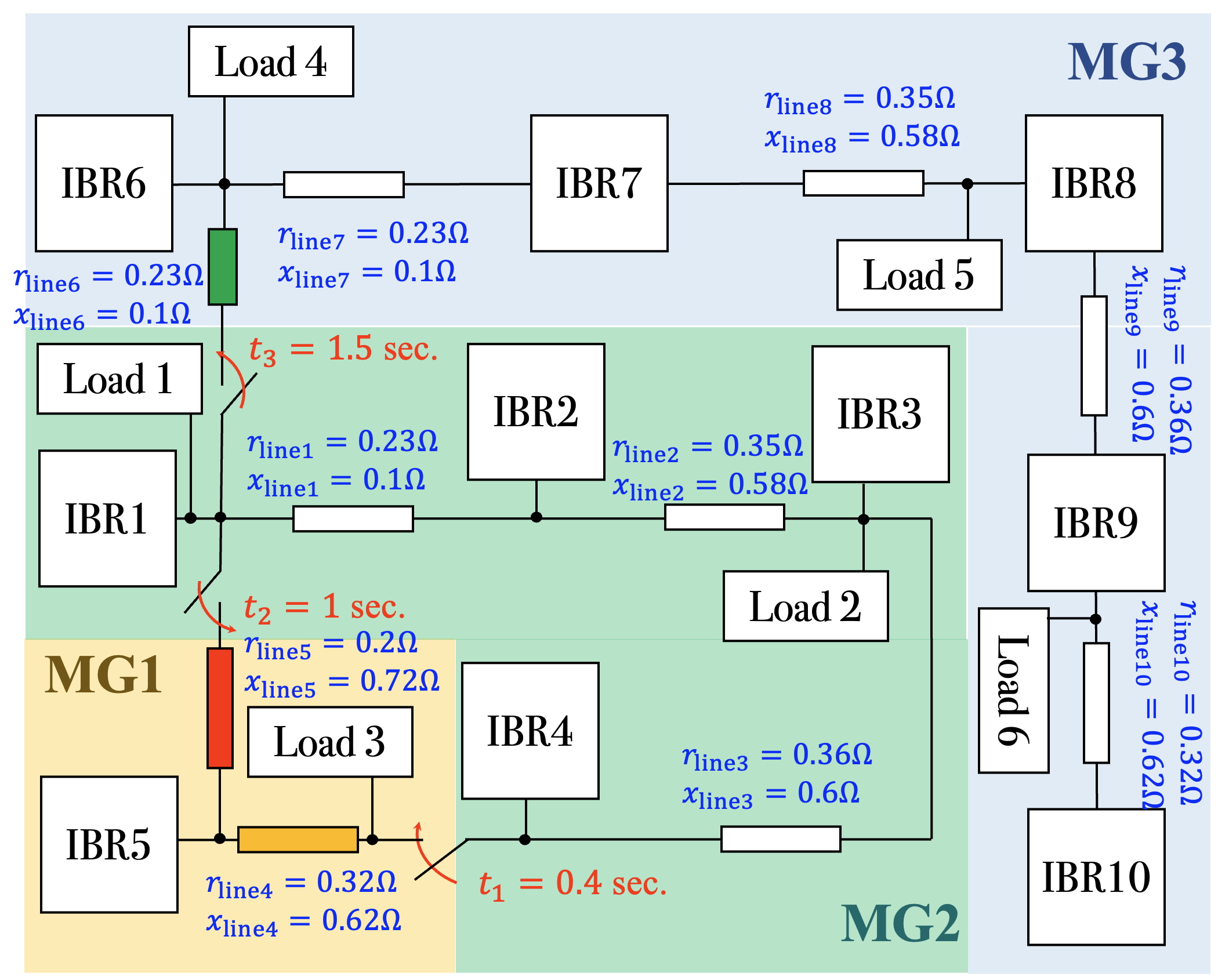}
    \caption{\textcolor{black}{Networked microgrids with 10 IBRs.}}
    \label{fig: ten_IBR_MG}
\end{figure}

\subsection{\textcolor{black}{Networked Microgrids with 10 IBRs}}

\textcolor{black}{We proceed to present the performance of the SEI in three networked microgrids with 10 IBRs in Figure \ref{fig: ten_IBR_MG}. In the simulation, the switching dynamics with a switching frequency of $8$ kHz is modeled for each IBR. For IBR $5$, $k_{\text{iv}5}= 10$, and its rest parameters are from \cite{TimGreenModel}. The control parameters of other IBRs can be found in \cite{TimGreenModel}. Before $t=0.4$s, IBR 5 supplies Load 3 autonomously, operating as an autonomous microgrid, i.e., MG1 in Figure \ref{fig: ten_IBR_MG}, and it connects to the rest of the microgrid through closing the tie line between IBRs 4 and 5 (the yellow tie line in Figure \ref{fig: ten_IBR_MG}) after $t_1=0.4$s.}

\textcolor{black}{Figure \ref{fig: FiveIBR_motivation_example_IBR1} presents the three-phase currents $\{i_{\text{a}{1}}, i_{\text{b}{1}}, i_{\text{c}{1}}\}$ at IBR 1 during the event. It shows that networking MGs 1 and 2 incurs instability: The amplitudes of the currents keep oscillating with growing magnitudes after $t=0.4$s. Note that such instability occurs throughout the microgrid, as the similar phenomenon can be observed at each IBR. The figures presenting the rest IBRs' behaviors are omitted for brevity.}

With the SEI at IBR 5, Figure \ref{fig: FiveIBR_corr_IBR1} shows that the amplitudes of the three-phase currents at IBR 1 are stabilized by the SEI after the disturbance. \textcolor{black}{As shown in Appendix \ref{app:figures_10_IBRs}, the terminal current and voltage amplitudes, and power outputs of the remaining IBRs can be stabilized,} suggesting that the SEI can stabilize the microgrid without reprogramming the IBRs. To show that the 5 IBRs in MG1 and MG2 can withstand disturbances due to topology changes, we add two additional disturbances. One disturbance is created by closing the tie line between IBRs 1 and 5 (the red tie line in Figure \ref{fig: ten_IBR_MG}) at $t_2=1$s, which changes the chain network architecture to a meshed one. Another disturbance is caused by closing the tie line between IBR 1 and IBR 6 (the green tie line in Figure \ref{fig: ten_IBR_MG}) at $t_3=1.5$s, which interconnect the ten IBRs. \textcolor{black}{Figures \ref{fig: FiveIBR_corr_IBR1}, \ref{fig: TenIBR_iabc}, \ref{fig: TenIBR_vabc}, and \ref{fig: TenIBR_P} suggest} the 10-IBR system can withstand the two new disturbances.

\begin{figure}
				\centering
				\subfloat[]{\includegraphics[width=0.7\linewidth]{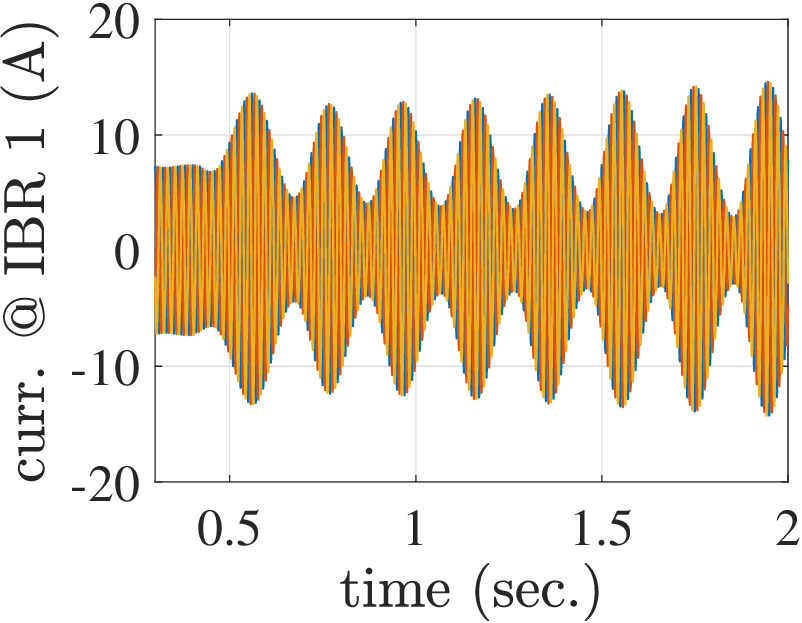}}
				\hfil
				\subfloat[]{\includegraphics[width=0.7\linewidth]{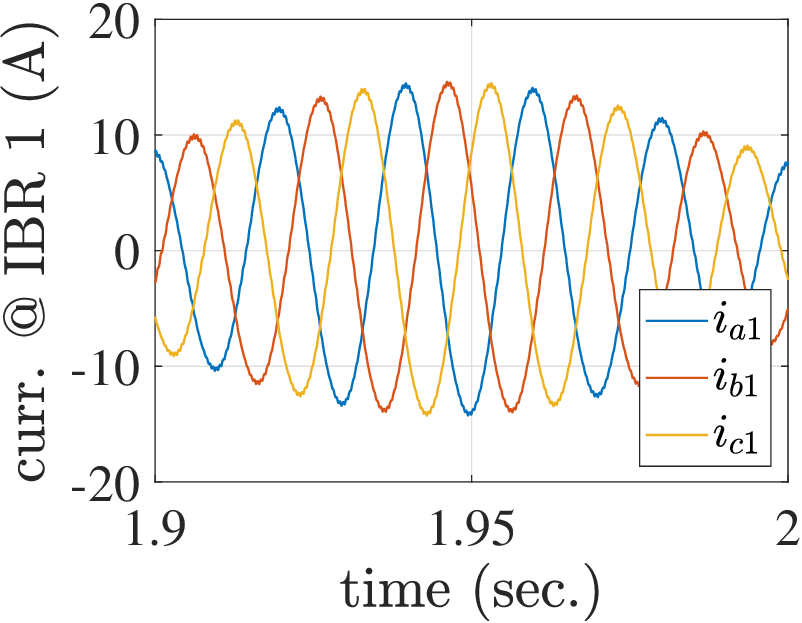}}
				\hfil
				\caption{\textcolor{black}{(a) Three-phase instantaneous terminal currents $\{i_{\text{a}{1}}, i_{\text{b}{1}}, i_{\text{c}{1}}\}$ at IBR 1 when microgrids 1 and 2 are networked; and (b) their zoom-in version from $1.9$s to $2$s without SEI.}}
				\label{fig: FiveIBR_motivation_example_IBR1}
			\end{figure}

\begin{figure}
				\centering
				\subfloat[]{\includegraphics[width=0.5\linewidth]{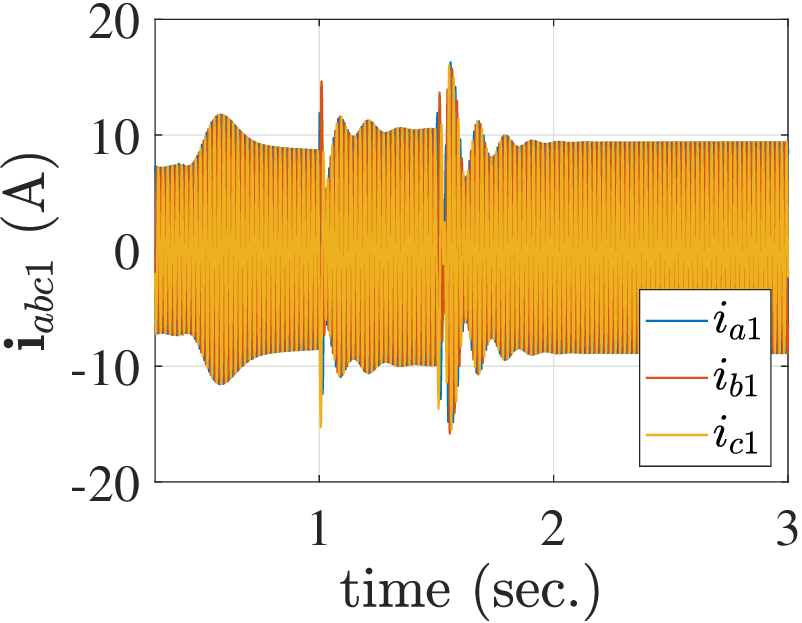}}
				\hfil
				\subfloat[]{\includegraphics[width=0.5\linewidth]{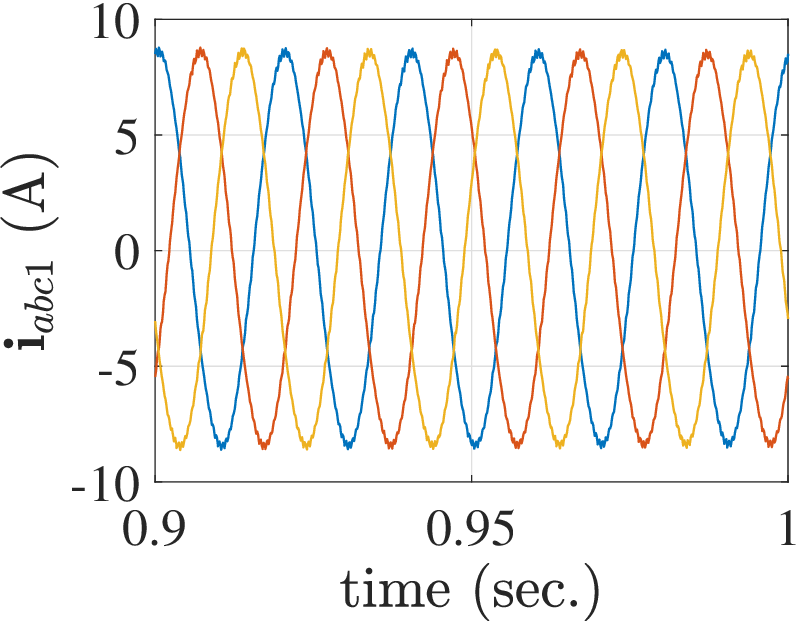}}
				\hfil
                \subfloat[]{\includegraphics[width=0.5\linewidth]{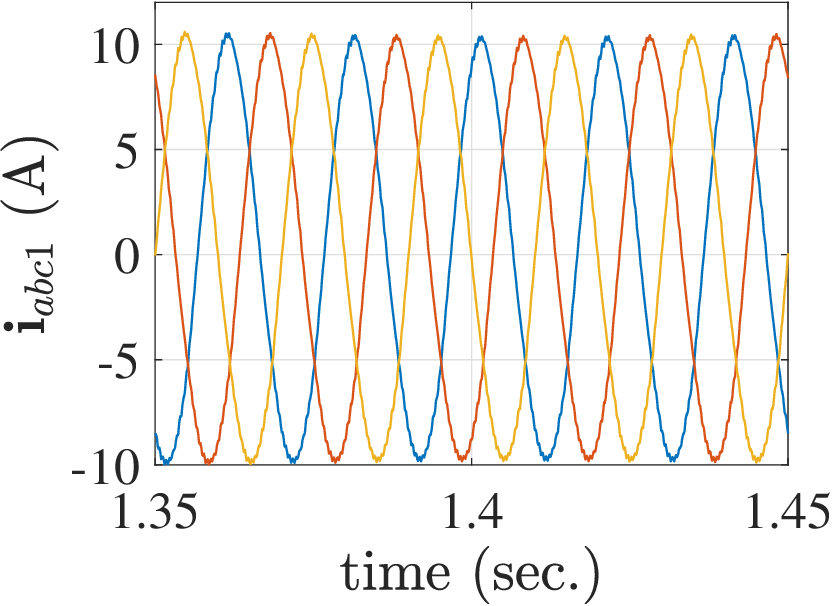}}
				\hfil
                \subfloat[]{\includegraphics[width=0.5\linewidth]{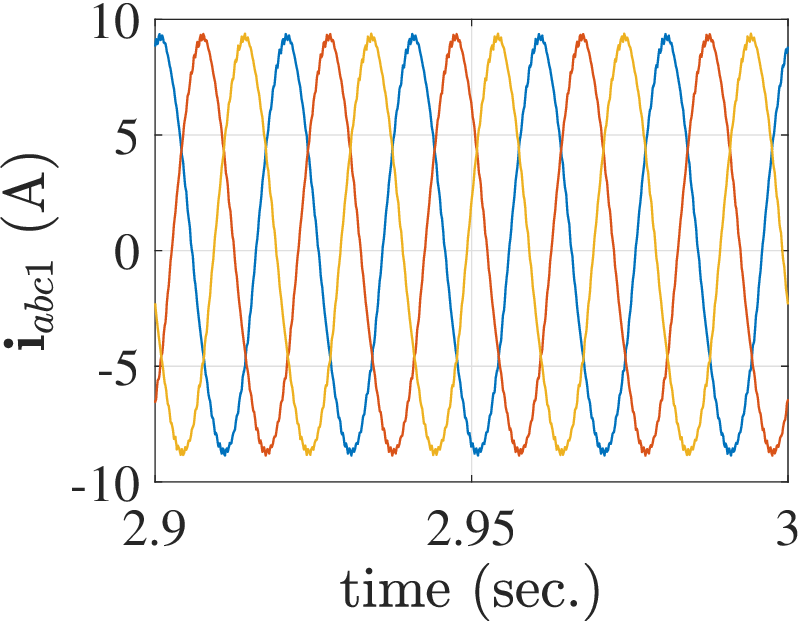}}
				\hfil
				\caption{\textcolor{black}{(a) Three-phase instantaneous terminal currents $\{i_{\text{a}{1}}, i_{\text{b}{1}}, i_{\text{c}{1}}\}$ at IBR 1 in the 10-IBR microgrid under the disturbances at $0.4$s, $1$s, and $1.5$s with the SEI; and their zoom-in version around $1$s (b), $1.4$s (c) and $3$s (d).}}
				\label{fig: FiveIBR_corr_IBR1}
			\end{figure}

\section{Conclusion} \label{sec:concl}

This paper introduces an interface that stabilizes AC microgrids dominated by IBRs in a non-intrusive, decentralized manner.
The interface has three advantages: 1) it enforces the system-level, EMT stability without the need of reprogramming IBR controllers, accessing internal state variables of IBRs, or communication among IBRs; 2) designing the interface only requires a scalar, and sharing the scalar with the NMPs does not cause any concerns on IP privacy; and 3) the interface explicitly addresses the complexity from IBR's high-order dynamics. 
The method is tested by the 2-IBR and \textcolor{black}{10-IBR} microgrids with benchmark parameters. Simulations show that growing oscillations can occur, when two stable AC microgrids are networked, and they also suggest that the interface can mitigate such a system-level symptom. Future work will theoretically quantify the minimal power required by SEIs for stabilization, develop the internal cyber and physical architectures of voltage and current sources for the SEIs, \textcolor{black}{and perform hardware-in-the-loop and purely hardware validations.}
%Another research direction is to investigate the power-electronics implementation of the SEIs.

\bibliographystyle{cas-model2-names}

\bibliography{ref.bib}

@ARTICLE{TimGreenModel,
  author={Pogaku, Nagaraju and Prodanovic, Milan and Green, Timothy C.},
  journal={IEEE Transactions on Power Electronics}, 
  title={Modeling, Analysis and Testing of Autonomous Operation of an Inverter-Based Microgrid}, 
  year={2007},
  volume={22},
  number={2},
  pages={613-625},
  doi={10.1109/TPEL.2006.890003}}

@ARTICLE{THneuralTSG,
  author={Huang, Tong and Gao, Sicun and Xie, Le},
  journal={IEEE Transactions on Smart Grid}, 
  title={A Neural Lyapunov Approach to Transient Stability Assessment of Power Electronics-Interfaced Networked Microgrids}, 
  year={2022},
  volume={13},
  number={1},
  pages={106-118},
  doi={10.1109/TSG.2021.3117889}}

@INPROCEEDINGS{ERCOT_SSO,
  title={Analysis of subsynchronous control interactions in DFIG-based wind farms: ERCOT case study},
  author={Mohammadpour, Hossein Ali and Santi, Enrico},
  booktitle={2015 IEEE Energy Conversion Congress and Exposition (ECCE)},
  pages={500--505},
  year={2015},
  organization={IEEE}
}

@article{yin2019review,
  title={Review of oscillations in VSC-HVDC systems caused by control interactions},
  author={Yin, Congqi and Xie, Xiaorong and Xu, Shukai and Zou, Changyue},
  journal={The Journal of Engineering},
  volume={2019},
  number={16},
  pages={1204--1207},
  year={2019},
  publisher={Wiley Online Library}
}

@InProceedings{huang2021neural,
  author    = {Huang, Tong and Gao, Sicun and others},
  title     = {A Neural Lyapunov Approach to Transient Stability Assessment in Interconnected Microgrids},
  booktitle = {HICSS},
  year      = {2021},
}

@Article{6880421,
  author  = {Shamsi, Pourya and Fahimi, Babak},
  title   = {Stability Assessment of a DC Distribution Network in a Hybrid Micro-Grid Application},
  journal = {IEEE Transactions on Smart Grid},
  year    = {2014},
  volume  = {5},
  number  = {5},
  pages   = {2527-2534},
  doi     = {10.1109/TSG.2014.2302804},
}

@Article{481632,
    author={Hsiao-Dong Chang and Chia-Chi Chu and Cauley, G.},
  journal={Proceedings of the IEEE}, 
  title={Direct stability analysis of electric power systems using energy functions: theory, applications, and perspective}, 
  year={1995},
  volume={83},
  number={11},
  pages={1497-1529},
  doi={10.1109/5.481632}}

@Article{FOUAD1988233,
   title={The transient energy function method},
  author={Fouad, A Ahmed and Vittal, V},
  journal={International Journal of Electrical Power \& Energy Systems},
  volume={10},
  number={4},
  pages={233--246},
  year={1988},
  publisher={Elsevier}
}

@Article{Kabalan2017,
  author  = {Kabalan, Mahmoud and Singh, Pritpal and Niebur, Dagmar},
  title   = {Large Signal Lyapunov-Based Stability Studies in Microgrids: A Review},
  journal = {IEEE Transactions on Smart Grid},
  year    = {2017},
  volume  = {8},
  number  = {5},
  pages   = {2287-2295},
  doi     = {10.1109/TSG.2016.2521652},
}

@ARTICLE{huangilicTSG,
  author={Huang, Tong and Wu, Dan and Ilić, Marija},
  journal={IEEE Transactions on Smart Grid}, 
  title={Cyber-Resilient Automatic Generation Control for Systems of AC Microgrids}, 
  year={2024},
  volume={15},
  number={1},
  pages={886-898},
  doi={10.1109/TSG.2023.3272632}}

@ARTICLE{1338120,
   author={Morison, K. and Lei Wang and Kundur, P.},
  journal={IEEE Power and Energy Magazine}, 
  title={Power system security assessment}, 
  year={2004},
  volume={2},
  number={5},
  pages={30-39},
  doi={10.1109/MPAE.2004.1338120}}

@ARTICLE{IIT_paper,
  author={Dey, Kaustav and Kulkarni, A. M.},
  journal={IEEE Transactions on Power Systems}, 
  title={Passivity-Based Decentralized Criteria for Small-Signal Stability of Power Systems With Converter-Interfaced Generation}, 
  year={2023},
  volume={38},
  number={3},
  pages={2820-2833},
  doi={10.1109/TPWRS.2022.3192302}}

@ARTICLE{8536454,
  author={Alam, Mahamad Nabab and Chakrabarti, Saikat and Ghosh, Arindam},
  journal={IEEE Transactions on Industrial Informatics}, 
  title={Networked Microgrids: State-of-the-Art and Future Perspectives}, 
  year={2019},
  volume={15},
  number={3},
  pages={1238-1250},
  doi={10.1109/TII.2018.2881540}}

@INPROCEEDINGS{985003,
  author={Lasseter, R.H.},
  booktitle={2002 IEEE Power Engineering Society Winter Meeting. Conference Proceedings}, 
  title={MicroGrids}, 
  year={2002},}

@ARTICLE{self_displ,
   author={Gu, Yunjie and Li, Wuhua and He, Xiangning},
  journal={IEEE Transactions on Power Systems}, 
  title={Passivity-Based Control of DC Microgrid for Self-Disciplined Stabilization}, 
  year={2015},
  volume={30},
  number={5},
  pages={2623-2632},
  doi={10.1109/TPWRS.2014.2360300}}

@ARTICLE{tsinghua_condition,
  author={Yang, Peng and Liu, Feng and Wang, Zhaojian and Shen, Chen},
  journal={IEEE Transactions on Power Systems}, 
  title={Distributed Stability Conditions for Power Systems With Heterogeneous Nonlinear Bus Dynamics}, 
  year={2020},
  volume={35},
  number={3},
  pages={2313-2324},
  doi={10.1109/TPWRS.2019.2951202}}

@INPROCEEDINGS{amit_dist,
  author={Jena, Amit and others},
  booktitle={IEEE Conference on Decision and Control}, 
  title={Distributed Learning-based Stability Assessment for Large Scale Networks of Dissipative Systems}, 
  year={2021},
  }

@ARTICLE{Panos_thm,
  author={Xia, Meng and Rahnama, Arash and Wang, Shige and Antsaklis, Panos J.},
  journal={IEEE Transactions on Automatic Control}, 
  title={Control Design Using Passivation for Stability and Performance}, 
  year={2018},
  volume={63},
  number={9},
  pages={2987-2993},
  doi={10.1109/TAC.2018.2789681}}

@book{alexander2013fundamentals,
  title={Fundamentals of electric circuits},
  author={Alexander, Charles K},
  year={2013},
}

@ARTICLE{UCF_paper,
  author={Xu, Ying and Qu, Zhihua and Harvey, Roland and Namerikawa, Toru},
  journal={IEEE Transactions on Power Systems}, 
  title={Data-Driven Wide-Area Control Design of Power System Using the Passivity Shortage Framework}, 
  year={2021},
  volume={36},
  number={2},
  pages={830-841},
  doi={10.1109/TPWRS.2020.3009630}}

@article{hill1977stability,
  title={Stability results for nonlinear feedback systems},
  author={Hill, David J and Moylan, Peter James},
  journal={Automatica},
  volume={13},
  number={4},
  pages={377--382},
  year={1977},
  publisher={Elsevier}
}

@book{nonlinear_ctr,
  title={Nonlinear Control},
  author={Hassan K. Khalil},
  year={2015},
  publisher={Pearson}
}

@ARTICLE{MG_def,
   author={Farrokhabadi, Mostafa and Cañizares, Claudio A. and Simpson-Porco, John W. and Nasr, Ehsan and Fan, Lingling and Mendoza-Araya, Patricio A. and Tonkoski, Reinaldo and Tamrakar, Ujjwol and Hatziargyriou, Nikos and Lagos, Dimitris and Wies, Richard W. and Paolone, Mario and Liserre, Marco and Meegahapola, Lasantha and Kabalan, Mahmoud and Hajimiragha, Amir H. and Peralta, Dario and Elizondo, Marcelo A. and Schneider, Kevin P. and Tuffner, Francis K. and Reilly, Jim},
  journal={IEEE Transactions on Power Systems}, 
  title={Microgrid Stability Definitions, Analysis, and Examples}, 
  year={2020},
  volume={35},
  number={1},
  pages={13-29},
  doi={10.1109/TPWRS.2019.2925703}}

@Article{dq0trans,
AUTHOR = {Levron, Yoash and others},
TITLE = {A Tutorial on Dynamics and Control of Power Systems with Distributed and Renewable Energy Sources Based on the DQ0 Transformation},
JOURNAL = {Applied Sciences},
YEAR = {2018},
}

@article{ilic2000dynamics,
  title={Dynamics and control of large electric power systems},
  author={Ilic, Marija D and Zaborszky, John},
  year={2000}
}

@article{XIE2022100640,
  title={Energy system digitization in the era of AI: A three-layered approach toward carbon neutrality},
  author={Xie, Le and Huang, Tong and Zheng, Xiangtian and Liu, Yan and Wang, Mengdi and Vittal, Vijay and Kumar, PR and Shakkottai, Srinivas and Cui, Yi},
  journal={Patterns},
  volume={3},
  number={12},
  year={2022},
  publisher={Elsevier}
}

@ARTICLE{active_damper,
  author={Li, Yang and Wu, Xiangyang and Shuai, Zhikang and Zhou, Quan and Chen, Haojie and Shen, Zheng John},
  journal={IEEE Transactions on Power Electronics}, 
  title={A Systematic Stability Enhancement Method for Microgrids With Unknown-Parameter Inverters}, 
  year={2023},
  volume={38},
  number={3},
  pages={3029-3043},
  doi={10.1109/TPEL.2022.3219400}}

@ARTICLE{apparent_impedance_para_tuning,
   author={Göthner, Fredrik and Torres-Olguin, Raymundo E. and Roldán-Pérez, Javier and Rygg, Atle and Midtgård, Ole-Morten},
  journal={IEEE Transactions on Power Electronics}, 
  title={Apparent Impedance-Based Adaptive Controller for Improved Stability of a Droop-Controlled Microgrid}, 
  year={2021},
  volume={36},
  number={8},
  pages={9465-9476},
  doi={10.1109/TPEL.2021.3050615}}

@ARTICLE{bikash_assessment,
  author={Cifuentes, Nicolás and Sun, Mingyu and Gupta, Robin and Pal, Bikash C.},
  journal={IEEE Transactions on Power Systems}, 
  title={Black-Box Impedance-Based Stability Assessment of Dynamic Interactions Between Converters and Grid}, 
  year={2022},
  volume={37},
  number={4},
  pages={2976-2987},
  doi={10.1109/TPWRS.2021.3128812}}

@ARTICLE{XWang_22,
  author={Liao, Yicheng and Wang, Xiongfei and Wang, Xiao},
  journal={IEEE Transactions on Power Electronics}, 
  title={Frequency-Domain Participation Analysis for Electronic Power Systems}, 
  year={2022},
  volume={37},
  number={3},
  pages={2531-2537},
  doi={10.1109/TPEL.2021.3118439}}

@ARTICLE{FLEE_02,
  author={Xiaogang Feng and Jinjun Liu and Lee, F.C.},
  journal={IEEE Transactions on Power Electronics}, 
  title={Impedance specifications for stable DC distributed power systems}, 
  year={2002},
  volume={17},
  number={2},
  pages={157-162},
  doi={10.1109/63.988825}}

@ARTICLE{Sun_11,
  author={Sun, Jian},
  journal={IEEE Transactions on Power Electronics}, 
  title={Impedance-Based Stability Criterion for Grid-Connected Inverters}, 
  year={2011},
  doi={10.1109/TPEL.2011.2136439}}

@ARTICLE{XWang_14,
   author={Wang, Xiongfei and Blaabjerg, Frede and Wu, Weimin},
  journal={IEEE Transactions on Power Electronics}, 
  title={Modeling and Analysis of Harmonic Stability in an AC Power-Electronics-Based Power System}, 
  year={2014},
  volume={29},
  number={12},
  pages={6421-6432},
  doi={10.1109/TPEL.2014.2306432}}

@ARTICLE{Control_para_17,
   author={Cao, Wenchao and Ma, Yiwei and Wang, Fred},
  journal={IEEE Transactions on Power Electronics}, 
  title={Sequence-Impedance-Based Harmonic Stability Analysis and Controller Parameter Design of Three-Phase Inverter-Based Multibus AC Power Systems}, 
  year={2017},
  volume={32},
  number={10},
  pages={7674-7693},
  doi={10.1109/TPEL.2016.2637883}}

@ARTICLE{Sun_09,
  author={Sun, Jian},
  journal={IEEE Transactions on Power Electronics}, 
  title={Small-Signal Methods for AC Distributed Power Systems–A Review}, 
  year={2009},
  volume={24},
  number={11},
  pages={2545-2554},
  doi={10.1109/TPEL.2009.2029859}}

@book{sauer2017power,
  title={Power system dynamics and stability: with synchrophasor measurement and power system toolbox},
  author={Sauer, Peter W and Pai, Mangalore A and Chow, Joe H},
  year={2017},
  publisher={John Wiley \& Sons}
}

@ARTICLE{720455,
   author={Mertens, A. and Eckardt, D.},
  journal={IEEE Transactions on Industry Applications}, 
  title={Voltage and current sensing in power electronic converters using sigma-delta A/D conversion}, 
  year={1998},
  volume={34},
  number={5},
  pages={1139-1146},
  doi={10.1109/28.720455}}

@article{bruinsma1990fast,
  title={A fast algorithm to compute the H-Infinity norm of a transfer function matrix},
  author={Bruinsma, NA and Steinbuch, M},
  journal={Systems \& Control Letters},
  volume={14},
  number={4},
  pages={287--293},
  year={1990},
  publisher={Elsevier}
}

\appendix
%\input{Section/App_GFM}
%\noindent \textcolor{blue}{\emph{This document is written by Dr. Tong Huang, and it serves as supplementing materials of the manuscript titled ``A Non-intrusive Decentralized Approach to Stabilizing IBR-dominated AC Microgrids'' submitted to IEEE Transactions on Smart Grid, in order to meet the page limit requirement of the journal. The index system of sections, references cited, equations, and figures in this document is consistent with the one used in the manuscript. For example, ``Figure 1'' in this document refers to ``Figure 1'' of the submitted manuscript.}}

\section{Acronym Explanation} \label{app:aronyms}
\textcolor{black}{For the convenience of readers, all acronyms are explained in Table \ref{tab:acronym}.}

\begin{table}
    \caption{\textcolor{black}{Acronym Explanation}}
    \centering
    {\color{black}
        \begin{tabular}{l|l}
            \hline\hline
            d-q frame & local direct-quadrature frame\\
            D-Q frame & common direct-quadrature frame\\
            DER & distributed energy resource \\
            DSO & distribution system operator\\
            EMT & electromagnetic transient\\
            GFL & grid-following\\
            GFM & grid-forming\\
            IBR & inverter-based resources\\
            IP & intellectual property\\
            KCL & Kirchhoff's current law\\
            MG & microgrid\\
            MGO & microgrid operator\\
            NMP & non-manufacturer party\\
            OFP & output-feedback passive\\
            PE & power-electronic\\
            PLL & phase-locked loop\\
            RLC & resistor-inductor-capacitor\\
            SEI & stability enforcement interface\\
            SISO & single-input-single-output\\
            SO & system operator\\
            SSCI & sub-synchronous control interactions\\
            \hline\hline
        \end{tabular}
    }
    \label{tab:acronym}
\end{table}

\section{\textcolor{black}{Dynamics of Grid-forming IBRs}} \label{app: GFM_IBR}
Suppose that the $n$-th IBR is grid-forming. As shown in Figure \ref{fig:IBRn}, the GFM IBR includes a DC voltage source, an inverter, a resistor-inductor-capacitor (RLC) low-pass filter, a power controller, a voltage controller, and a current controller. The dynamics of each block in Figure \ref{fig:IBRn} is introduced as follows.
\subsection{RLC filter}
The inverter connects to the rest of the microgrid via an RLC filter whose dynamics are \cite{TimGreenModel}
\begin{subequations} \label{eq:LC_filter}
  \begin{align}
    L_{\text{f}n}\dot{i}_{\text{ld}n} &= -r_{\text{f}n}i_{\text{ld}n} + L_{\text{f}n}\omega_0  i_{\text{lq}n} + v_{\text{id}n}-v_{\text{od}n}\\
    L_{\text{f}n}\dot{i}_{\text{lq}n} &= -r_{\text{f}n}i_{\text{lq}n} - L_{\text{f}n}\omega_0 i_{\text{ld}n} + v_{\text{iq}n}-v_{\text{oq}n}\\
    C_{\text{f}n}\dot{v}_{\text{od}n} &= C_{\text{f}n}\omega_0 v_{\text{oq}n} + i_{\text{ld}n}+i_{\text{od}n}\\
    C_{\text{f}n}\dot{v}_{\text{oq}n} &= -C_{\text{f}n}\omega_0 v_{\text{od}n} + i_{\text{lq}n}+i_{\text{oq}n}
  \end{align}
\end{subequations}
where $i_{\text{ld}n}$ and $i_{\text{od}n}$ ($i_{\text{lq}n}$ and $i_{\text{oq}n}$) are the direct (quadrature) component of the current $\mathbf{i}_{\text{l}n}$ and $\mathbf{i}_{\text{o}n}$ annotated in Figure \ref{fig:IBRn}; $v_{\text{id}n}$ and $v_{\text{od}n}$ ($v_{\text{iq}n}$ and $v_{\text{oq}n}$) are the direct (quadrature) components of the voltage $\mathbf{v}_{\text{i}n}$ and $\mathbf{v}_{\text{o}n}$; resistance $r_{\text{f}n}$, inductance $L_{\text{f}n}$, and capacitance $C_{\text{f}n}$ of the RLC circuit are labeled in Figure \ref{fig:IBRn}; and $\omega_{0}$ is the nominal frequency (i.e., 377 or 314 rad/s). Note that the reference positive direction of $\mathbf{i}_{\text{o}n}$ is \emph{pointing into} the IBR.

%Note the reference direction of $i_{\text{od}n}$ and $i_{\text{oq}n}$ is into the $n$-th IBR.
\subsection{Power controller}
A power controller contains a power calculator, a power filter, and a droop controller. The power calculator computes the instantaneous real power $\tilde{p}_n$ and reactive power $\tilde{q}_n$ \emph{injecting into} the rest of the microgrid, based on IBR $n$'s terminal voltages ($v_{\text{od}n}$ and $v_{\text{oq}n}$) and current ($i_{\text{od}n}$ and $i_{\text{oq}n}$) in the direct-quadrature (d-q) reference frame of IBR $n$. With the positive reference directions assigned to $\mathbf{v}_{\text{o}n}$ and $\mathbf{i}_{\text{o}n}$ in Figure \ref{fig:IBRn}, $\tilde{p}_n$ and $\tilde{q}_n$ are computed by %\cite{TimGreenModel}
\begin{subequations} \label{eq:inst_pq}
  \begin{align}
        \tilde{p}_n =-\frac{3}{2}(v_{\text{od}n}i_{\text{od}n} + v_{\text{oq}n}i_{\text{oq}n}) \label{eq:inst_p}\\
       \tilde{q}_n = -\frac{3}{2}(v_{\text{oq}n}i_{\text{od}n} - v_{\text{od}n}i_{\text{oq}n}) \label{eq:inst_q}.
  \end{align}
\end{subequations}
%The first negative sign in \eqref{eq:inst_p} as we
The instantaneous real and reactive power feed the power filter, i.e., a digital low-pass filter, whose dynamics is described by
\begin{subequations} \label{eq:power_sensor}
  \begin{align}
        \dot{P}_n &= -\omega_{\text{c}n}P_n + \omega_{\text{c}n} \tilde{p}_n\\
        \dot{Q}_n &= -\omega_{\text{c}n}Q_n + \omega_{\text{c}n} \tilde{q}_n
  \end{align}
\end{subequations}
where $\omega_{\text{c}n}$ is the cut-off frequency; and $P_n$ and $Q_n$ are the real and reactive power filtered by the power filter. The droop controller takes $P_n$ and $Q_n$ as inputs and it specifies frequency $\omega_n$, phase angle $\delta_n$ and voltage setpoints $v_{\text{od}n}^*$ and $v_{\text{oq}n}^*$ via
\begin{subequations} \label{eq:droop}
  \begin{align}
        & \dot{\delta}_n = \omega_n - \omega_{0}, \quad
        \omega_n = \omega_{\text{s}n} - \alpha_n P_n \\
        & v_{\text{od}n}^* = V_{\text{0}n} - \beta_n Q_n, \quad
         v_{\text{oq}n}^* = 0   
  \end{align}
\end{subequations}
where $\omega_{\text{s}n}$ is set by a secondary controller; $V_{\text{0}n}$ is a voltage setpoint; and $\alpha_n$ and $\beta_n$ are droop control parameters. The angle $\delta_n$ is used in the Park and the inverse Park transformations that bridge three-phase variables with variables in the d-q-0 frame.

\subsection{Voltage controller}
The dynamics of the voltage controller is governed by
\begin{subequations} \label{eq:v_controller}
  \begin{align}
    &\dot{\phi}_{\text{d}n} = - v_{\text{od}n} + v_{\text{od}n}^* , \quad \dot{\phi}_{\text{q}n} = - v_{\text{oq}n} + v_{\text{oq}n}^* ,\\
    &i_{\text{ld}n}^* = K_{\text{pv}n}(v_{\text{od}n}^* - v_{\text{od}n})  - F_{n}i_{\text{od}n} - \omega_{\text{0}}C_{\text{f}n}v_{\text{oq}n} +K_{\text{iv}n} \phi_{\text{d}n}\\
    &i_{\text{lq}n}^* = K_{\text{pv}n}(v_{\text{oq}n}^* - v_{\text{oq}n}) - F_{n}i_{\text{oq}n} + \omega_{\text{0}}C_{\text{f}n}v_{\text{od}n} + K_{\text{iv}n} \phi_{\text{q}n}
  \end{align}
\end{subequations}
where ${\phi}_{\text{d}n}$ and ${\phi}_{\text{q}n}$ are state variables for the voltage controller; $i_{\text{ld}n}^*$ and $i_{\text{lq}n}^*$ are setpoints of the current controller provided by the voltage controller; and $K_{\text{pv}n}$, $F_n$, and $K_{\text{iv}n}$ are control parameters.

\subsection{Current controller} The dynamics of the current controller is described by
\begin{subequations} \label{eq:I_controller}
  \begin{align}
    &\dot{\gamma}_{\text{d}n} = - i_{\text{ld}n} + i_{\text{ld}n}^* , \quad \dot{\gamma}_{\text{q}n} = - i_{\text{lq}n} + i_{\text{lq}n}^*,\\
    &v_{\text{id}n}^* = K_{\text{pc}n}(i_{\text{ld}n}^* - i_{\text{ld}n})  - \omega_{\text{0}}L_{\text{f}n}i_{\text{lq}n} + K_{\text{ic}n} \gamma_{\text{d}n}\\
    &v_{\text{iq}n}^* = K_{\text{pc}n}(i_{\text{lq}n}^* - i_{\text{lq}n})  + \omega_{\text{0}}L_{\text{f}n}i_{\text{ld}n} + K_{\text{ic}n} \gamma_{\text{q}n}
  \end{align}
where ${\gamma}_{\text{d}n}$ and ${\gamma}_{\text{q}n}$ are state variables for the current controller; and $K_{\text{pc}n}$, and $K_{\text{ic}n}$ are control parameters.

\subsection{Time scale separation} \label{sssection: time_scale_separation}
%[\textbf{How to define stabilization time?}] 
The state variables of dynamics \eqref{eq:LC_filter}, \eqref{eq:power_sensor}, \eqref{eq:droop}, \eqref{eq:v_controller}, and \eqref{eq:I_controller} include $\delta_n$, $P_n$, $Q_n$, $\phi_{\text{d}n}$, $\phi_{\text{q}n}$, $\gamma_{\text{d}n}$, $\gamma_{\text{q}n}$, $i_{\text{ld}n}$, $i_{\text{lq}n}$, $v_{\text{od}n}$, and $v_{\text{oq}n}$. Define $\mathcal{S}_n^{\text{s}}=\{\delta_n, P_n, Q_n\}$ and $\mathcal{S}_n^{\text{f}}=\{\phi_{\text{d}n}, \phi_{\text{q}n}, \gamma_{\text{d}n}, \gamma_{\text{q}n}, i_{\text{ld}n}, i_{\text{lq}n}, v_{\text{od}n}, v_{\text{oq}n}\}$. Next we show that the states in $\mathcal{S}_n^{\text{f}}$ can be stabilized much faster than those in $\mathcal{S}_n^{\text{s}}$ via simulating a grid-connected IBR with a representative parameter setting \cite{TimGreenModel}. In the simulation\footnote{Figures \ref{fig: time_scale_separation} and \ref{fig:stabilization_time} are obtained by simulating Microgrid 1 in Figure \ref{fig:two_ibr}-(a). The per-phase impedance at Load 1 changes from $25\Omega$ to $10\Omega$ at $t=0.5$s.}, the load changes at time $t=0.5$s, Figure \ref{fig: time_scale_separation} visualizes state variables $P_1$ and $\phi_{\text{d}1}$. It can be observed that it takes more than $0.15$s to stabilize $P_1$, while $\phi_{\text{d}1}$ is stabilized around $0.006$s after the disturbance occurs. Figure \ref{fig:stabilization_time} presents the stabilization time of key variables of the IBR. Figure \ref{fig:stabilization_time} suggests that $\omega_1$, $P_1$, and $Q_1$ are stabilized much slower than the states in $\mathcal{S}_n^{\text{f}}$. 
%A similar observation is also reported in \cite{lara2023revisiting}.

%As this paper focuses on the transients of the states in $\mathcal{S}_n^{\text{f}}$, the dynamics \eqref{eq:power_sensor} and \eqref{eq:droop} are ignored by assuming $\delta_n$, $P_n$, and $Q_n$ to constants.

\begin{figure}
				\centering
				\subfloat[]{\includegraphics[width=0.6\linewidth]{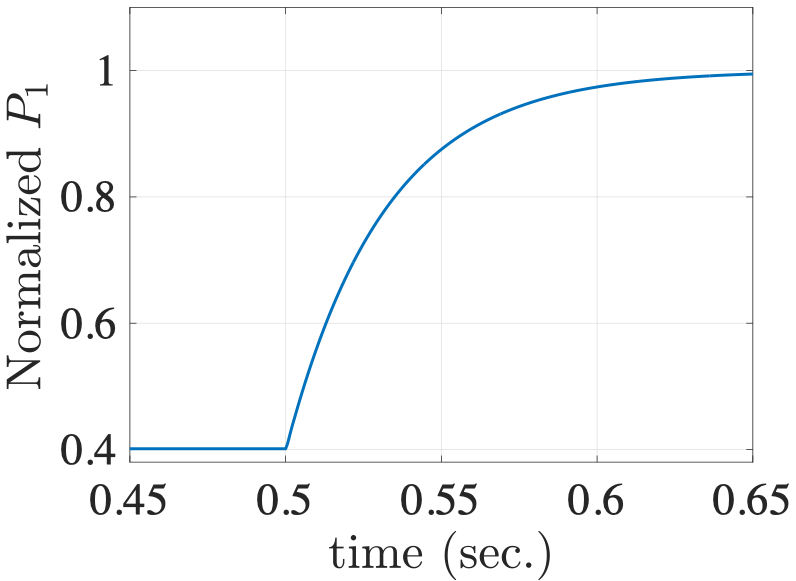}}
				\hfil
				\subfloat[]{\includegraphics[width=0.6\linewidth]{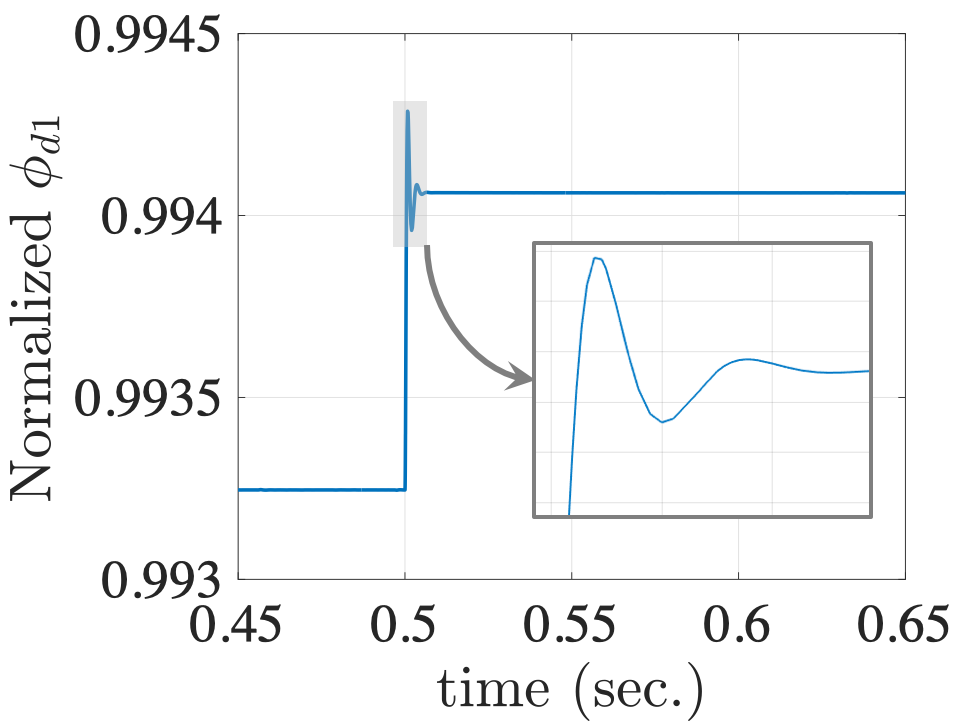}}
				\hfil
				\caption{Time-domain evolution of normalized $P_1$ and $\phi_{\text{d}1}$.}
				\label{fig: time_scale_separation}
			\end{figure}

\begin{figure}
    \centering
    \includegraphics[width = 0.8\linewidth]{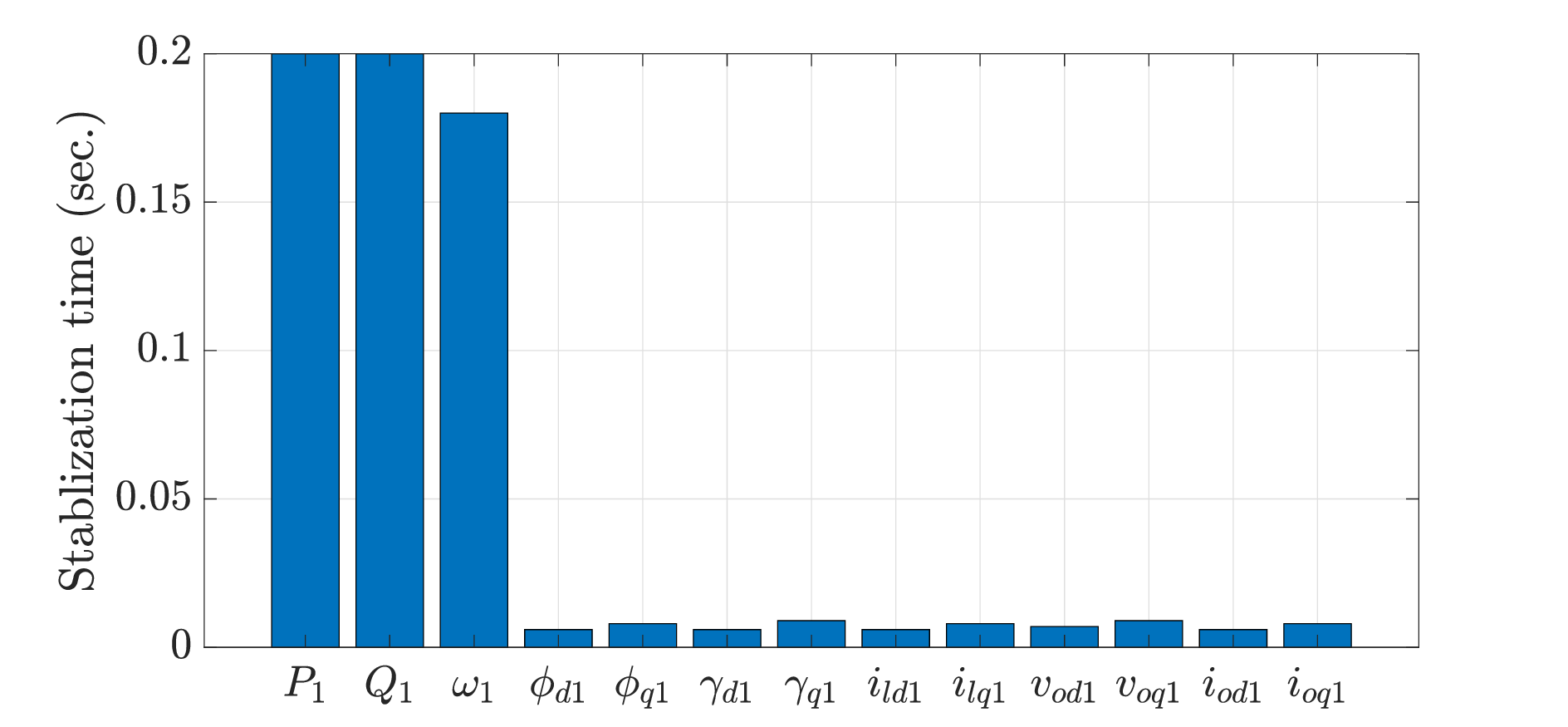}
    \caption{Stabilization time of key variables}
    \label{fig:stabilization_time}
\end{figure}

A very large body of literature (see \cite{THneuralTSG} and the references therein) studies the slow dynamics defined by the states in $\mathcal{S}_n^{\text{s}}$ by assuming that the fast states in $\mathcal{S}_n^{\text{f}}$ are stabilized fast. This paper examines the interaction among the fast states in $\mathcal{S}_n^{\text{f}}$ by assuming the states in $\mathcal{S}_n^{\text{s}}$ as constants. \textcolor{black}{With such an assumption, we exclude the dynamics of the slow states in $\mathcal{S}_n^{\text{s}}$ to derive \eqref{eq:compact_deviation}.}

\end{subequations}

\section{Dynamics of Grid-following IBRs} \label{app: GFL_IBR}
Suppose that the $n-$th IBR is grid-following (GFL). The cyber-physical architecture of the GFL IBR is summarized in Figure \ref{fig:IBRn-GFL}. The dynamics of the RLC output filter and the current controller in Figure \ref{fig:IBRn-GFL}  can be characterized by \eqref{eq:LC_filter} and \eqref{eq:I_controller}. Next, we elaborate the phase locked loop (PLL) and the block that generates the current set points for the current controller.

\subsection{Phase locked loop} The PLL aims to track the frequency of the grid that hosts the GFL IBR. This is done by a proportional-integral (PI) controller described by
\begin{subequations} \label{eq:PLL}
  \begin{align}
        & \dot{\eta}_n = K_{\text{ip}n}v_{\text{oq}n}\\
        & \omega_n = \eta_n + K_{\text{pp}n}v_{\text{oq}n} + \omega_0
  \end{align}
\end{subequations}
\textcolor{black}{where $\eta_n$ is the state variable of the PLL; and $K_{\text{ip}n}$ and $K_{\text{pp}n}$ are control parameters. The integral of $\omega_n$ is used in the Park and inverse Park transformation.} 

\subsection{\textcolor{black}{The block generating current set points}} \textcolor{black}{Given the real and reactive power set points ($P^*_{\text{n}}$ and $Q^*_{\text{n}}$), the current set points ($i_{\text{ld}n}^*$ and $i_{\text{lq}n}^*$) are produced by the following algebraic equations:}
\begin{equation} \label{eq:set_point_gen}
    \textcolor{black}{i_{\text{ld}n}^* = -\frac{2}{3}\frac{P_n^*}{v_{\text{od}n}}, \quad i_{\text{lq}n}^* = \frac{2}{3}\frac{Q_n^*}{v_{\text{od}n}}.}
\end{equation}
\textcolor{black}{Equation \eqref{eq:set_point_gen} is linearized to derive \eqref{eq:compact_deviation} for the GFL IBRs.}

%\textcolor{black}{The simulation parameters in Section \ref{subs:GFL_IBR} for the GFL IBR are as follows: $L_{\text{f}1} = 1.35$mH, $C_{\text{f}1} = 50\mu$F, $r_f = 0.1\Omega$, $K_{\text{ip}1}=2.14$, $K_{\text{pp}1}=0.37$, $P_1^*=2500$W, and $Q_1^*=0$.}
%\input{Section/App_Para}

%\input{Section/Appendix1}


\section{\textcolor{black}{Visualization of Current, Voltage, and Power in the 10-IBR System}} \label{app:figures_10_IBRs}
\textcolor{black}{Figures \ref{fig: TenIBR_iabc}, \ref{fig: TenIBR_vabc}, and \ref{fig: TenIBR_P} visualize three-phase currents, three-phase voltages, and power measured at each IBR in the 10-IBR system described in Section \ref{sec:case_study}.2.}
\begin{figure*}
				\centering
				\subfloat[]{\includegraphics[width=0.25\linewidth]{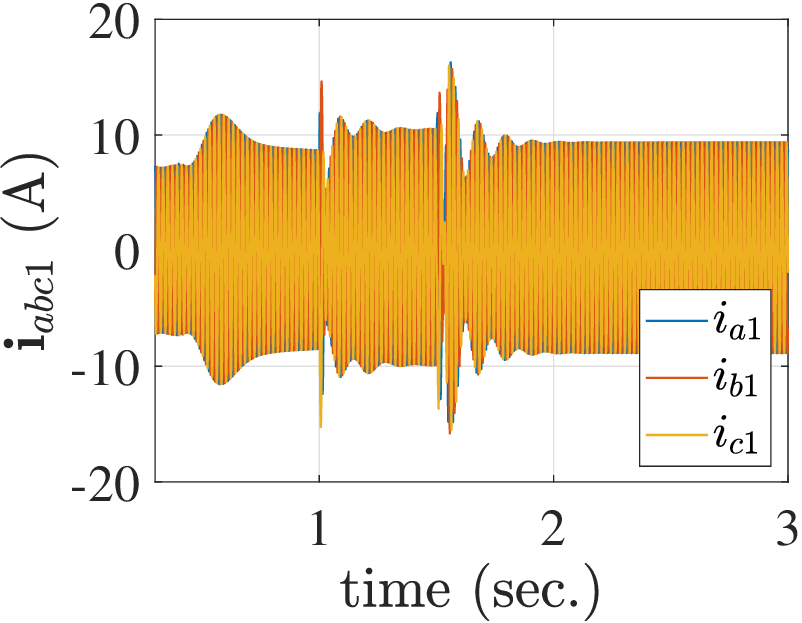}}
				\hfil
				\subfloat[]{\includegraphics[width=0.25\linewidth]{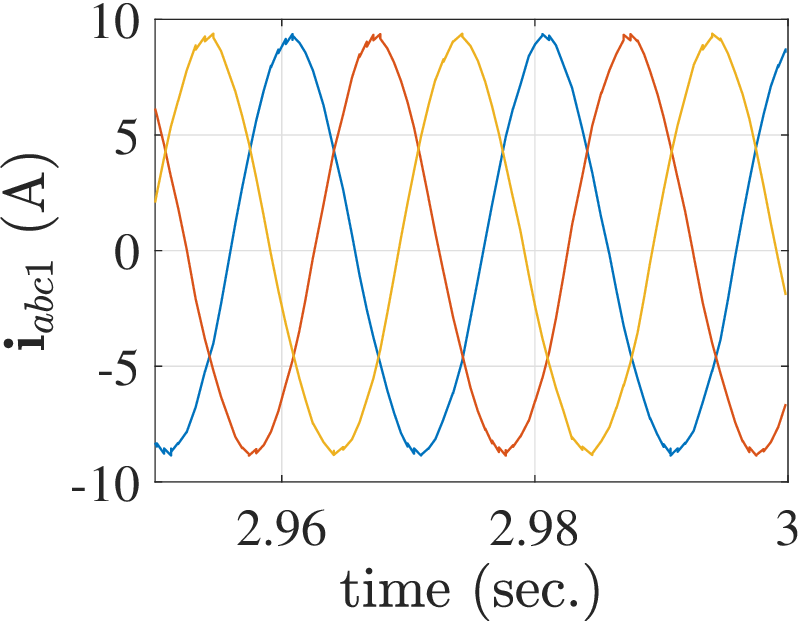}}
				\hfil
                \subfloat[]{\includegraphics[width=0.25\linewidth]{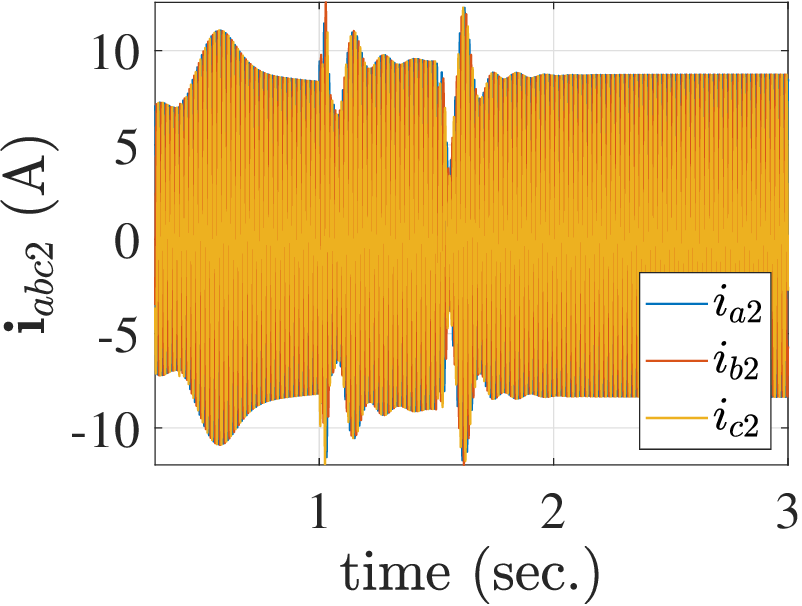}}
				\hfil
				\subfloat[]{\includegraphics[width=0.25\linewidth]{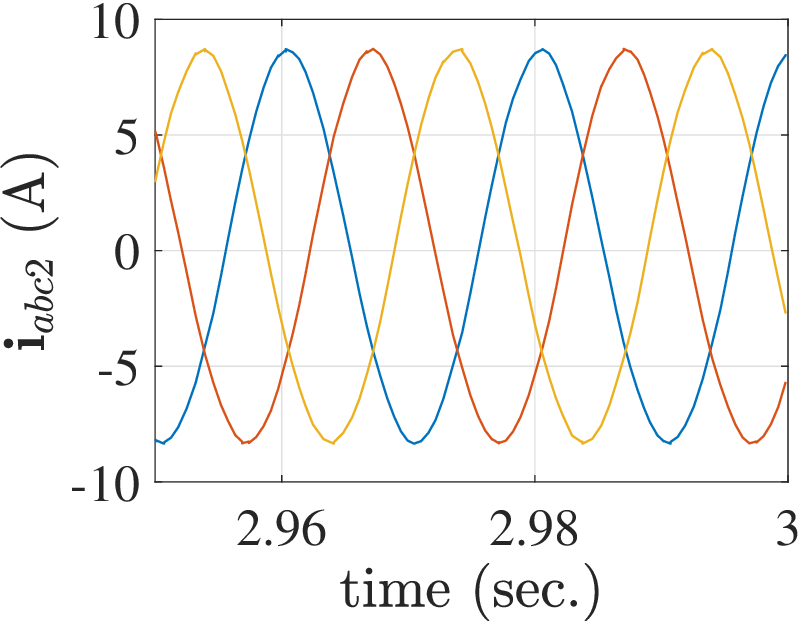}}
				\hfil
                \subfloat[]{\includegraphics[width=0.25\linewidth]{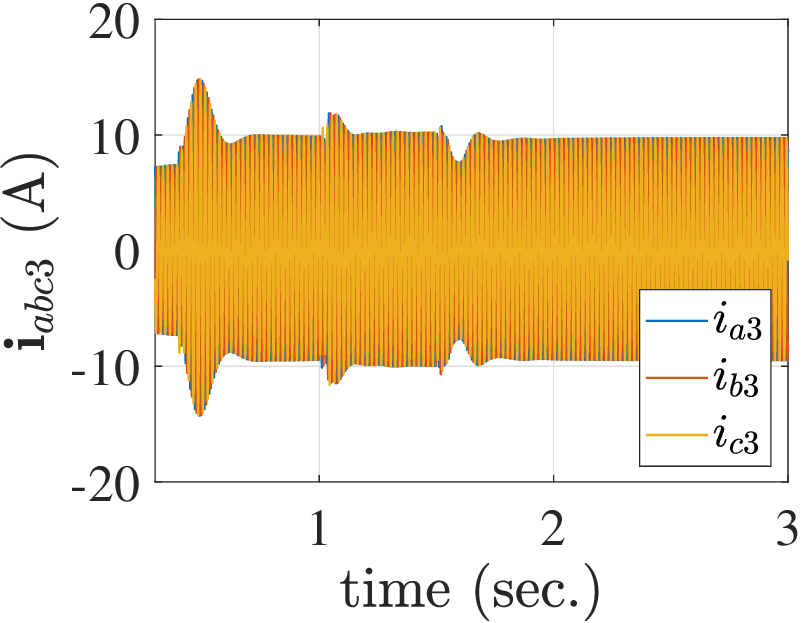}}
				\hfil
				\subfloat[]{\includegraphics[width=0.25\linewidth]{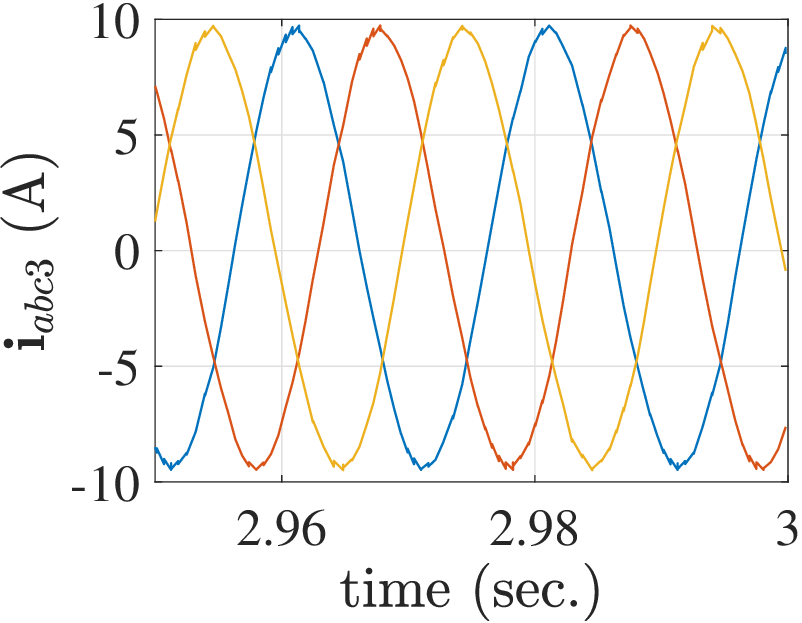}}
				\hfil
                \subfloat[]{\includegraphics[width=0.25\linewidth]{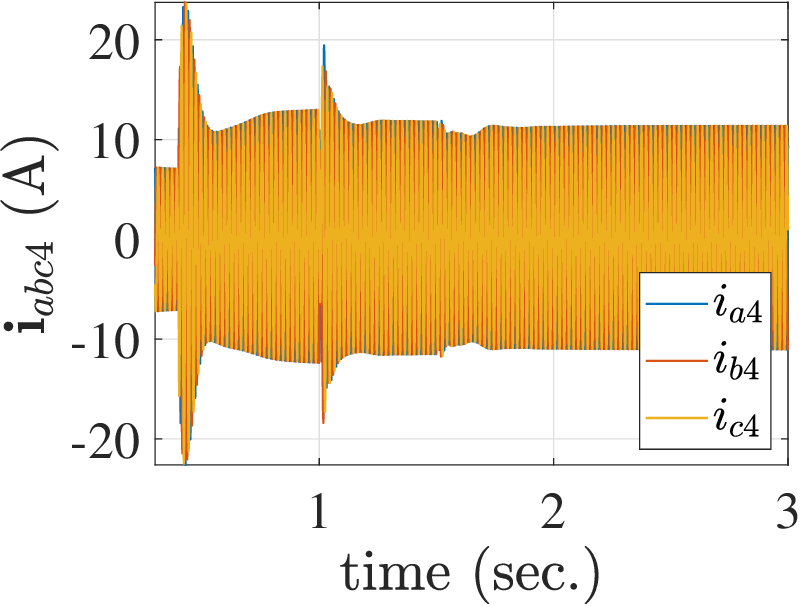}}
				\hfil
				\subfloat[]{\includegraphics[width=0.25\linewidth]{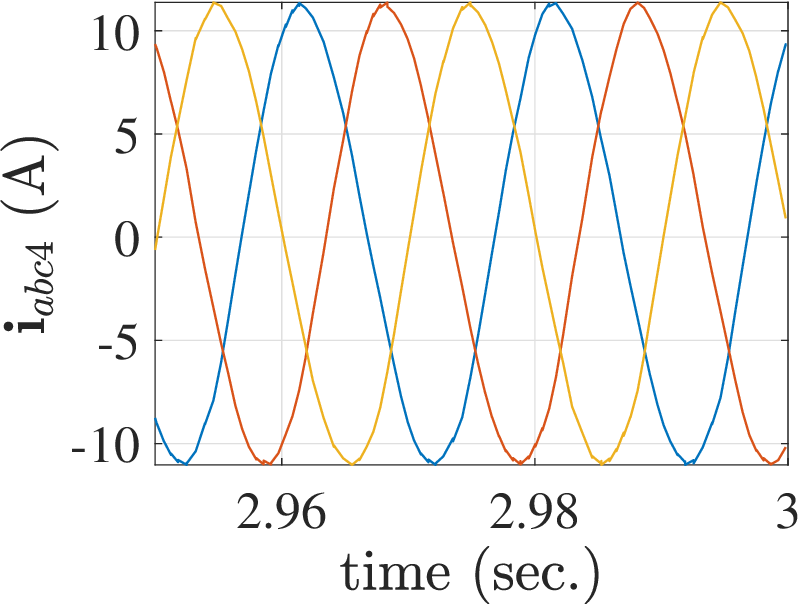}}
				\hfil
                \subfloat[]{\includegraphics[width=0.25\linewidth]{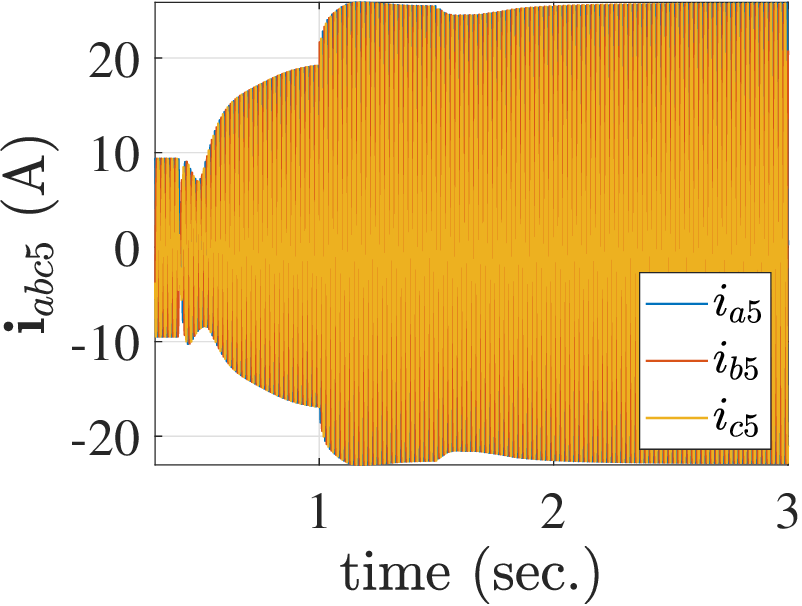}}
				\hfil
				\subfloat[]{\includegraphics[width=0.25\linewidth]{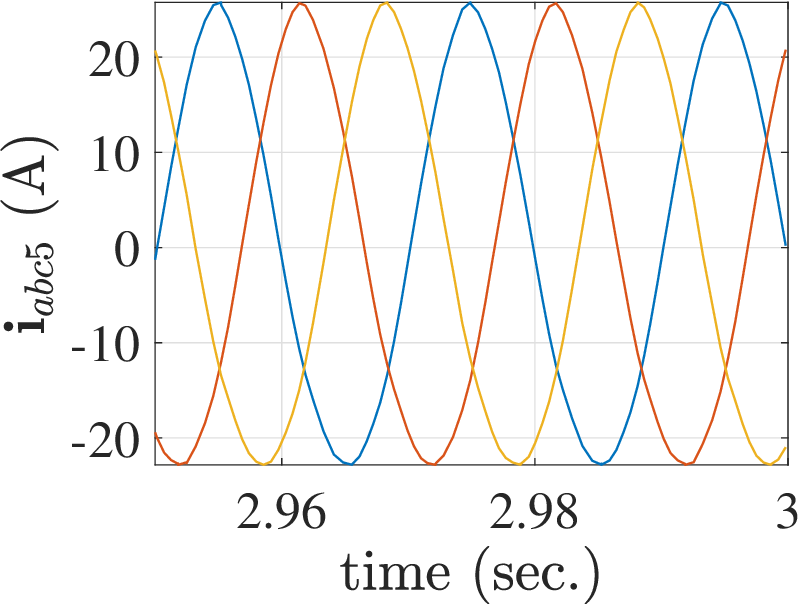}}
				\hfil
                \subfloat[]{\includegraphics[width=0.25\linewidth]{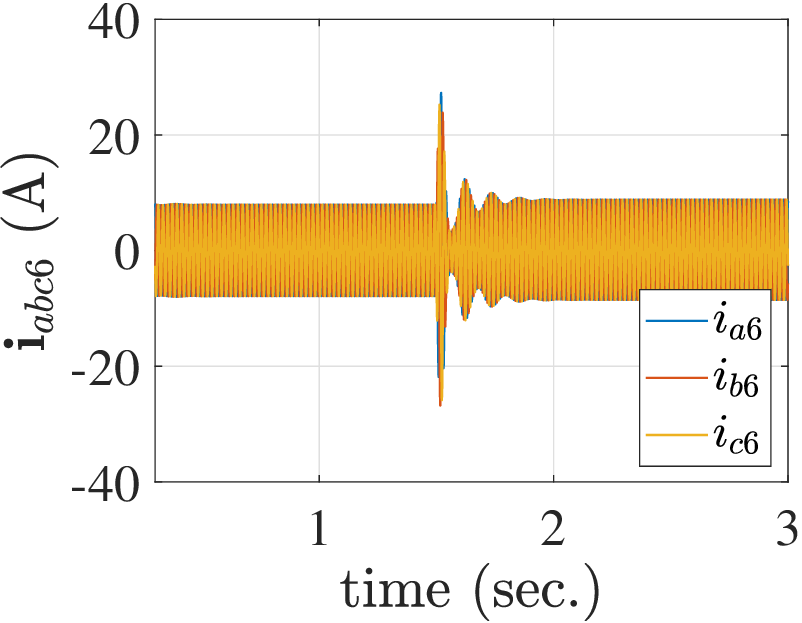}}
				\hfil
				\subfloat[]{\includegraphics[width=0.25\linewidth]{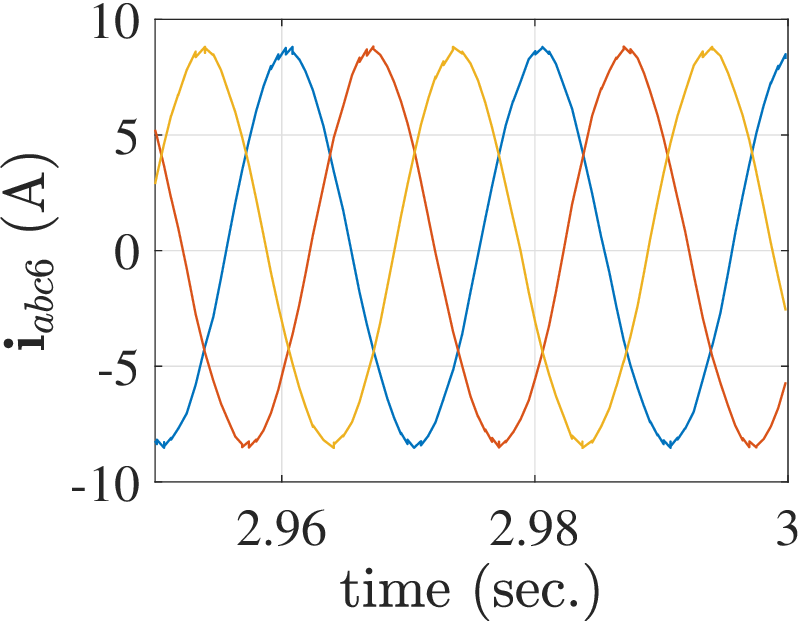}}
				\hfil
                \subfloat[]{\includegraphics[width=0.25\linewidth]{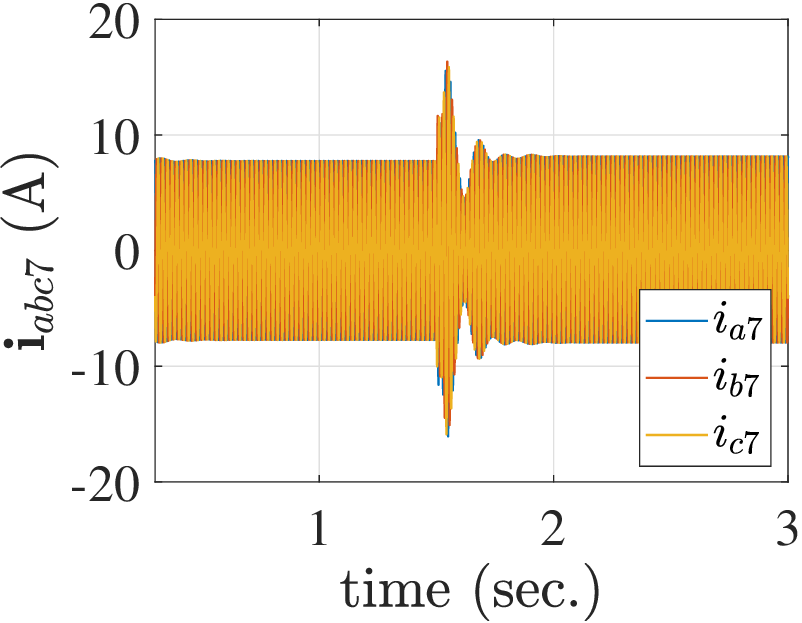}}
				\hfil
				\subfloat[]{\includegraphics[width=0.25\linewidth]{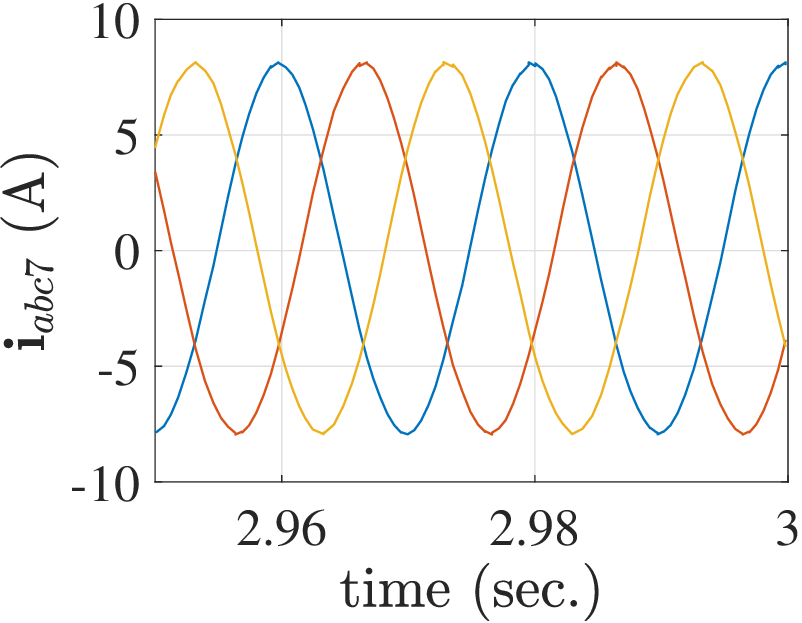}}
				\hfil
                \subfloat[]{\includegraphics[width=0.25\linewidth]{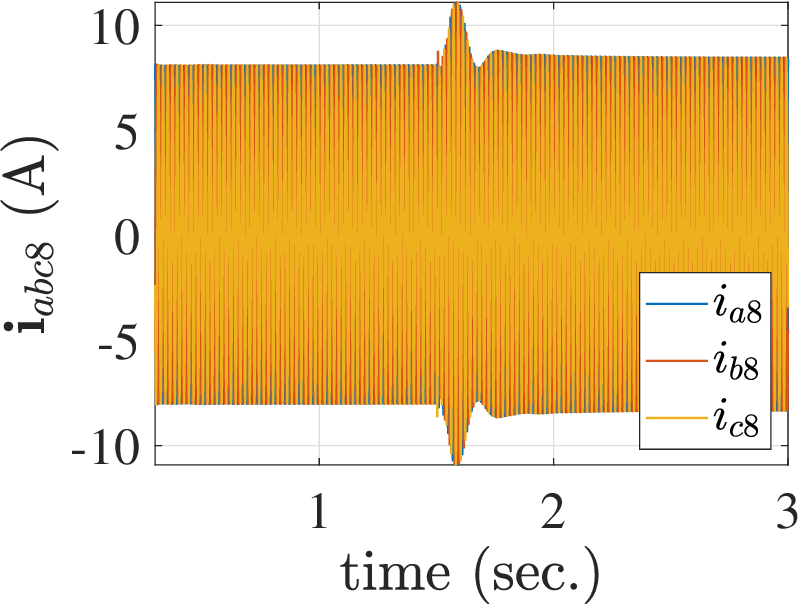}}
				\hfil
				\subfloat[]{\includegraphics[width=0.25\linewidth]{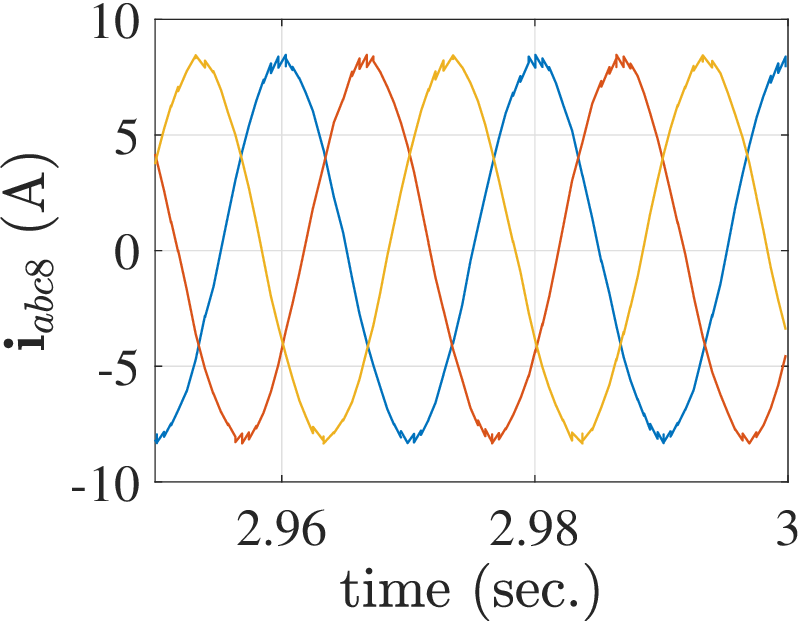}}
				\hfil
                \subfloat[]{\includegraphics[width=0.25\linewidth]{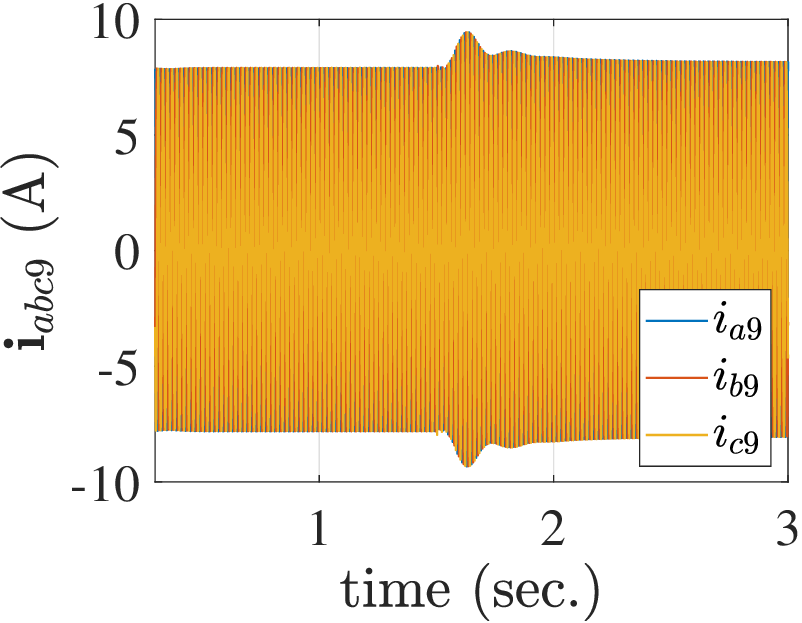}}
				\hfil
				\subfloat[]{\includegraphics[width=0.25\linewidth]{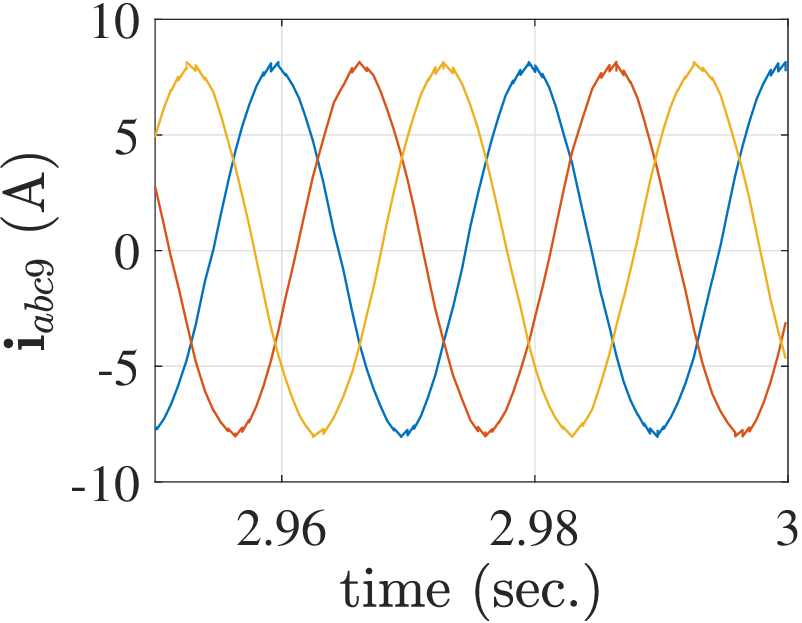}}
				\hfil
                \subfloat[]{\includegraphics[width=0.25\linewidth]{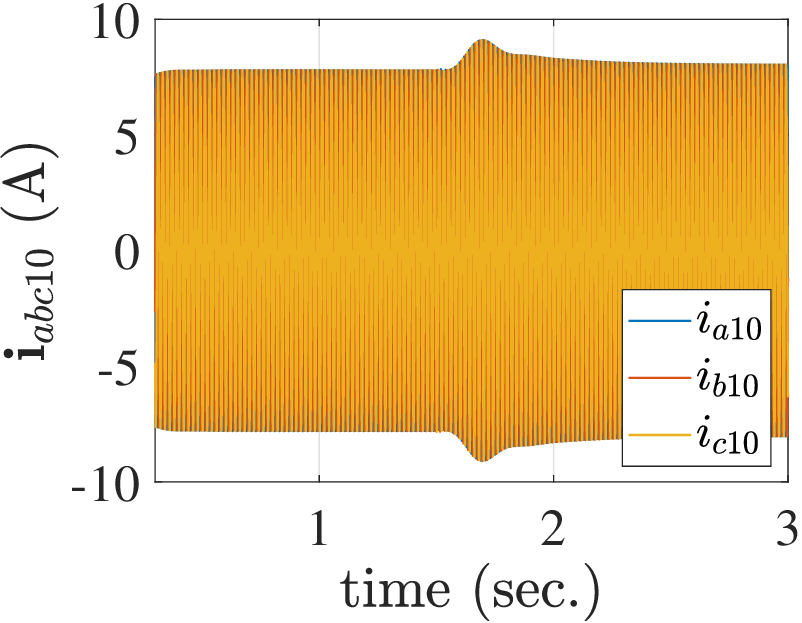}}
				\hfil
				\subfloat[]{\includegraphics[width=0.25\linewidth]{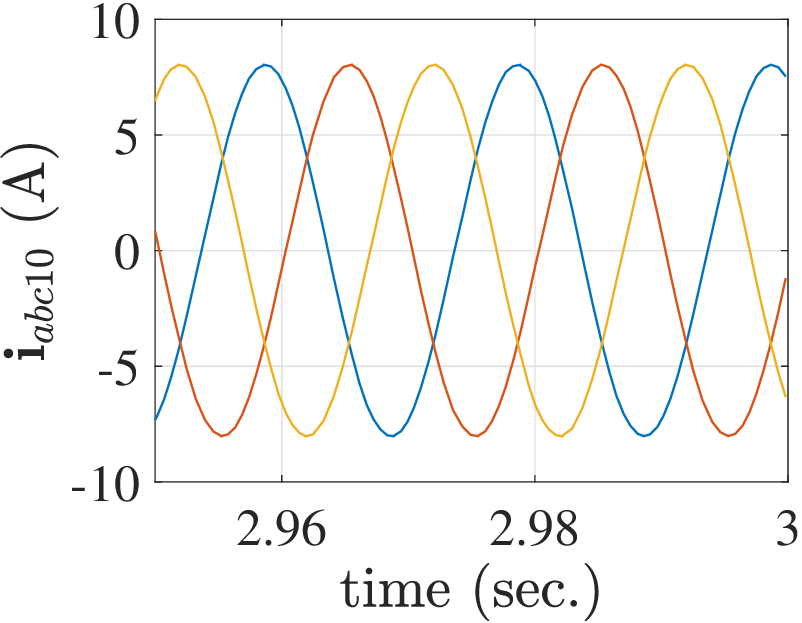}}
				\hfil
				\caption{\textcolor{black}{Terminal currents of the 10-IBR system (a,c,e,g,i,k,m,o,q,s) and their zoomed-in version (b,d,f,h,j,l,n,p,r,t) around 2.98 sec.}}
				\label{fig: TenIBR_iabc}
			\end{figure*}
\begin{figure*}
				\centering
				\subfloat[]{\includegraphics[width=0.25\linewidth]{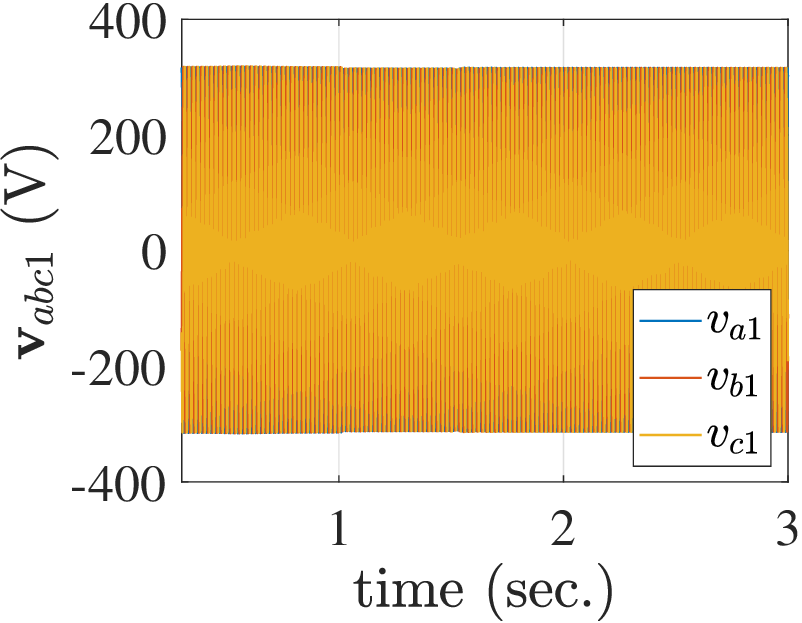}}
				\hfil
				\subfloat[]{\includegraphics[width=0.25\linewidth]{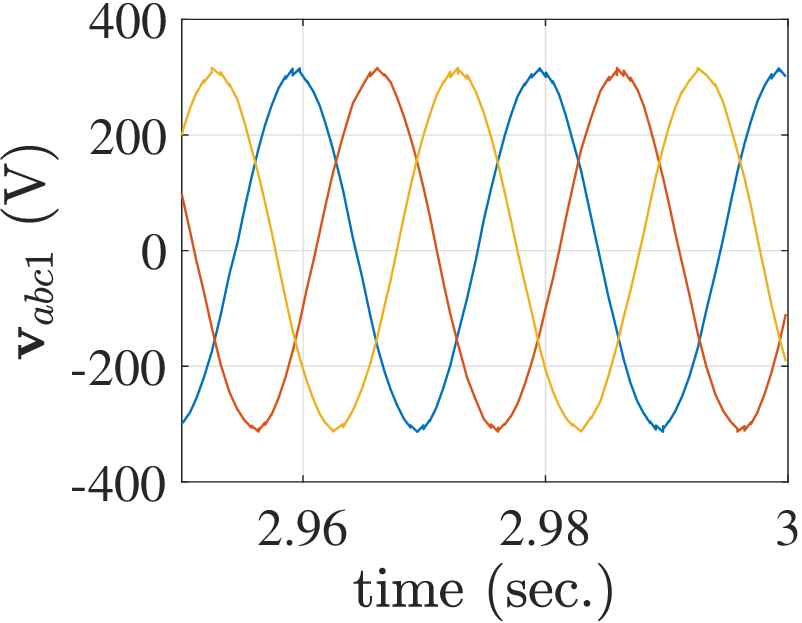}}
				\hfil
                \subfloat[]{\includegraphics[width=0.25\linewidth]{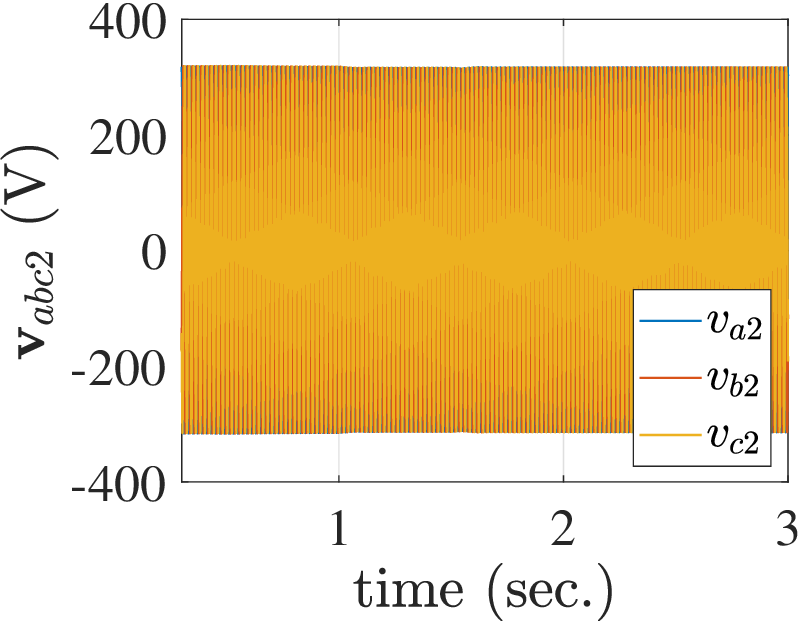}}
				\hfil
				\subfloat[]{\includegraphics[width=0.25\linewidth]{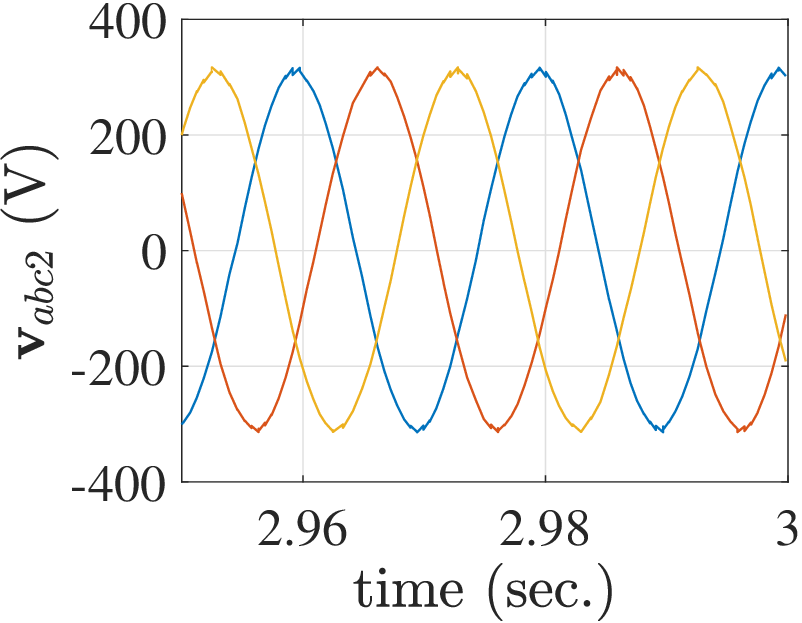}}
				\hfil
                \subfloat[]{\includegraphics[width=0.25\linewidth]{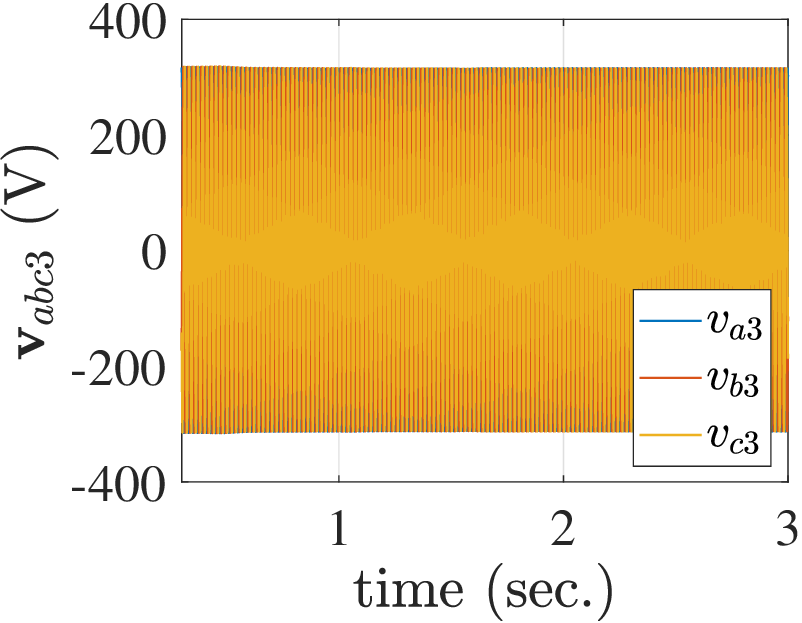}}
				\hfil
				\subfloat[]{\includegraphics[width=0.25\linewidth]{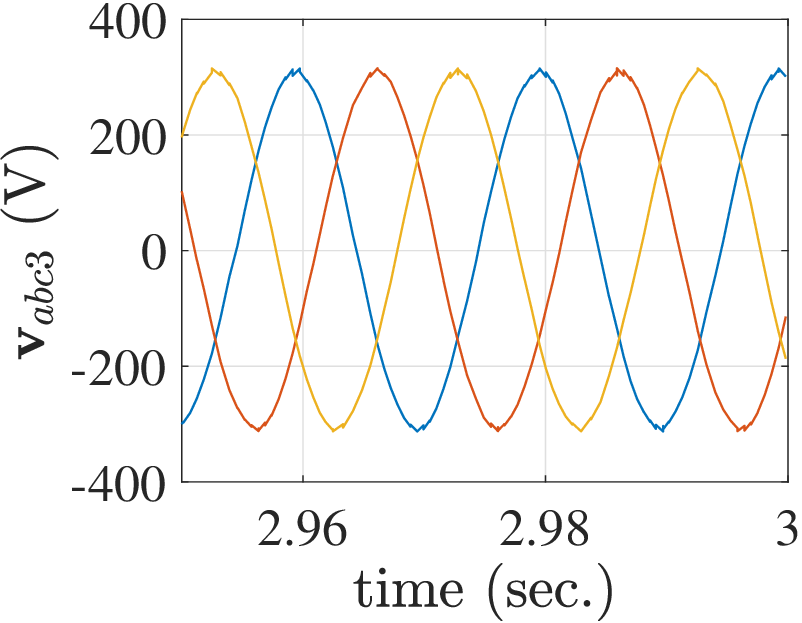}}
				\hfil
                \subfloat[]{\includegraphics[width=0.25\linewidth]{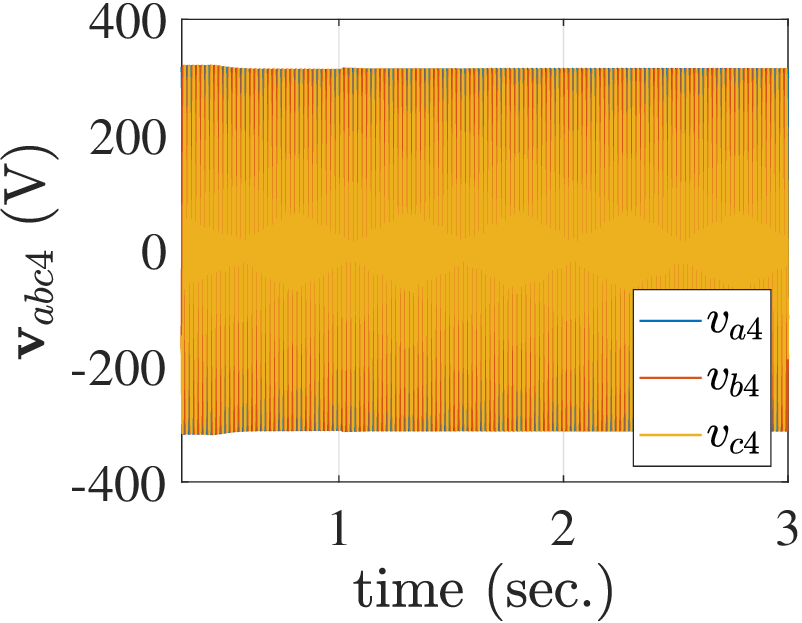}}
				\hfil
				\subfloat[]{\includegraphics[width=0.25\linewidth]{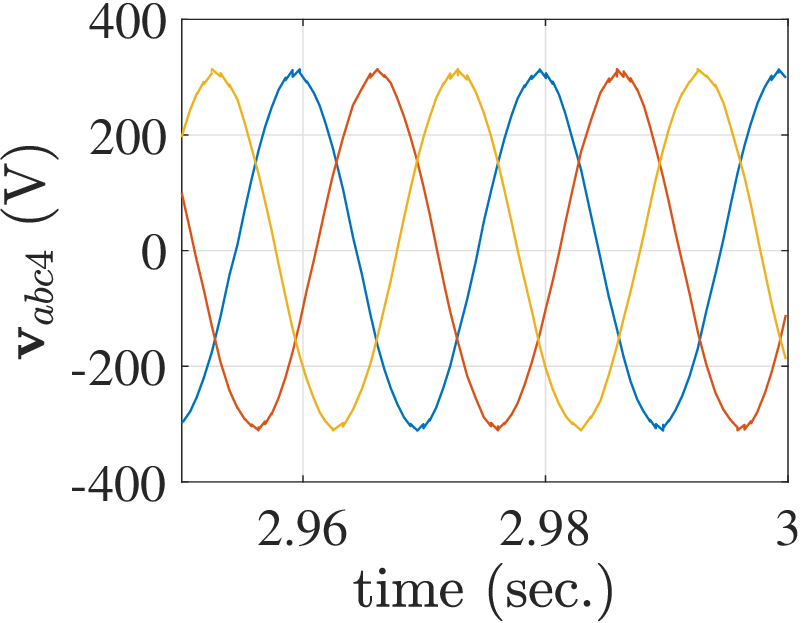}}
				\hfil
                \subfloat[]{\includegraphics[width=0.25\linewidth]{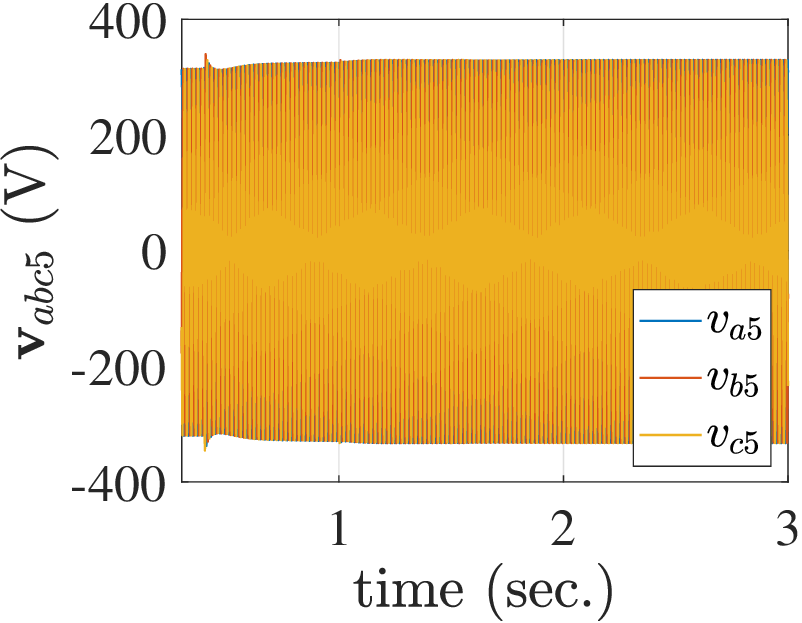}}
				\hfil
				\subfloat[]{\includegraphics[width=0.25\linewidth]{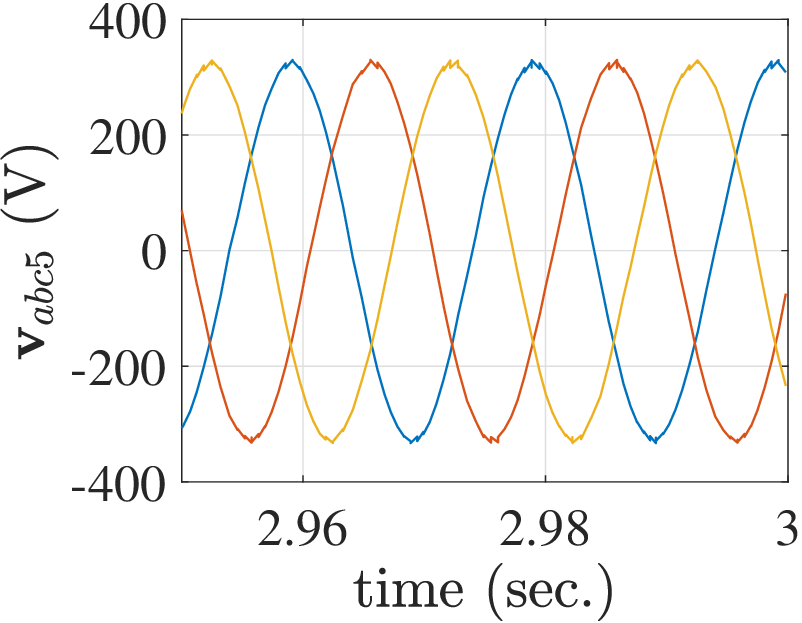}}
				\hfil
                \subfloat[]{\includegraphics[width=0.25\linewidth]{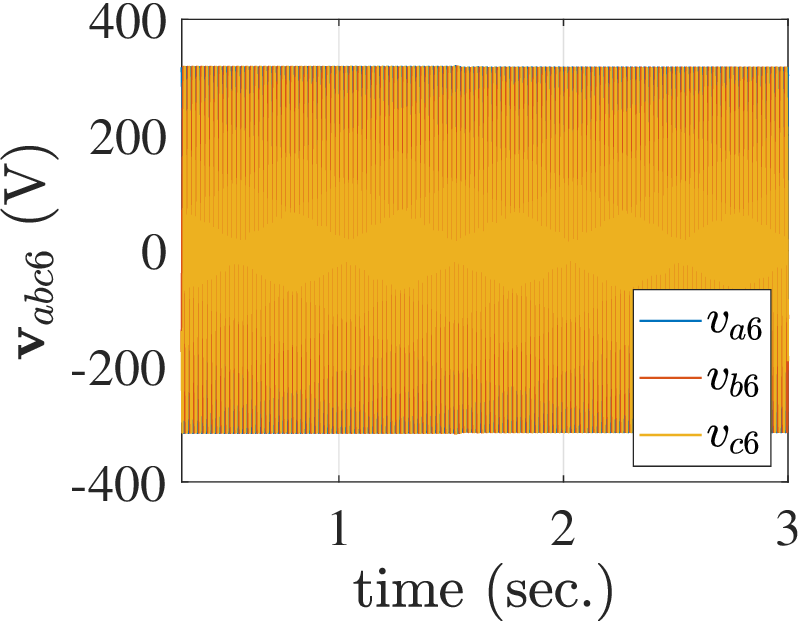}}
				\hfil
				\subfloat[]{\includegraphics[width=0.25\linewidth]{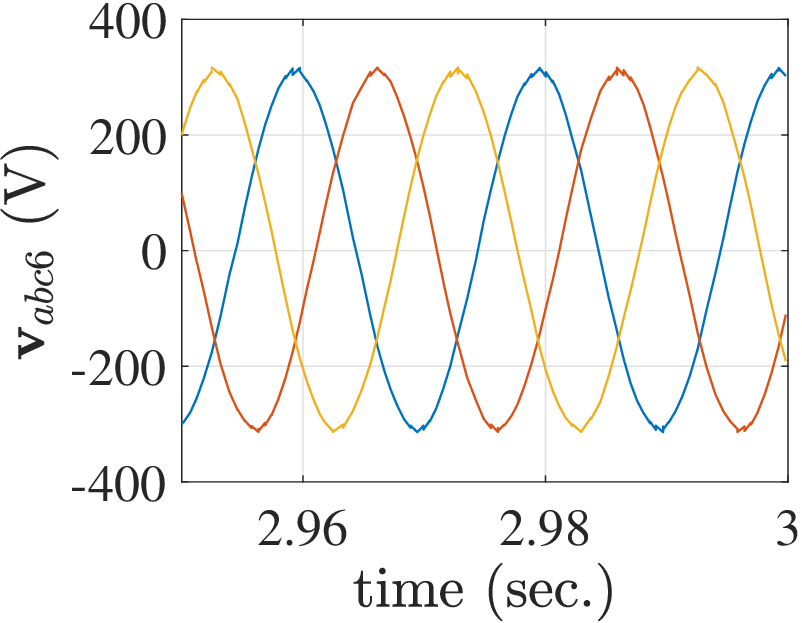}}
				\hfil
                \subfloat[]{\includegraphics[width=0.25\linewidth]{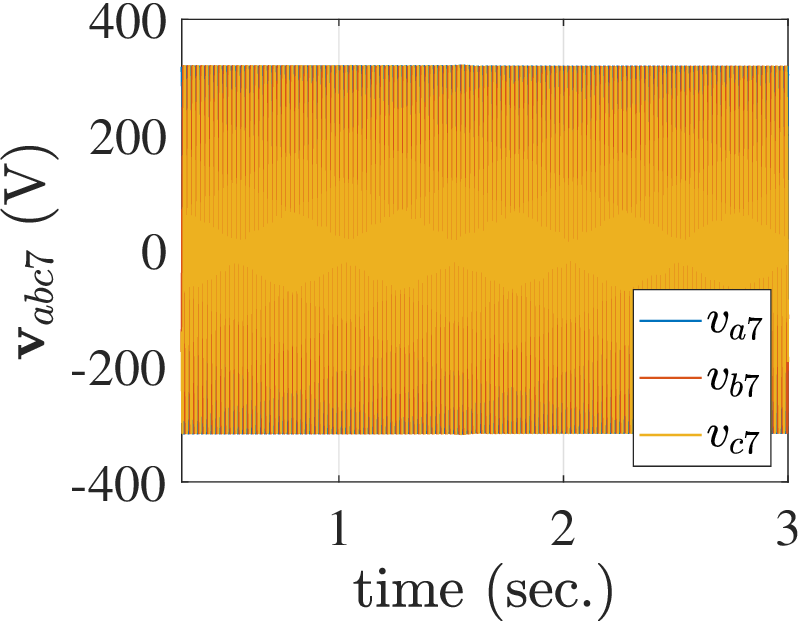}}
				\hfil
				\subfloat[]{\includegraphics[width=0.25\linewidth]{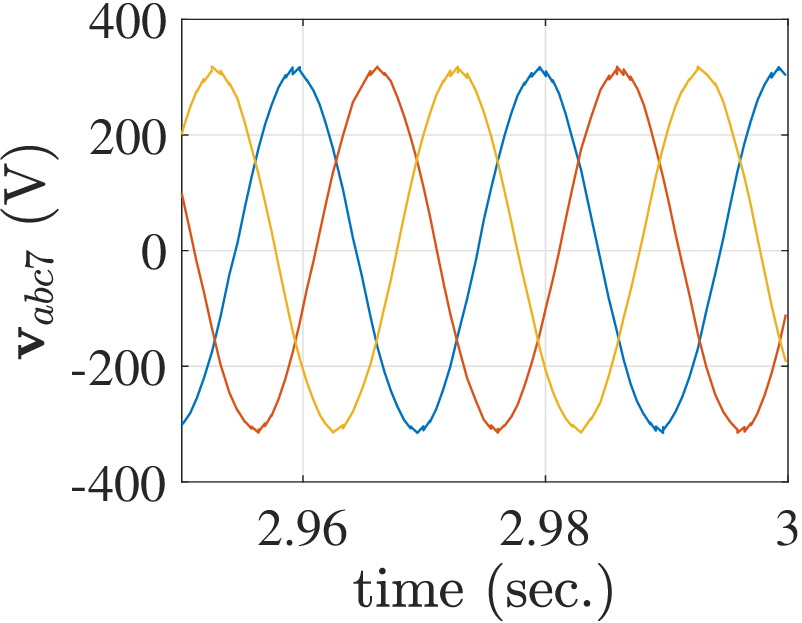}}
				\hfil
                \subfloat[]{\includegraphics[width=0.25\linewidth]{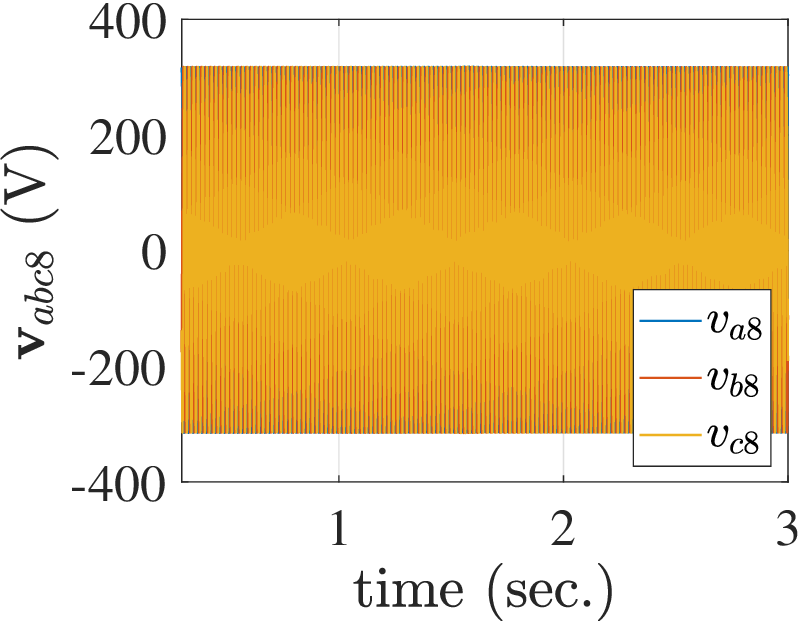}}
				\hfil
				\subfloat[]{\includegraphics[width=0.25\linewidth]{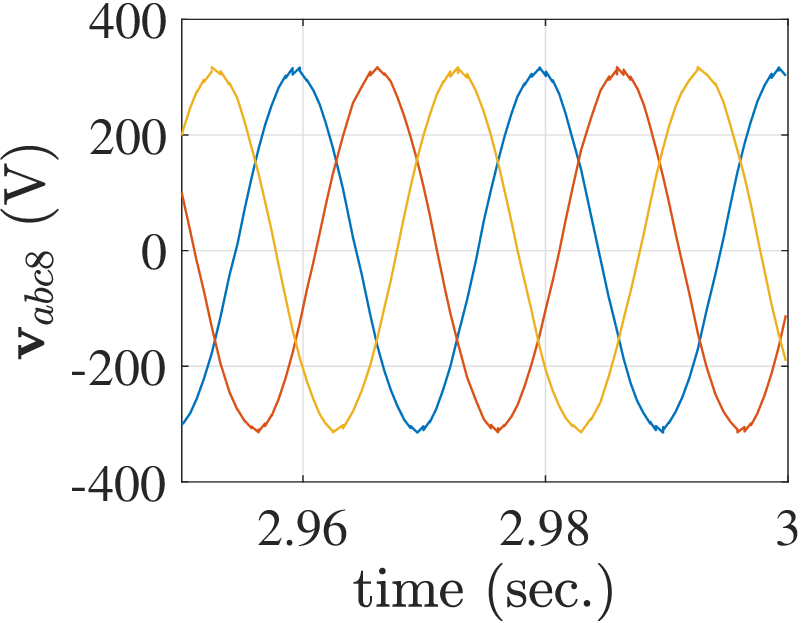}}
				\hfil
                \subfloat[]{\includegraphics[width=0.25\linewidth]{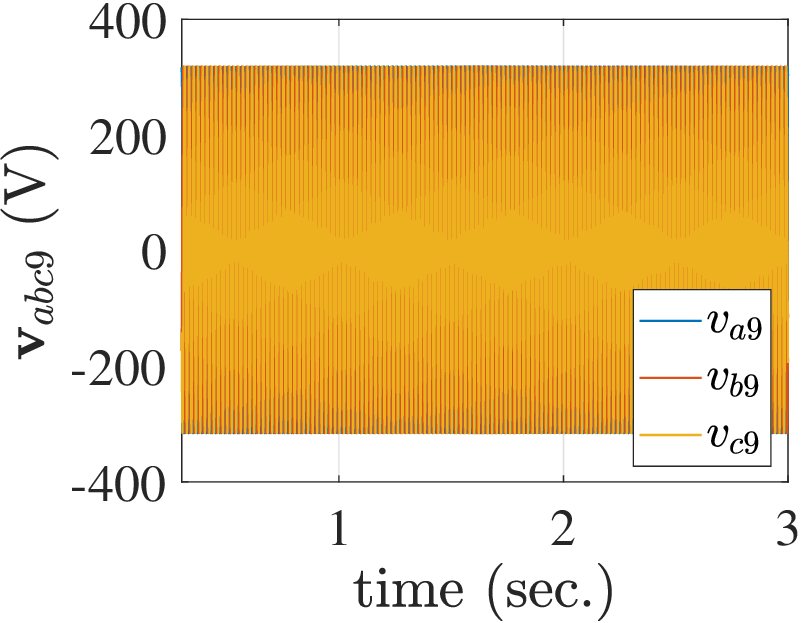}}
				\hfil
				\subfloat[]{\includegraphics[width=0.25\linewidth]{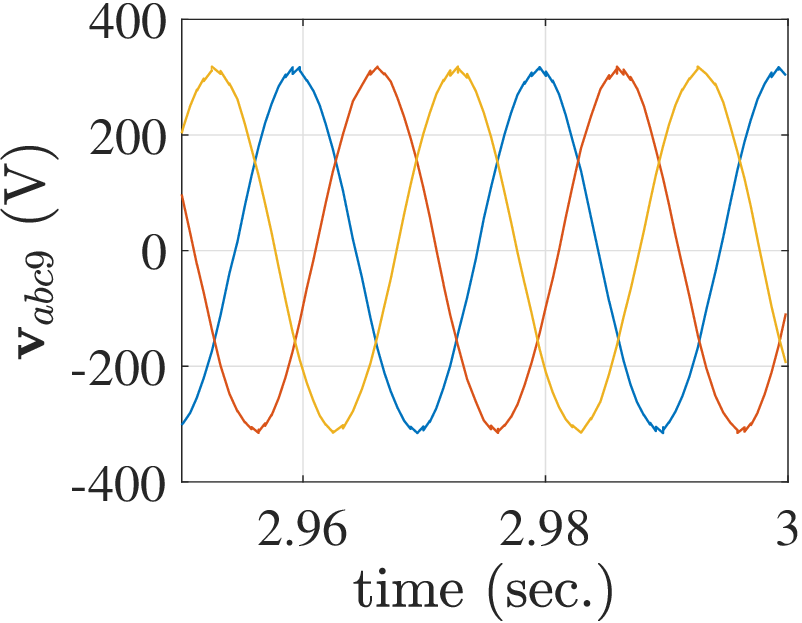}}
				\hfil
                \subfloat[]{\includegraphics[width=0.25\linewidth]{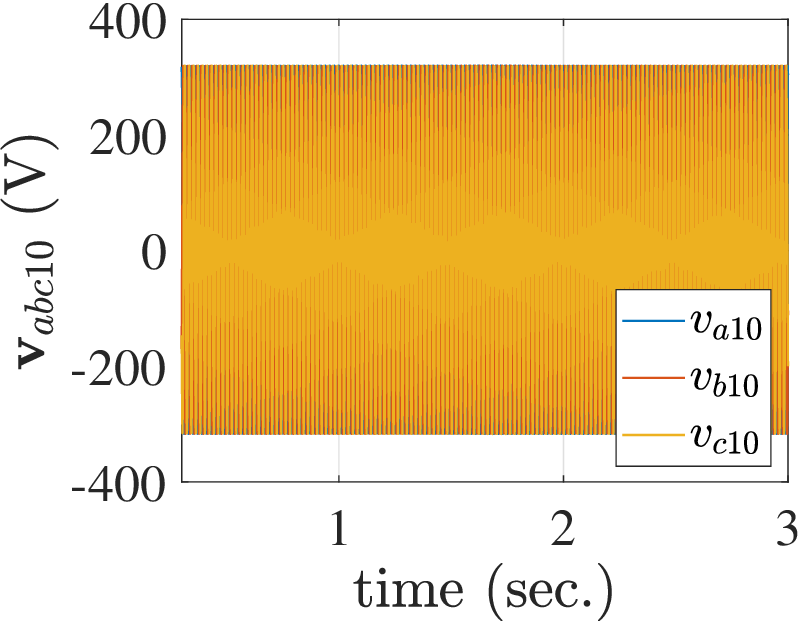}}
				\hfil
				\subfloat[]{\includegraphics[width=0.25\linewidth]{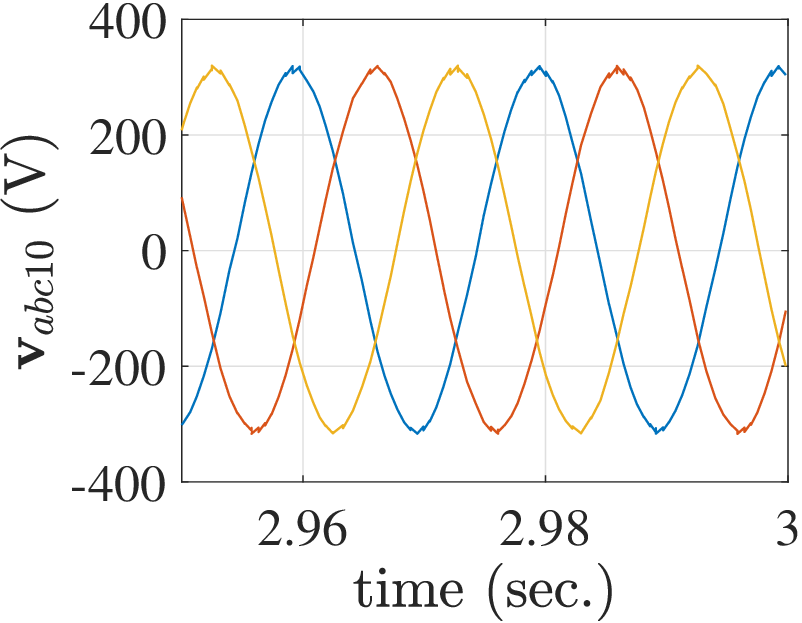}}
				\hfil
				\caption{\textcolor{black}{Terminal voltages of the 10-IBR system (a,c,e,g,i,k,m,o,q,s) and their zoomed-in version (b,d,f,h,j,l,n,p,r,t) around 2.98 sec.}}
				\label{fig: TenIBR_vabc}
			\end{figure*}

\begin{figure}
				\centering
				\subfloat[]{\includegraphics[width=0.5\linewidth]{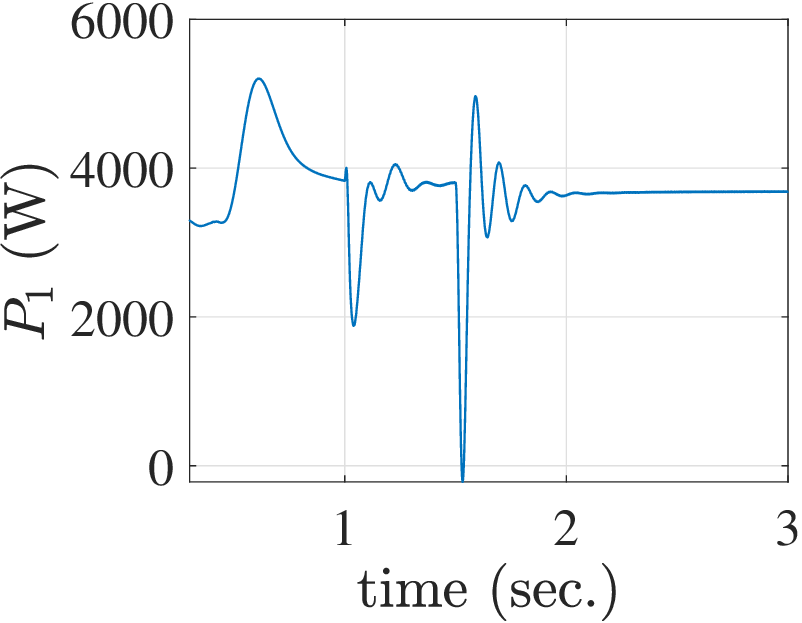}}
				\hfil
				\subfloat[]{\includegraphics[width=0.5\linewidth]{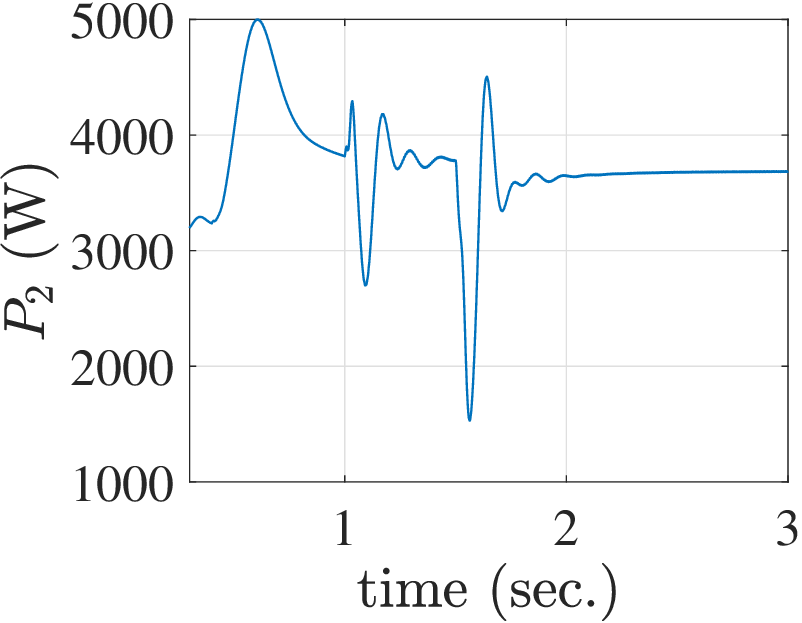}}
				\hfil
                \subfloat[]{\includegraphics[width=0.5\linewidth]{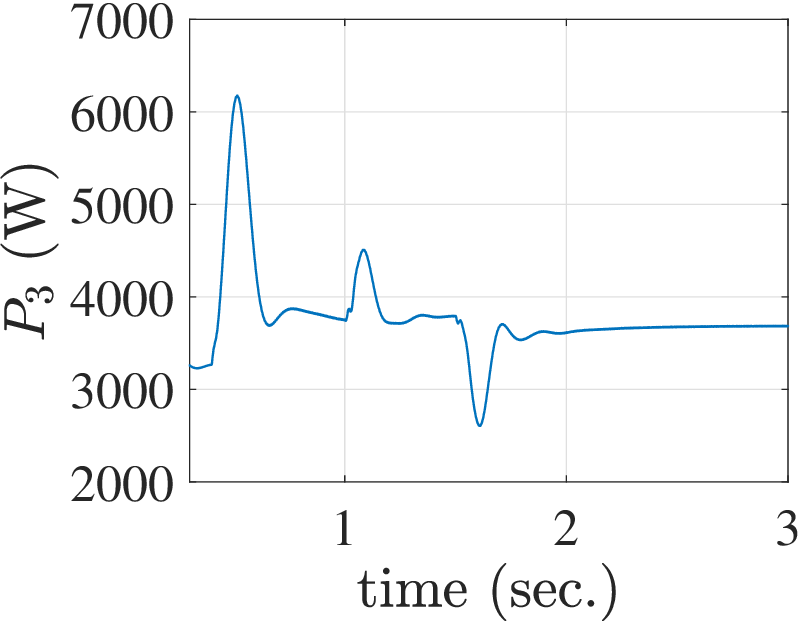}}
				\hfil
                \subfloat[]{\includegraphics[width=0.5\linewidth]{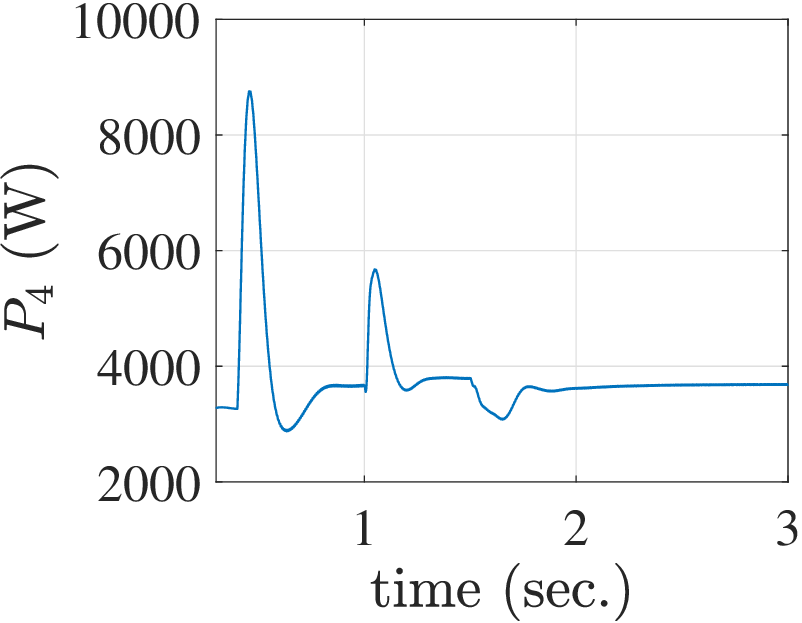}}
				\hfil
                \subfloat[]{\includegraphics[width=0.5\linewidth]{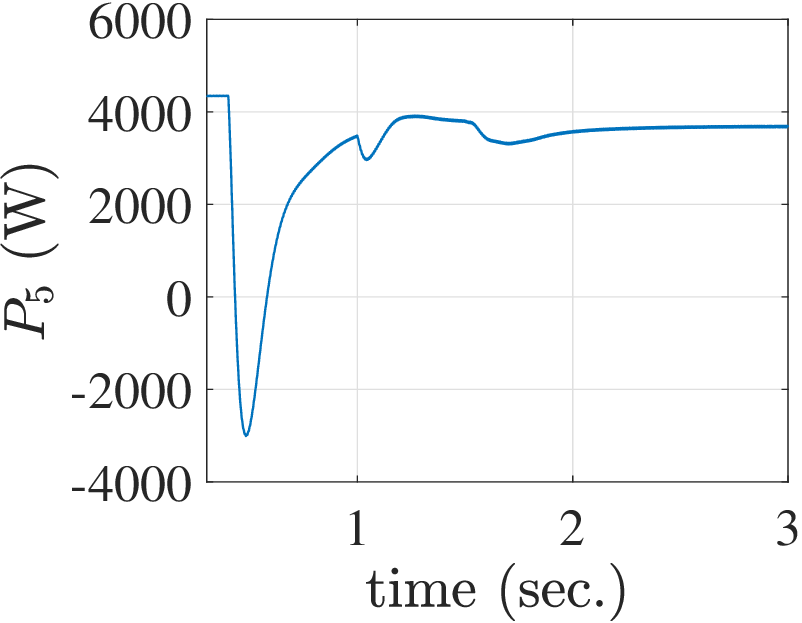}}
				\hfil
                \subfloat[]{\includegraphics[width=0.5\linewidth]{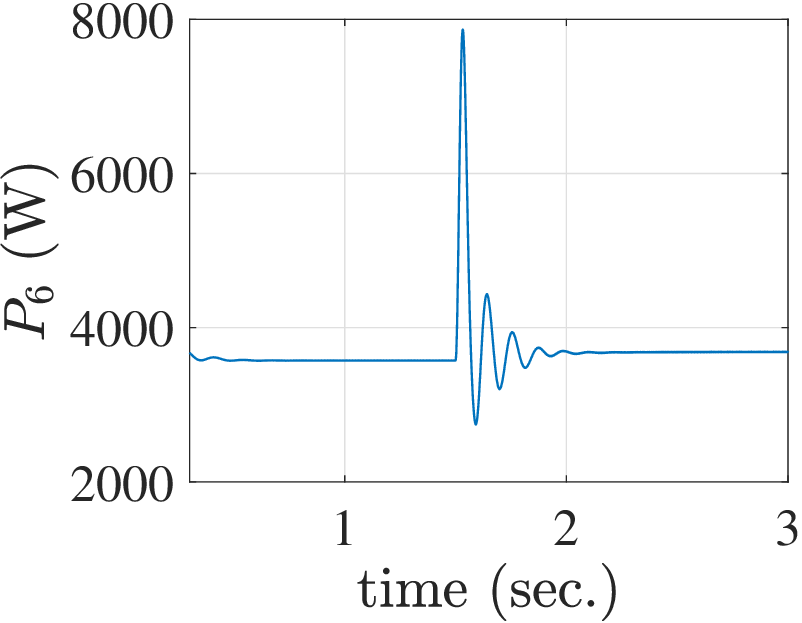}}
				\hfil
                \subfloat[]{\includegraphics[width=0.5\linewidth]{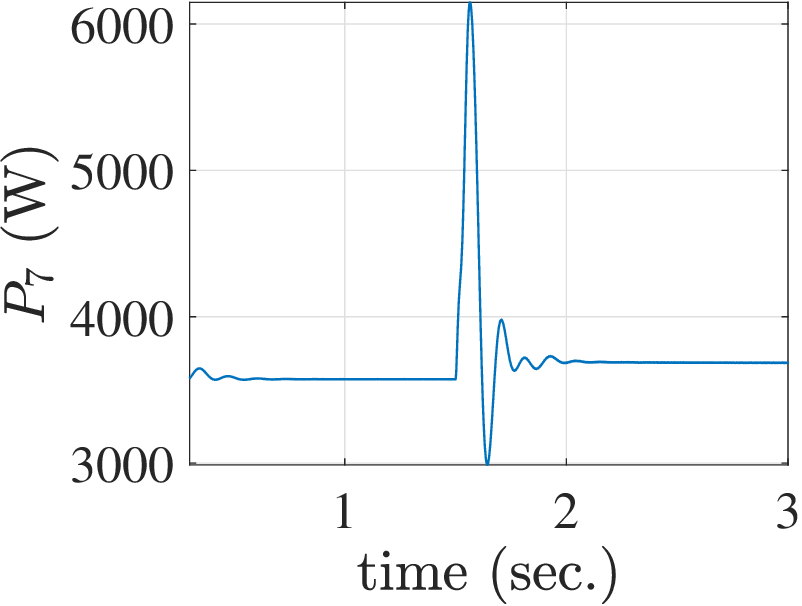}}
				\hfil
                \subfloat[]{\includegraphics[width=0.5\linewidth]{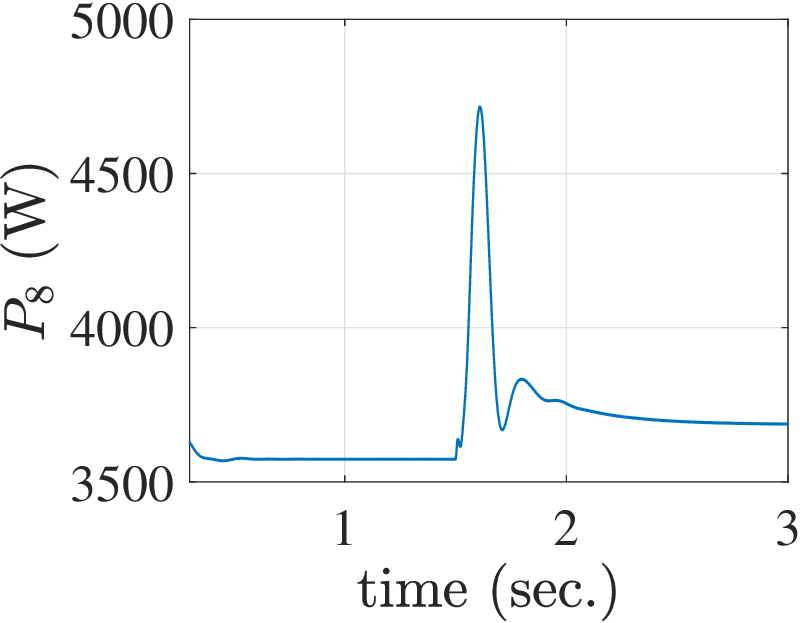}}
				\hfil
                \subfloat[]{\includegraphics[width=0.5\linewidth]{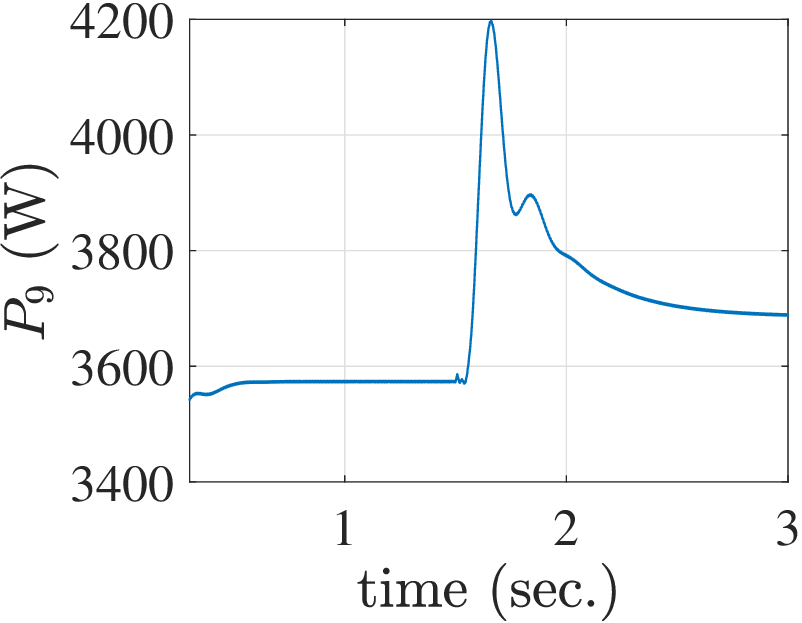}}
				\hfil
                \subfloat[]{\includegraphics[width=0.5\linewidth]{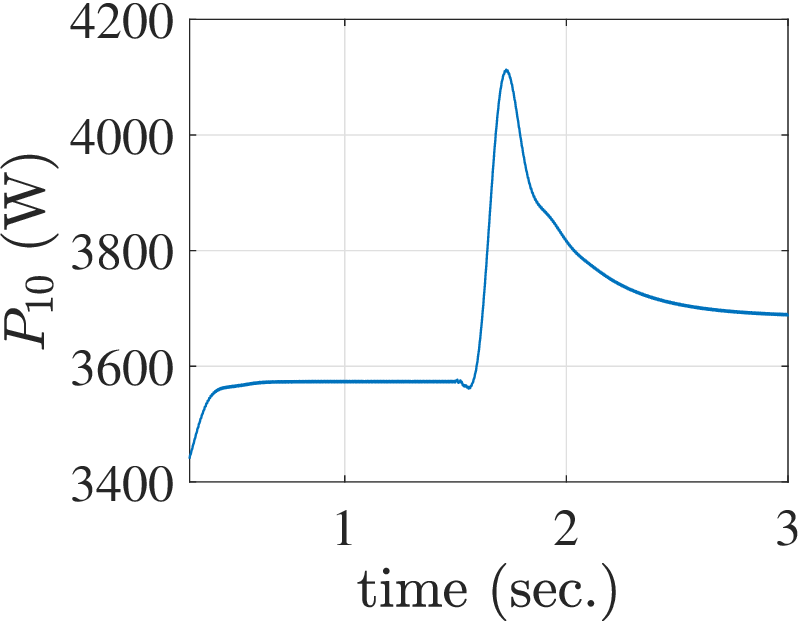}}
				\hfil
				\caption{\textcolor{black}{Instantaneous power of each IBR in the 10-IBR system.}}
				\label{fig: TenIBR_P}
			\end{figure}

\bio{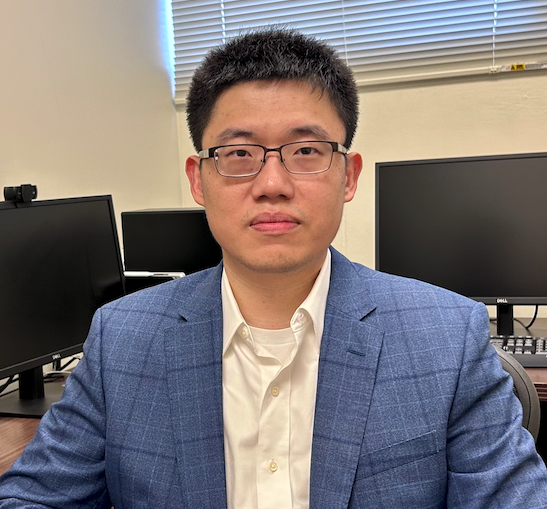}
\textbf{Tong Huang} received the Ph.D. degree from Texas A\&M University, College Station, TX, USA, in 2021. He is currently an Assistant Professor with the Department of Electrical and Computer Engineering, San Diego State University (SDSU), San Diego, CA, USA. Before joining SDSU, he was a Postdoctoral Associate with the Massachusetts Institute of Technology (MIT), Cambridge, MA, USA. In 2019, he had industry experience with ISO-New England, and in 2018 with Mitsubishi Electric Research Laboratories. As the lead PI, he was the recipient of the U.S. NSF ASCENT Award. As the first author, he was also the recipient of IEEE PES Technical Committee Prize Paper Award and two Best Paper Awards from 2020 IEEE PES General Meeting and 54th Hawaii International Conference on System Sciences.
\endbio

\end{document}